%% file: paper.tex
\documentclass{article}

\PassOptionsToPackage{square,numbers,sort&compress}{natbib}

\usepackage{ifthen}

\newboolean{arxv}
\setboolean{arxv}{true}
\newcommand{\ifarxiv}[2]{\ifthenelse{\boolean{arxv}}{#1}{#2}}

\ifarxiv{
	\usepackage[preprint]{neurips_2023}
}{
	\usepackage{neurips_2023}
}

\usepackage[utf8]{inputenc}      
\usepackage[T1]{fontenc}         
\usepackage{microtype}

\usepackage{amsmath}
\usepackage{amsfonts}
\usepackage{amssymb}
\usepackage{amsthm}
\usepackage{mathtools}
\usepackage{textcomp}

\usepackage{multibib}
\newcites{sm}{Additional References}
\bibliographystylesm{unsrtnat}

\usepackage[dvipsnames]{xcolor}  
\usepackage[pdftex]{graphicx}
\usepackage[inline]{enumitem}
\usepackage{booktabs}

\usepackage{caption}
\usepackage{subcaption}
\usepackage{titlesec}
\titlespacing{\section}{0pt}{\parskip}{0pt}
\titlespacing{\subsection}{0pt}{\parskip}{0pt}

\usepackage{hyperref}
\hypersetup{
	colorlinks=true,
	urlcolor=RoyalBlue,
	citecolor=Green
}

\graphicspath{{figures/}}

\input{pkg-todonotes.tex}
\renewcommand{\todo}[1]{}
\input{widebar.tex}

\newcommand{\T}{\mathsf{T}}
\newcommand{\given}{\,\vert\,}
\newcommand{\suff}{\succcurlyeq}

\newcommand{\thinsp}{\:\!}

\newcommand\indept{\protect\mathpalette{\protect\independenT}{\perp}}
\def\independenT#1#2{\mathrel{\rlap{$#1#2$}\mkern2mu{#1#2}}}

\newtheorem{theorem}{Theorem}
\newtheorem{proposition}[theorem]{Proposition}
\newtheorem{corollary}[theorem]{Corollary}
\newtheorem{definition}{Definition}
\newtheorem{property}{Property}
\theoremstyle{definition}
\newtheorem{remark}{Remark}
\newtheorem{example}{Example}
\newtheorem*{example*}{Example}

\makeatletter
\def\thm@space@setup{\thm@preskip=\parskip
\thm@postskip=\parskip}
\makeatother

\allowdisplaybreaks

\title{Gaussian Partial Information Decomposition: Bias Correction and Application to High-dimensional Data}

\author{%
	Praveen~Venkatesh\textsuperscript{1,2,*}, Corbett~Bennett\textsuperscript{1}, Sam~Gale\textsuperscript{1}, Tamina~K.~Ramirez\textsuperscript{1}, \\
	\textbf{Greggory~Heller\textsuperscript{1}, S\'everine~Durand\textsuperscript{1}, Shawn~Olsen\textsuperscript{1}, Stefan~Mihalas\textsuperscript{1,$\dagger$}} \\
	\textsuperscript{1}Allen Institute; \textsuperscript{2}University of Washington, Seattle, WA, USA \\
	\textsuperscript{*}\href{mailto:praveen.venkatesh@alleninstitute.org}{\texttt{praveen.venkatesh@alleninstitute.org}}; \textsuperscript{$\dagger$}\href{mailto:stefanm@alleninstitute.org}{\texttt{stefanm@alleninstitute.org}}
}

\begin{document}

\maketitle
\vspace{-2mm}

\begin{abstract}
	\vspace{-1mm}
	Recent advances in neuroscientific experimental techniques have enabled us to simultaneously record the activity of thousands of neurons across multiple brain regions.
	This has led to a growing need for computational tools capable of analyzing how task-relevant information is represented and communicated between several brain regions.
	Partial information decompositions (PIDs) have emerged as one such tool, quantifying how much unique, redundant and synergistic information two or more brain regions carry about a task-relevant message.
	However, computing PIDs is computationally challenging in practice, and statistical issues such as the bias and variance of estimates remain largely unexplored.
	In this paper, we propose a new method for efficiently computing and estimating a PID definition on multivariate Gaussian distributions.
	We show empirically that our method satisfies an intuitive additivity property, and recovers the ground truth in a battery of canonical examples, even at high dimensionality.
	We also propose and evaluate, for the first time, a method to correct the bias in PID estimates at finite sample sizes.
	Finally, we demonstrate that our Gaussian PID effectively characterizes inter-areal interactions in the mouse brain, revealing higher redundancy between visual areas when a stimulus is behaviorally relevant.
\end{abstract}

\setlength{\abovedisplayskip}{3pt}
\setlength{\belowdisplayskip}{3pt}


\vspace{-4mm}
\section{Introduction}

Neuroscientific experiments are increasingly collecting large-scale datasets with simultaneous recordings of multiple brain regions with single-unit resolution~\cite{devries2020large,siegle2021survey,stringer2019spontaneous}.
These experimental advances call for new computational tools that can allow us to probe how multiple brain regions jointly process relevant information in a behaving animal.

Partial Information Decompositions (PIDs) offer a new method for studying how different brain regions carry task-relevant information:
they provide measures to quantify the amount of \emph{unique}, \emph{redundant} and \emph{synergistic} information that one region has with respect to another.
The information itself could pertain to task-relevant variables such as stimuli, behavioral responses, or information contained in a third region.
For example, we may be interested in how much information about a stimulus is communicated or shared (i.e., redundantly present) between two brain regions over time.
Or, we might be interested in the extent to which one region's activity uniquely explains that of another, while excluding information corresponding to spontaneous behaviors.

Ideas such as redundancy and synergy have a long history in neuroscience, having been proposed for understanding noise correlations~\cite{schneidman2003synergy} and to understand differences in encoding complexity between different brain regions~\cite{gat1998synergy}.
PIDs have also been suggested for quantifying how much sensory information is used to execute behaviors~\cite{pica2017quantifying} and for tracking stimulus-dependent information flows between brain regions~\cite{pica2019using, bim2019non}.
Outside of neuroscience, PID has been used to understand interactions between different variables in financial markets~\cite{scagliarini2020synergistic}, to quantify the relevance of different features for the purpose of feature selection in machine learning~\cite{wollstadt2021rigorous}, and to define and quantify bias in the field of fair Machine Learning~\cite{dutta2020information}.

An important constraint that has limited the broader adoption of PIDs in neuroscience is the computational difficulty of estimating PIDs for high-dimensional data.
Many PID definitions that are operationally well-motivated involve solving an optimization problem over a space of probability distributions: the number of optimization variables can thus be exponential in the number of neurons~\cite{venkatesh2022partial}.
This has led to the use of poorly motivated PID definitions that are easy to compute (such as the ``MMI-PID'' of \cite{barrett2015exploration}, in works such as~\cite{scagliarini2020synergistic, colenbier2020disambiguating, boonstra2019information, krohova2019multiscale}), or limited analyses to very few dimensions~\cite{pakman2021estimating}.\todo{See if we can add more references}
Furthermore, due to the limited exploration of estimators for PIDs, issues such as the bias and variance of estimates have received no attention so far, to our knowledge.

In this paper, we make the following contributions:
\begin{enumerate}[leftmargin=*, topsep=-1pt, itemsep=-2pt]
	\item We provide a new and efficient method for computing and estimating a well-known PID definition called the $\sim$-PID or the BROJA-PID~\cite{bertschinger2014quantifying} on Gaussian distributions (Section~\ref{sec_computing_gpid}).
		By restricting our attention to Gaussian distributions, we are able to significantly reduce the number of optimization variables, so that this is just quadratic in the number of neurons, rather than exponential.
	\item We present a set of canonical examples for Gaussian distributions where ground truth is known, and show that our method outperforms others (Section~\ref{sec_examples}).
	\item We also raise (for what we believe is the first time) the issue of bias in PID estimates, propose a method for correcting the bias, and empirically evaluate its performance (Section~\ref{sec_estimation}).
	\item Finally, we show that our Gaussian PID estimator closely agrees with ground truth, even on non-Gaussian distributions, and show an example of its use on real neural data (Section~\ref{sec_neuroscience}).
\end{enumerate}

\textbf{Related work. }
Our method is \ifarxiv{based on our earlier work~\cite{venkatesh2022partial}, where we}{most closely related to that of~\citet{venkatesh2022partial}, who} also examined PIDs for Gaussian distributions.
Our \ifarxiv{current}{} work differs in a few key aspects: \begin{enumerate*}[label=(\roman*)]
	\item we estimate the PID of a different PID definition, the $\sim$-PID rather than the $\delta$-PID, because the $\delta$-PID does not satisfy a basic property called additivity~\cite{rauh2022continuity};
	\item our \ifarxiv{current}{} method provides an exact upper bound to the PID definition being computed, rather than an approximate upper bound; and
	\item we \ifarxiv{now}{} consider the problem of estimation, not just computation, and explore the issue of the bias of PID estimates.
\end{enumerate*}
Several other works have also considered methods for efficient estimation of PIDs:
\citet{banerjee2018computing} consider computing discrete PIDs efficiently, but their method does not scale to higher dimensions;
\citet{pakman2021estimating} estimate PIDs for continuous variables using copulas, but their method would also potentially be computationally prohibitive at high dimensionalities;
\citet{liang2023quantifying} use convex optimization to directly estimate the $\sim$-PID for general high-dimensional distributions, but they do not compare with ground truth at high dimensionality or examine bias in their estimates.


\section{Background: An Introduction to PIDs and the $\sim$-PID}
\label{sec_background}

In this section, we provide an introduction to the concept of partial information decomposition along with an illustrative example.
Let $M$, $X$ and $Y$ be three random variables with joint distribution $P_{MXY}$.
A PID decomposes the total mutual information between the \emph{message} $M$ and two \emph{constituent} random variables $X$ and $Y$ into a sum of four non-negative components that satisfy~\cite{williams2010nonnegative, bertschinger2014quantifying}:
\begin{align}
	I\bigl(M ; (X, Y)\bigr) &= UI(M : X \setminus Y) + UI(M : Y \setminus X) + RI(M : X ; Y) + SI(M : X ; Y) \label{eq_pid} \\
	I(M ; X) &= UI(M : X \setminus Y) + RI(M : X ; Y) \label{eq_pidx} \\
	I(M ; Y) &= UI(M : Y \setminus X) + RI(M : X ; Y) \label{eq_pidy}
\end{align}
Here, $I(A; B)$ is the \emph{Shannon mutual information}~\cite{cover2012elements} between the random variables $A$ and $B$,
and the four terms in the RHS of \eqref{eq_pid} are respectively the information about $M$ that is \begin{enumerate*}[label=(\roman*)]
	\item \emph{uniquely} present in $X$ and not in $Y$;
	\item \emph{uniquely} present in $Y$ and not in $X$;
	\item \emph{redundantly} present in both $X$ and $Y$ and can be extracted from either; and
	\item \emph{synergistically} present in $X$ and in $Y$, i.e., information which cannot be extracted from either of them individually, but can be extracted from their interaction.
\end{enumerate*}
For the sake of brevity, we may also refer to these partial information components as $UI_X$, $UI_Y$, $RI$ and $SI$ respectively.
Notwithstanding notation, they should all properly be understood to be functions of the joint distribution $P_{MXY}$.

Now, $UI_X$, $UI_Y$, $RI$ and $SI$ consist of four undefined quantities, subject to the three equations in~\eqref{eq_pid}--\eqref{eq_pidy}.
In addition, they are typically assumed to be non-negative, $RI$ and $SI$ are each constrained to be symmetric in $X$ and $Y$, and the functional forms of $UI_X$ and $UI_Y$ should be identical when exchanging $X$ for $Y$.
Despite the number of constraints, many definitions satisfy all of them, each differing in its motivation and interpretation~\cite{williams2010nonnegative,bertschinger2014quantifying,griffith2014quantifying,harder2013bivariate,kolchinsky2019novel} (see~\cite{lizier2018information,kolchinsky2019novel} for a review), and we need to formally define one of these partial information components to determine the other three.

\begin{example} \label{ex_binary}
	Before we jump into a specific definition, we provide an intuition into what these terms mean using a simple example.
	Suppose $M = [A, B, C]$, $X = [A, B, C\!\oplus\!Z]$, and $Y = [B, Z]$, where $A, B, C, Z \sim$ i.i.d.~Ber(0.5).%
	\footnote{\,i.i.d.\ stands for ``independent and identically distributed''; $X \indept Y $ means $X$ and $Y$ are independent.}
	Then, $X$ has 1 bit of unique information about $M$, i.e., $A$; $Y$ has no unique information;
	$X$ and $Y$ both have 1 bit of redundant information, i.e., $B$, since it can be obtained from either $X$ or $Y$; and
	$X$ and $Y$ have 1 bit of synergistic information, i.e., $C$, which cannot be obtained from either $X$ or $Y$ individually (since $C\!\oplus\!Z \indept C$), but can only be recovered when both $X$ and $Y$ are known.
	For more examples on binary variables, we refer the reader to~\cite{bertschinger2014quantifying}.
\end{example}

\vspace{-6pt}
In this manuscript, we consider a definition that we refer to as the $\sim$-PID%
\footnote{This PID is also sometimes referred to as the BROJA PID (after the authors of \cite{bertschinger2014quantifying}), or the minimum-synergy PID in the literature.
We prefer to use an author-agnostic nomenclature as introduced \ifarxiv{in our earlier work}{by the authors of}~\cite{venkatesh2022partial}, because this PID was also introduced contemporaneously by \cite{griffith2014quantifying}.}%
~\cite{bertschinger2014quantifying, griffith2014quantifying}, which is defined below.
We chose to build an estimator for \emph{this} definition for two reasons: \begin{enumerate*}[label=(\roman*)]
	\item it is a \emph{Blackwellian PID} definition, i.e., it has a well-defined operational interpretation based on concepts from statistical decision theory (e.g., see~\cite{bertschinger2014quantifying, venkatesh2023capturing} for details); and
	\item it satisfies many desirable properties (e.g., see~\cite{bertschinger2014quantifying, banerjee2018unique}), and in particular, a property that we call \emph{additivity of independent components}.
\end{enumerate*}

\begin{definition}[$\sim$-PID \cite{bertschinger2014quantifying}] \label{def_tilde_pid}
	The unique information about $M$ present in $X$ and not in $Y$ is given by
	\begin{equation} \label{eq_tilde_pid}
		\widetilde{UI}(M : X \setminus Y) \coloneqq \min_{Q \in \Delta_P} I_Q(M ; X \given Y),
	\end{equation}
	where $\Delta_P \coloneqq \{Q_{MXY}: Q_{MX} = P_{MX},\; Q_{MY} = P_{MY}\}$ and $I_Q(\,\cdot\: ; \cdot \given \cdot)$ is the conditional mutual information over the joint distribution $Q_{MXY}$.
	The remaining $\sim$-PID components, $\widetilde{UI}(M : Y \setminus X)$, $\widetilde{RI}(M : X ; Y)$ and $\widetilde{SI}(M : X ; Y)$, follow from equations~\eqref{eq_pid}--\eqref{eq_pidy}.
\end{definition}

\begin{property}[Additivity of independent components] \label{ppty_additivity}
	Suppose $M = [M_1, M_2]$, $X = [X_1, X_2]$, and $Y = [Y_1, Y_2]$, such that $(M_1, X_1, Y_1) \indept (M_2, X_2, Y_2)$. Then, additivity implies that
	\begin{equation}
		UI(M : X \setminus Y) = UI(M_1 : X_1 \setminus Y_1) + UI(M_2 : X_2 \setminus Y_2),
	\end{equation}
	and similarly for the other three partial information components, $UI_Y$, $RI$ and $SI$.
\end{property}
\vspace{-6pt}
Property~\ref{ppty_additivity} ensures that we can compute the PIDs of independent systems separately and then add the components across systems, making it intuitive and highly desirable.
Of the many PID definitions examined by \citet{rauh2022continuity}, only the $\sim$-PID satisfied additivity (as proved in~\cite{bertschinger2014quantifying}).


\section{Computing the $\sim$-PID for Gaussian Distributions}
\label{sec_computing_gpid}

The first contribution of this paper is a method to efficiently compute bounds on the $\sim$-PID for jointly Gaussian random vectors $M$, $X$ and $Y$.
To be precise, our method computes an upper bound for $\widetilde{UI}_X$ and $\widetilde{UI}_Y$, and lower bounds for $\widetilde{RI}$ and $\widetilde{SI}$.
\ifarxiv{Similar to our earlier work}{Closely following}~\cite{venkatesh2022partial}, we present a new PID definition that we call the $\sim_G$-PID, which characterizes an upper bound on the unique information of the $\sim$-PID by restricting the optimization space to jointly Gaussian $Q_{MXY}$:
\begin{definition}[$\sim_G$-PID] \label{def_gauss_tilde_pid}
	The unique information about $M$ present in $X$ and not in $Y$ is given by
	\begin{equation}
		\widetilde{UI}_G(M : X \setminus Y) \coloneqq \min_{Q \in \Delta_P} I_Q(M ; X \given Y),
	\end{equation}
	where $\Delta_P \coloneqq \{Q_{MXY}: Q_{MXY} \text{ jointly Gaussian,} \; Q_{MX} = P_{MX},\; Q_{MY} = P_{MY}\}$ and $I_Q$ is the conditional mutual information over the joint distribution $Q_{MXY}$.
\end{definition}
\vspace{-6pt}
Definition~\ref{def_gauss_tilde_pid} is identical to Definition~\ref{def_tilde_pid}, except for the fact that $Q_{MXY}$ is constrained to be jointly Gaussian.
If the optimal $Q$ in the original $\sim$-PID of Definition~\ref{def_tilde_pid} is in fact Gaussian for some $P_{MXY}$, then the $\sim_G$-PID would be identical to the $\sim$-PID for that $P_{MXY}$.
We conjecture that this happens whenever $P_{MXY}$ is Gaussian: for example, in a similar optimization problem for computing the information bottleneck~\cite{tishby2000information}, the optimal distribution is Gaussian whenever $P$ is Gaussian~\cite{chechik2003information,globerson2004optimality}.
We leave this conjecture as an open question for future work.
For now, the unique information of the $\sim_G$-PID provides only an upper bound on the unique information of the $\sim$-PID, in general.

Nonetheless, restricting the search space to Gaussian $Q_{MXY}$ dramatically simplifies the optimization problem, allowing us to compute the $\sim_G$-PID for much higher dimensionalities of $M$, $X$ and $Y$.
In what follows, we show how the optimization problem for the $\sim_G$-PID can be written out in closed-form and then solved using projected gradient descent.

\subsection{Notation and Preliminaries}

Suppose $M$, $X$ and $Y$ are jointly Gaussian random vectors of dimensions $d_M$, $d_X$ and $d_Y$ respectively, with a joint covariance matrix given by $\Sigma_{MXY}$.
We will make extensive use of the submatrices of $\Sigma_{MXY}$, so we explain their notation here:
\begin{itemize}[leftmargin=*, topsep=-3pt, itemsep=-2pt]
	\item $\Sigma_{XY}$ will denote the $(d_X + d_Y) \times (d_X + d_Y)$ joint (auto-)covariance matrix of the vector $[X^\T, Y^\T]^\T$.
	\item $\Sigma_{X,Y}$ (note the comma) will denote the $d_X \times d_Y$ cross-covariance matrix between $X$ and $Y$.
	\item $\Sigma_{XY,M}$ will denote the $(d_X + d_Y) \times d_M$ cross-covariance matrix between the concatenated vector $[X^\T, Y^\T]^\T$ and the vector $M$.
\end{itemize}
In general, groupings of vectors without commas represent joint covariance, while a comma represents a cross-covariance between the groups on either side of the comma.
The same notation will also be used for conditional covariance matrices: for example, $\Sigma_{XY|M}$ is the conditional \emph{joint} covariance of $(X, Y)$ given $M$, while $\Sigma_{X,Y|M}$ is the conditional \emph{cross}-covariance \emph{between} $X$ and $Y$ given $M$.

We will also use an equivalent notation for the joint distribution~\cite{venkatesh2022partial}, where $P_{MXY}$ is parameterized as a channel from $M$ to $X$ and $Y$:
\begin{equation} \label{eq_mxy_channel}
	X = H_X M + N_X \quad\text{and}\quad  Y = H_Y M + N_Y.
\end{equation}
Here, $H_X \coloneqq \Sigma_{X,M}$ and $H_Y \coloneqq \Sigma_{Y,M}$ represent the channel gain matrices, while $N_X$ and $N_Y$ represent additive noise and are not necessarily independent of each other: $[N_X^\T, N_Y^\T]^\T \sim \mathcal N(0, \Sigma_{XY|M})$.

\begin{remark} \label{rem_whitening}
	Without loss of generality, we can assume that $M$, $X$ and $Y$ are all zero-mean, and that $\Sigma_{M} = I$.
	Further, we explicitly assume that the $X$ and $Y$ channels are individually whitened, i.e., that $\Sigma_{X|M} = I$ and $\Sigma_{Y|M} = I$.
	This assumption precludes deterministic relationships between $M$ and $X$ or $Y$, and is required to ensure that information quantities remain finite~\cite{venkatesh2022partial}.
\end{remark}

\subsection{Optimizing the Union Information}

\citet{bertschinger2014quantifying} showed that the minimizer for the unique information is also the minimizer for the ``union information'', $I^\cup(M : X; Y) \coloneqq UI_X + UI_Y + RI$.
In other words, we can also solve the following optimization problem, which yields simpler expressions for the objective and gradient:
\begin{equation} \label{eq_tilde_union_info}
	\widetilde{I^\cup}(M : X ; Y) \coloneqq \min_{Q_{MXY}} I_Q(M ; X, Y) \quad \text{s.t.} \quad Q_{MX} = P_{MX}, \; Q_{MY} = P_{MY}
\end{equation}
Now, suppose $P_{MXY}$ is Gaussian with covariance $\Sigma^P_{MXY}$ and the solution $Q_{MXY}$ is also assumed to be Gaussian with covariance $\Sigma_{MXY}^Q$.
Then, the constraint in \eqref{eq_tilde_union_info} implies that $\Sigma^Q_{MX} = \Sigma^P_{MX}$ and $\Sigma^Q_{MY} = \Sigma^P_{MY}$.
In other words, $\Sigma_{M}$, $\Sigma_X$, $\Sigma_Y$, and $\Sigma_{M,XY}$ are all constant across $P$ and $Q$.
Therefore, the only part of $\Sigma^Q_{MXY}$ that is variable is $\Sigma^Q_{X,Y}$, or equivalently, $\Sigma^Q_{X,Y|M}$.%
\footnote{We can use $\Sigma_{X,Y|M}$ in place of $\Sigma_{X,Y}$ because they differ by a constant: $\Sigma_{X,Y|M} - \Sigma_{X,Y}$ is an off-diagonal block in $\Sigma_{XY} - \Sigma_{XY|M}$, which is equal to $\Sigma_{XY,M}^{\vphantom{-\T}} \Sigma_M^{-1} \Sigma_{XY,M}^{\vphantom{-}\T}$, which is constant across $P$ and $Q$.}
In what follows, we will drop the superscripts denoting the distribution, as this will be clear from context.
Generally speaking, we will discuss the optimization problem and thus the distribution will be $Q$.

\begin{proposition} \label{prop_gpid_optim}
	The union information for the $\sim_G$-PID of Definition~\ref{def_gauss_tilde_pid} is given by
	\begin{equation} \label{eq_union_info_obj}
		\widetilde{I^\cup_G} \coloneqq \min_{\Sigma_{X,Y|M}} \frac{1}{2} \log\det \bigl(I + \Sigma_M^{-1} \Sigma_{XY,M}^{\vphantom{-}\T} \Sigma_{XY|M}^{-1} \Sigma_{XY,M}^{\vphantom{-}}\bigr) \quad \text{s.t.}\quad \Sigma_{XY|M} \suff 0
	\end{equation}
	where the optimization variable $\Sigma_{X,Y|M}$ is an off-diagonal block embedded within $\Sigma_{XY|M}$; all other matrices in the objective are constants that are derived from $\Sigma^P_{MXY}$.
\end{proposition}

\vspace{-6pt}
We solve the above optimization problem using projected gradient descent: we analytically derive the gradient and the projection operator for the constraint set as shown below.
Then, we use the RProp~\cite{hinton2018rprop} algorithm for gradient descent, which independently adjusts the learning rates for each optimization parameter (derivations and implementation details are in the supplementary material).

\begin{proposition} \label{prop_gpid_obj_grad_proj}
	The \textbf{objective} in Proposition~\ref{prop_gpid_optim} can be simplified to
	\begin{equation}
		f(\Sigma_{X,Y|M}) = \frac{1}{2} \log\det \bigl(I + H_Y^\T H_Y^{\vphantom{\T}} + B^\T S^{-1} B\bigr),
	\end{equation}
	where $B \coloneqq (H_X - \Sigma_{X,Y|M} H_Y)$ and $S \coloneqq (I - \Sigma_{X,Y|M}^{\vphantom{\T}} \Sigma_{X,Y|M}^\T)$.

	The \textbf{gradient} of the objective with respect to $\Sigma_{X,Y|M}$ is given by\todo{Fix missing parentheses}
	\begin{equation}
		\nabla f(\Sigma_{X,Y|M}) = S^{-1} B \bigl(I + H_Y^\T H_Y^{\vphantom{\T}} + B^\T S^{-1} B\bigr)^{-1} \bigl( B^\T S^{-1} \Sigma_{X,Y|M} - H_Y^\T \bigr).
	\end{equation}
	A \textbf{projection operator} on to the constraint set $\Sigma_{XY|M} \suff 0$ can be obtained as follows: let $\Sigma_{XY|M} \eqqcolon V \Lambda V^\T$ be the eigenvalue decomposition of $\Sigma_{XY|M}$, with $\Lambda \eqqcolon \textup{diag}(\lambda_i)$. Let $\widebar{\lambda}_i \coloneqq \textup{max}(0, \lambda_i)$ represent the rectified eigenvalues, and $\widebar\Lambda \coloneqq \textup{diag}(\widebar{\lambda}_i)$. Then, define
	\newcommand{\proj}{\textit{proj}}
	\begin{align}
		\widebar\Sigma_{XY|M} &\coloneqq V \widebar\Lambda V^\T, \\
		\Sigma_{X,Y|M}^\proj &\coloneqq \widebar\Sigma_{X|M}^{-1/2} \widebar\Sigma_{X,Y|M} \widebar\Sigma_{Y|M}^{-1/2}, \label{eq_norm_projection}
	\end{align}
	where $\widebar\Sigma_{X|M}$, $\widebar\Sigma_{Y|M}$ and $\widebar\Sigma_{X,Y|M}$ are submatrices of $\widebar\Sigma_{XY|M}$.
\end{proposition}


\section{Canonical Gaussian Examples}
\label{sec_examples}

In this section, we show how well our $\sim_G$-PID estimator performs on a series of examples of increasing complexity, which have known ground truth.
\citet{barrett2015exploration} showed that, for Gaussian distributions, the $\sim$-PID reduces to the MMI-PID (defined below), whenever $M$ is scalar.
These also happen to be cases when the optimal distribution $Q_{MXY}$ is Gaussian~\cite{venkatesh2022partial}, and thus the $\sim_G$-PID should recover the ground truth.
We then leverage additivity (Property~\ref{ppty_additivity}) to combine two or more simple examples into complex ones, where ground truth continues to be known.
\newcommand{\mmi}{{\text{MMI}}}
\begin{definition}[Minimum Mutual Information (MMI) PID] \label{def_mmi_pid}
	Let the redundant information be defined as the minimum of the two mutual informations:
	\begin{equation} \label{eq_mmi_ri}
		RI_\mmi(M : X ; Y) = \min\{ I(M ; X), I(M ; Y) \}.
	\end{equation}
	The remaining MMI-PID components, $UI_\mmi(M : X \setminus Y)$, $UI_\mmi(M : Y \setminus X)$ and $SI_\mmi(M : X ; Y)$, follow from equations~\eqref{eq_pid}--\eqref{eq_pidy}.
\end{definition}

\vspace{-6pt}
We first provide a Gaussian analog of Example~\ref{ex_binary} in Examples~\ref{ex_gauss_unique_only}--\ref{ex_gauss_synergy_only} (for $d_M = d_X = d_Y = 1$).
We will use the channel notation described in Equation~\eqref{eq_mxy_channel}.
Complete derivations for these examples (and some nuances that are omitted here) are presented in the supplementary material.

\begin{example}[Pure uniqueness: variable $A$ from Example~\ref{ex_binary}] 
	\label{ex_gauss_unique_only}
	Suppose $M \sim \mathcal N(0, 1)$, $H_X = 1$ and $H_Y = 0$, with $N_X, N_Y \sim \text{i.i.d.\ } \mathcal N(0, 1)$.
	Here, only $X$ receives information about $M$, while $Y$ is pure noise.
	Thus, $X$ has unique information about $M$ ($UI_X = I(M ; X) > 0$), with no unique information in $Y$, and no redundancy or synergy ($UI_Y = RI = SI = 0$).
\end{example}

\begin{example}[Pure redundancy: variable $B$ from Example~\ref{ex_binary}] 
	\label{ex_gauss_redundant_only}
	Ideally, we would set $M \sim \mathcal N(0, 1)$, $X = M$ and $Y = M$.
	However, for continuous random variables, $I(M ; X) = \infty$ when $M = X$.
	So instead, we set $M \sim \mathcal N(0, 1)$, $H_X = 1$ and $H_Y = 1$, with $N_X \sim \mathcal N(0, 1)$ while $N_Y = N_X$ (i.e., $X = Y$, so they are both the same noisy version of $M$).
	In this case, $X$ and $Y$ are fully redundant since they both contain exactly the same information about $M$.
	Thus, $RI = I(M ; (X, Y)) > 0$, while $UI_X = UI_Y = SI = 0$.
\end{example}

\begin{example}[Pure synergy: variable $C$ from Example~\ref{ex_binary}] 
	\label{ex_gauss_synergy_only}
	We cannot replicate pure synergy for Gaussian variables, but we can approach it in a limit.
	Let $M \sim \mathcal N(0, 1)$, $H_X = 1$ and $H_Y = 0$, with $N_X \sim \mathcal N(0, \sigma^2)$ and $N_Y = N_X$ (i.e., $X = M + Y$).
	Further, let $\sigma^2 \to \infty$.
	In this case, $I(M ; Y) = 0$ and $I(M ; X) \to 0$ as $\sigma^2 \to \infty$, so $X$ and $Y$ \emph{individually} convey little to no information about $M$.
	However, we can recover information about $M$ from $X$ and $Y$ \emph{together} by taking their difference, since $X - Y = M$.
	Thus, $SI > 0$, while $UI_Y = RI = 0$ and $UI_X \to 0$.
\end{example}

\begin{figure}[t]
	\centering
	\includegraphics[width=0.9\linewidth]{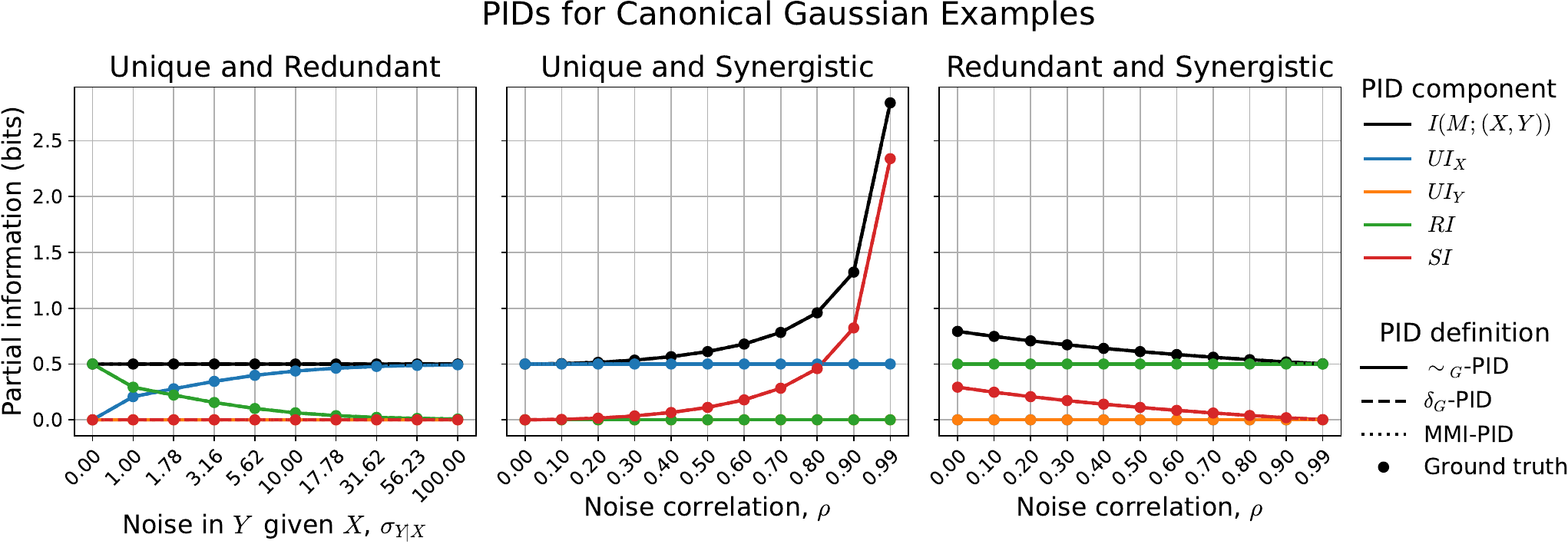}
	\caption{PID values for Examples \ref{ex_gauss_ui_ri}, \ref{ex_gauss_ui_si} and \ref{ex_gauss_ri_si}.
		The $\sim_G$-PID and the $\delta_G$-PID agree exactly with the MMI-PID, which is known to be the ground truth, since $M$ is scalar~\cite{barrett2015exploration}.\vspace{2mm}
	}
	\label{fig_gauss_pairwise_pid}
\end{figure}

\begin{figure}[t]
	\centering
	\includegraphics[width=0.6\linewidth]{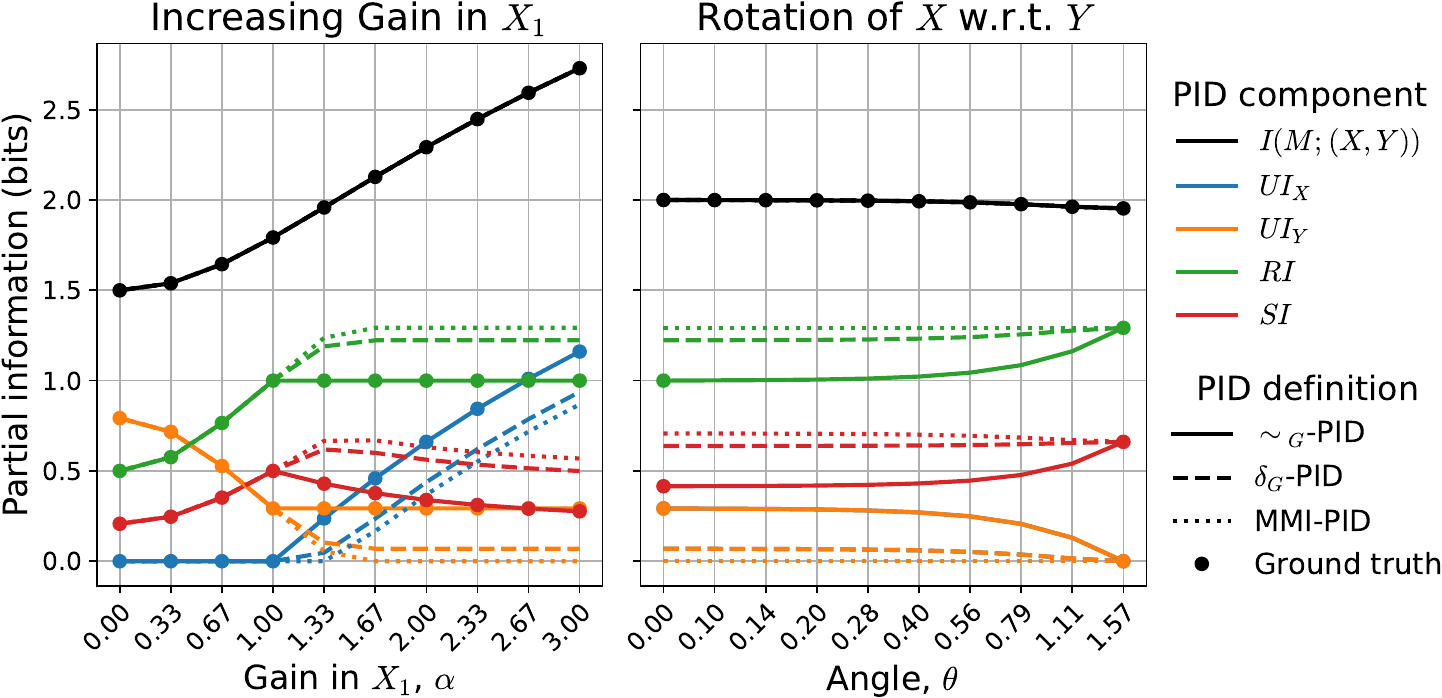}
	\caption{PID values for Examples \ref{ex_gain_sweep} (left) and \ref{ex_angle_sweep} (right), which combine two scalar examples with known ground truth, using Property~\ref{ppty_additivity}.
		The $\sim_G$-PID diverges from the $\delta_G$- and MMI-PIDs, and is the only one that agrees with the ground truth.}
	\label{fig_gauss_gain_angle_sweeps}
\end{figure}

\begin{figure}[t]
	\centering
	\includegraphics[width=0.8\linewidth]{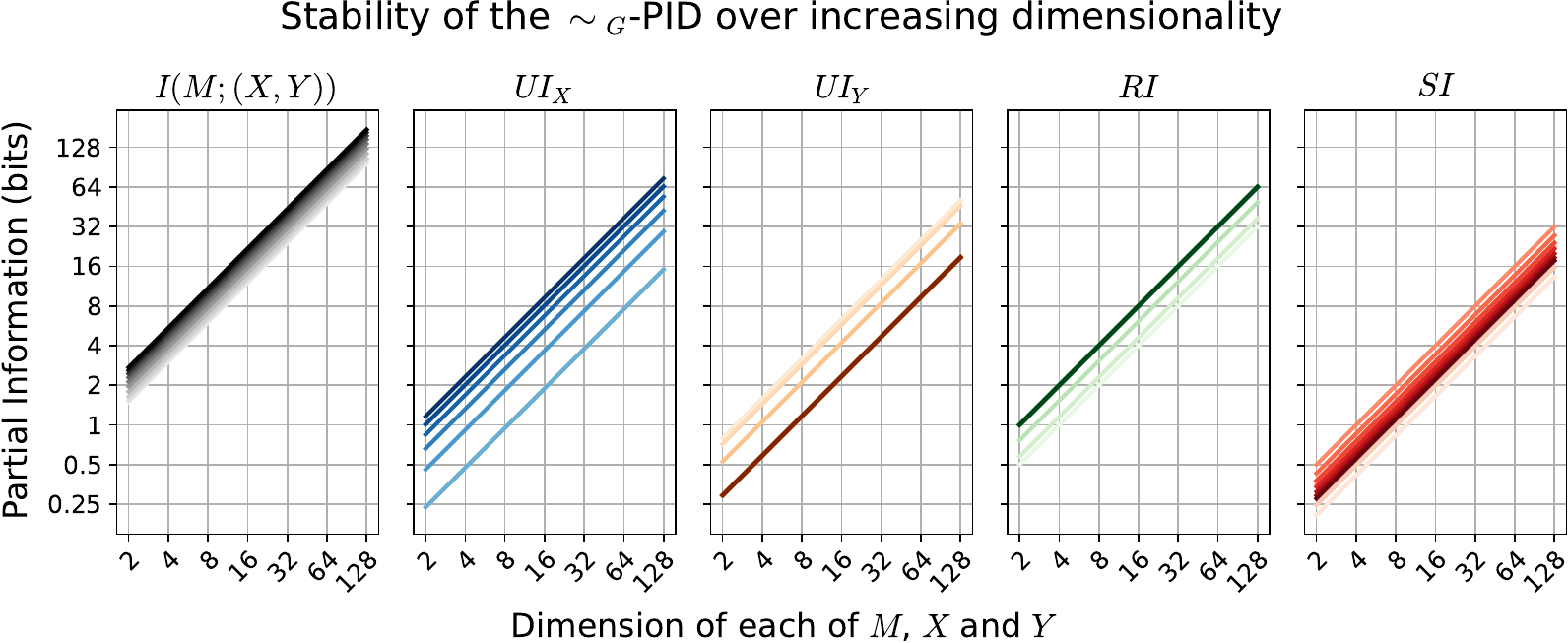}
	\caption{PID values for Example~\ref{ex_doubling}.
		Different shadings represent different values of gain in $X_1$ ($\alpha$) from Example~\ref{ex_gain_sweep}.
		The $\sim_G$-PID doubles every time $d$ doubles as seen by the constant 45\textdegree{} slope on the base-2 log-log plot, even when $d_M = d_X = d_Y = 128$.
	}
	\label{fig_doubling}
\end{figure}

\vspace{-6pt}
Examples~\ref{ex_gauss_unique_only}, \ref{ex_gauss_redundant_only} and \ref{ex_gauss_synergy_only} have been provided solely for intuition.
Their PIDs can be inferred directly from Equations~\eqref{eq_pid}--\eqref{eq_pidy}.
We next describe three one-dimensional examples that have each have \emph{two} non-zero PID components.
For lack of space, we only provide a brief description and defer details to the supplementary material.
We estimate the $\sim_G$-PID (as well as the $\delta_G$-PID \cite{venkatesh2022partial} and the ground-truth MMI-PID \cite{barrett2015exploration}) for these examples and show that all three are equal (see Fig.~\ref{fig_gauss_pairwise_pid}).

\begin{example}[Unique and redundant information] \label{ex_gauss_ui_ri}
	Let $X$ be a noisy representation of $M$, and let $Y$ be a noisy representation of $X$ with standard deviation $\sigma_{Y|X}$.
	When $Y = X$ (zero noise), this example reduces to Example~\ref{ex_gauss_redundant_only}.
	As $\sigma_{Y|X} \to \infty$, $RI$ reduces while $UI_X$ approaches $I(M ; X)$.
\end{example}

\begin{example}[Unique and synergistic information] \label{ex_gauss_ui_si}
	Let $M \sim \mathcal N(0, 1)$, $H_X = 1$, $H_Y = 0$ and $N_X, N_Y \sim \mathcal N(0, \sigma^2)$ such that their correlation is $\rho$.
	When $\sigma^2$ is finite and $\rho = 0$, this example reduces to Example~\ref{ex_gauss_unique_only}, since there can be no synergy between $X$ and $Y$.
	As $\rho \to 1$, $X - Y \to M$; so the total mutual information $I(M ; (X, Y)) \to \infty$, driven by synergy growing unbounded, while the unique component remains constant at $I(M ; X)$.
\end{example}

\begin{example}[Redundant and synergistic information] \label{ex_gauss_ri_si}
	Let $M \sim \mathcal N(0, 1)$, $H_X = H_Y = 1$ and $N_X, N_Y \sim \mathcal N(0, 1)$ such that their correlation is $\rho$.
	When $\rho < 1$, $I(M ; X)$ and $I(M ; Y)$ are both equal by symmetry, and thus equal to $RI$ (see Def.~\ref{def_mmi_pid} for the MMI-PID, which is ground truth here).
	As $\rho$ reduces, the two channels $X$ and $Y$ have noisy representations of $M$ with increasingly independent noise terms.
	Averaging the two, $(X + Y) / 2$, will provide more information about $M$ than either one of them individually (i.e., synergy), and thus $SI$ increases as $\rho$ reduces.
\end{example}

\vspace{-6pt}
The next set of examples will use the examples presented above in different combinations.
This ensures that, where possible, the ground truth remains known in accordance with Property~\ref{ppty_additivity}.
These examples are also designed to reveal the differences between the $\sim$-PID, the MMI-PID and the $\delta$-PID: in particular, they show how the MMI-PID and the $\delta$-PID fail where the $\sim$-PID does not.
These examples use two-dimensional $M$, $X$ and $Y$, i.e., $(d_M, d_X, d_Y) = (2, 2, 2)$.
A diagrammatic representation of Examples~\ref{ex_gain_sweep} and \ref{ex_angle_sweep} is given in the supplementary material.

\begin{example} \label{ex_gain_sweep}
	Let $X_1 = \alpha M_1 + N_{X,1}$, $Y_1 = M_1 + N_{Y,1}$, $X_2 = M_2 + N_{X,2}$ and $Y_2 = 3 M_2 + N_{Y,2}$, where $M_1, M_2, N_{X,i}, N_{Y,i} \sim $ i.i.d.\ $\mathcal N(0, 1)$, $i = 1, 2$.
	Here, $(M_1, X_1, Y_1)$ is independent of $(M_2, X_2, Y_2)$, therefore using Property~\ref{ppty_additivity}, we can add the PID values from their individual decompositions (which each have known ground truth via the MMI-PID since $M_1$ and $M_2$ are scalar).
	Fig.~\ref{fig_gauss_gain_angle_sweeps}(l) compares the $\sim_G$-PID, the $\delta_G$-PID and the MMI-PID for the joint decomposition of $I(M ; (X, Y))$, at different values of $\alpha$, the gain in $X_1$.
	Only the $\sim_G$-PID matches the ground truth, as it is the only definition here that is additive.
\end{example}

\begin{example} \label{ex_angle_sweep}
	Let $M$ and $Y$ be as in Example~\ref{ex_gain_sweep}.
	Suppose $X = H_X \; R(\theta) \; M$, where $H_X$ is a diagonal matrix with diagonal entries 3 and 1, and $R(\theta)$ is a $2 \times 2$ rotation matrix that rotates $M$ by an angle $\theta$.
	When $\theta = 0$, $X$ has higher gain for $M_1$ while $Y$ has higher gain for $M_2$. When $\theta$ increases to $\pi / 2$, $X$ and $Y$ have equal gains for both $M_1$ and $M_2$ (barring a difference in sign).
	Since $(M_1, X_1, Y_1)$ is not independent of $(M_2, X_2, Y_2)$ for all $\theta$, we know the ground truth only at the end-points.
	Nonetheless, the example shows a difference between the three definitions, as shown in Fig.~\ref{fig_gauss_gain_angle_sweeps}(r).
\end{example}

\begin{example} \label{ex_doubling}
	In this example, we test the stability of the $\sim_G$-PID as the dimensionality, $d \coloneqq d_M = d_X = d_Y$ increases.
	By Property~\ref{ppty_additivity}, if we take two i.i.d.\ systems of variables $(M, X, Y)$ at dimensionality $d$ and concatenate their respective variables, every PID component of the composite system of dimensionality $2d$ should be double that of the original.
	This process can be repeated, taking two independent $2d$-dimensional systems and concatenating them to create a $4d$-dimensional system.
	Fig.~\ref{fig_doubling} shows precisely this process starting with the system in Example~\ref{ex_gain_sweep} with $d = 2$, and continually doubling its size until $d = 128$.
	The $\sim_G$-PID accurately matches ground truth by doubling in value, and remains stable with small relative errors (shown in the supplementary material).
\end{example}


\section{Estimation and Bias-correction for the $\sim_G$-PID}
\label{sec_estimation}

\begin{figure}[t]
	\centering
	\includegraphics[width=0.5\linewidth]{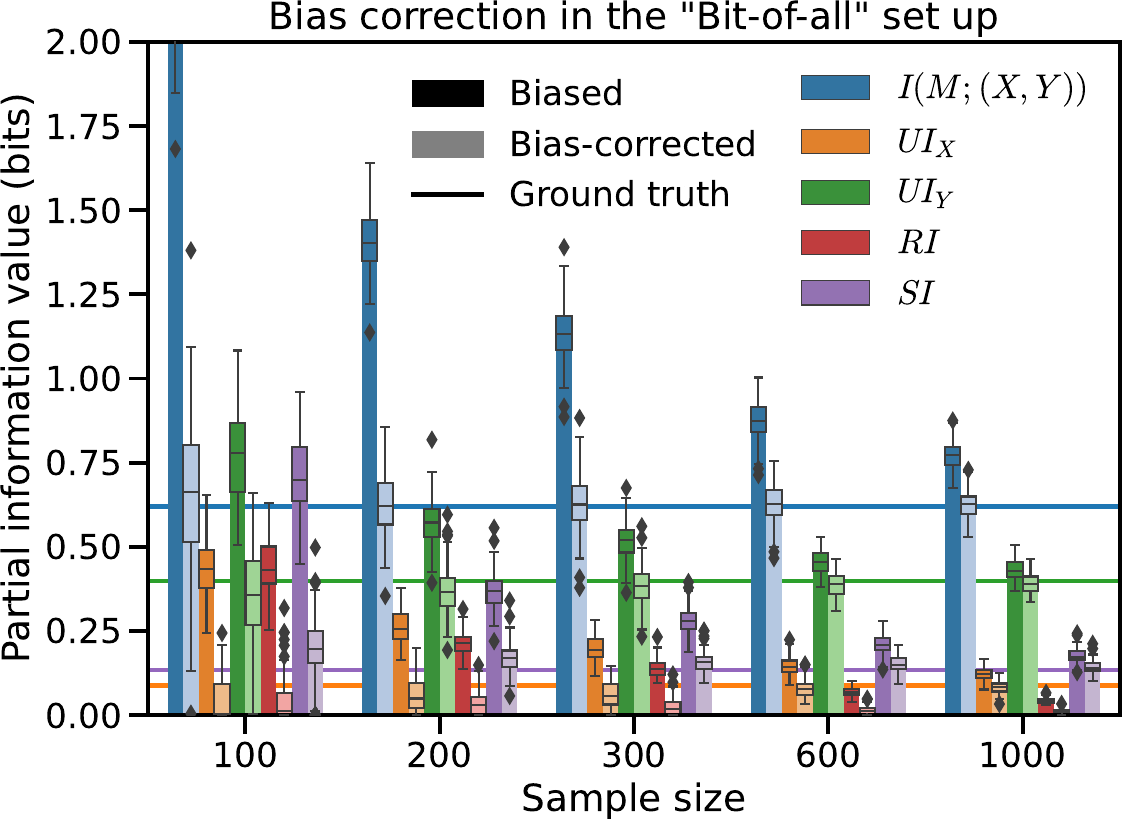}
	\includegraphics[width=0.5\linewidth]{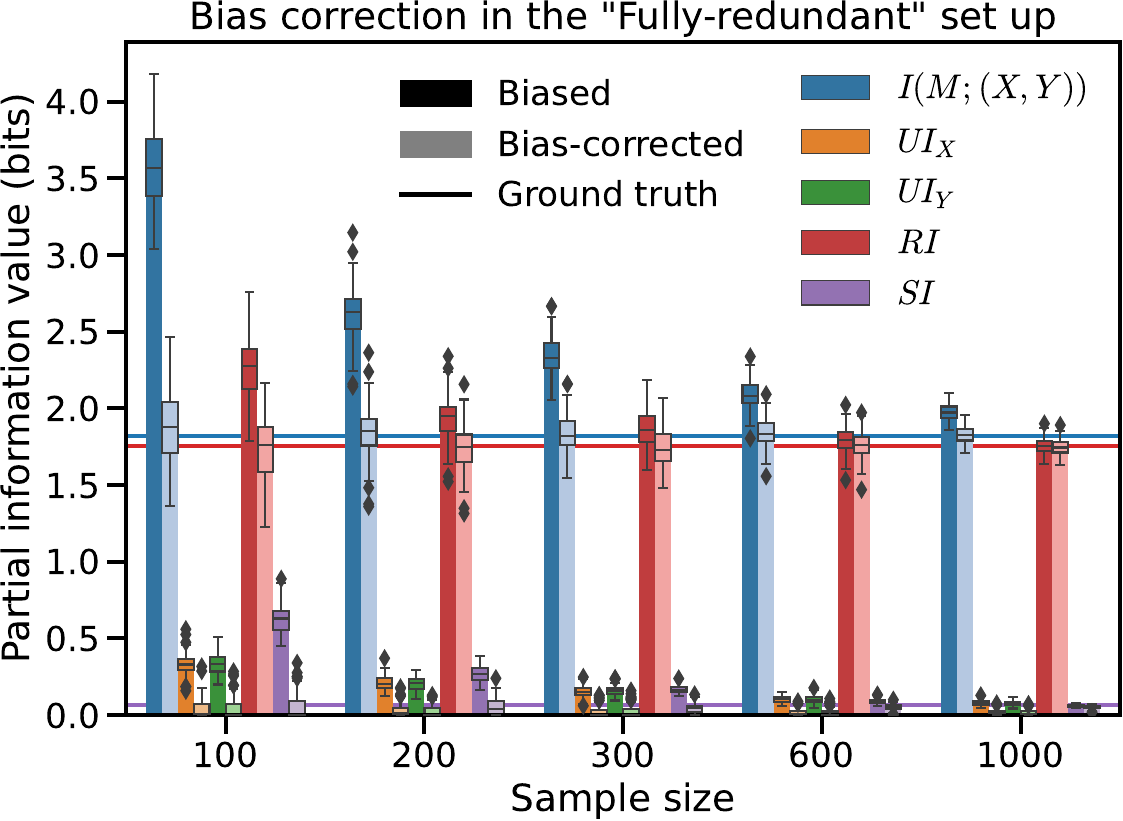}
	\caption{Empirical evaluation of bias-corrected PID estimates with increasing sample size for two configurations with $d_M = d_X = d_Y = 10$, as described in Section~\ref{sec_estimation}.
		Solid horizontal lines represent ground truth (as computed from the true covariance matrix);
	    dark colored bars represent biased PID components and light colored bars represent bias-corrected PID components, as estimated from the sample covariance matrix.
		Overlaid box plots indicate results from 100 random draws.
		Empirically, we find our bias-corrected estimates are both unbiased and consistent.}
	\label{fig_bias_correction}
\end{figure}

Having discussed how to compute the $\sim_G$-PID and shown that it agrees well with ground truth in several canonical examples, we discuss how the $\sim_G$-PID may be estimated from data.
Given a sample of $n$ realizations of $M$, $X$ and $Y$ drawn from $P_{MXY}$, we may estimate the sample joint covariance matrix $\hat{\Sigma}_{MXY}$.
The straightforward, ``plug-in'' estimator for the $\sim_G$-PID is to use the sample covariance matrix in the optimization problem in equation~\eqref{eq_union_info_obj}.

However, it is well-known that estimators of information-theoretic quantities suffer from large biases for moderate sample sizes~\cite{paninski2003estimation}.
\citet{cai2015law} characterized the bias in the entropy of a $d$-dimensional Gaussian random vector, for a fixed sample size $n$.
\begin{proposition}[Bias in Gaussian entropy \cite{cai2015law}]
	Suppose $M \in \mathbb R^{d_M}$ has an auto-covariance matrix $\Sigma_M$.
	The entropy of $M$ is $H(M) = \frac{1}{2} \log\det(2 \pi e \Sigma_M)$ when $\Sigma_M$ is known~\cite{cover2012elements}.
	For the sample covariance matrix $\hat{\Sigma}_M$, the bias is given by:
	\begin{equation} \label{eq_bias_entropy}
		\mathrm{Bias} \thinsp\bigl[\thinsp \hat H(M) \thinsp\bigr] = \sum_{k=1}^{d_M} \log(1 - k/n).
	\end{equation}
\end{proposition}
\vspace{-2mm}
For a proof, we refer the reader to \cite[Corollary~2]{cai2015law}.
This result may be naturally extended to compute the bias of each of the mutual information quantities in the LHS of equations \eqref{eq_pid}--\eqref{eq_pidy}:
\begin{corollary}[Bias in Gaussian mutual information] \label{cor_imxy_bias}
	For the joint mutual information $I(M ; (X, Y))$,
	\begin{equation}
		\mathrm{Bias} \thinsp\bigl[\thinsp \hat I\bigl(M ; (X, Y)\bigr) \thinsp\bigr] = \sum_{k=1}^{d_M} \log(1 - k/n) \;+\; \sum_{k=1}^{\mathclap{d_X + d_Y}} \log(1 - k/n) \;-\; \sum_{k=1}^{\mathclap{d_M + d_X + d_Y}} \log(1 - k/n)
	\end{equation}
\end{corollary}
\vspace{-2mm}
This follows directly from the fact that $I(M ; (X, Y)) = H(M) + H(X, Y) - H(M, X, Y)$.
Similarly, we can compute the bias of $\hat I(M ; X)$ and $\hat I(M ; Y)$.
But this does not uniquely determine the bias in the individual PID components, and as with defining PIDs, we need to decompose the bias in Corollary~\ref{cor_imxy_bias} across the four PID components such that they agree with these constraints.
We solve this problem by defining a bias-corrected version of the union information from Proposition~\ref{prop_gpid_optim}.
\begin{definition}[Bias-corrected Union Information] \label{def_bias_corr_union_info}
	We assign the bias in the union information to be the \emph{same fraction} as the bias in the joint mutual information $I(M ; (X, Y))$. This gives rise to a bias-corrected estimate of the union information:
	\begin{equation}
		\widetilde{I^\cup_G} \Big|_{\textup{bias-corr}}(M : X ; Y) \coloneqq \widetilde{I^\cup_G}(M : X ; Y) \biggl(1 - \frac{\mathrm{Bias} \thinsp\bigl[\thinsp \hat I\bigl(M ; (X, Y)\bigr) \thinsp\bigr]}{\hat I\bigl(M ; (X, Y)\bigr)}\biggr).
	\end{equation}
\end{definition}
\vspace{-2mm}
We do not analyze theoretically whether the bias-corrected union information is consistent and unbiased, thus, it may still have some residual bias relative to the true union information.
However, we find empirically that this bias correction process works reasonably well and appears both consistent and unbiased, in a number of examples.
Figure~\ref{fig_bias_correction} shows biased and bias-corrected PID values for 100 runs of two configurations called ``Bit-of-all'' (with a little bit of each PID component) and ``Fully-redundant'' (which has predominantly redundancy), with $d_M = d_X = d_Y = 10$ (details and additional setups in the supplementary material).
We find that bias correction brings the PID values closer to their true values even at small sample sizes.
In the supplementary material, we also include a preliminary analysis of the variance of PID estimates using bootstrap~\cite[Ch.~8]{wasserman2004all}.


\section{Application to Simulated and Real Neural Data}
\label{sec_neuroscience}

\begin{figure}
	\centering
	\includegraphics[width=0.8\linewidth]{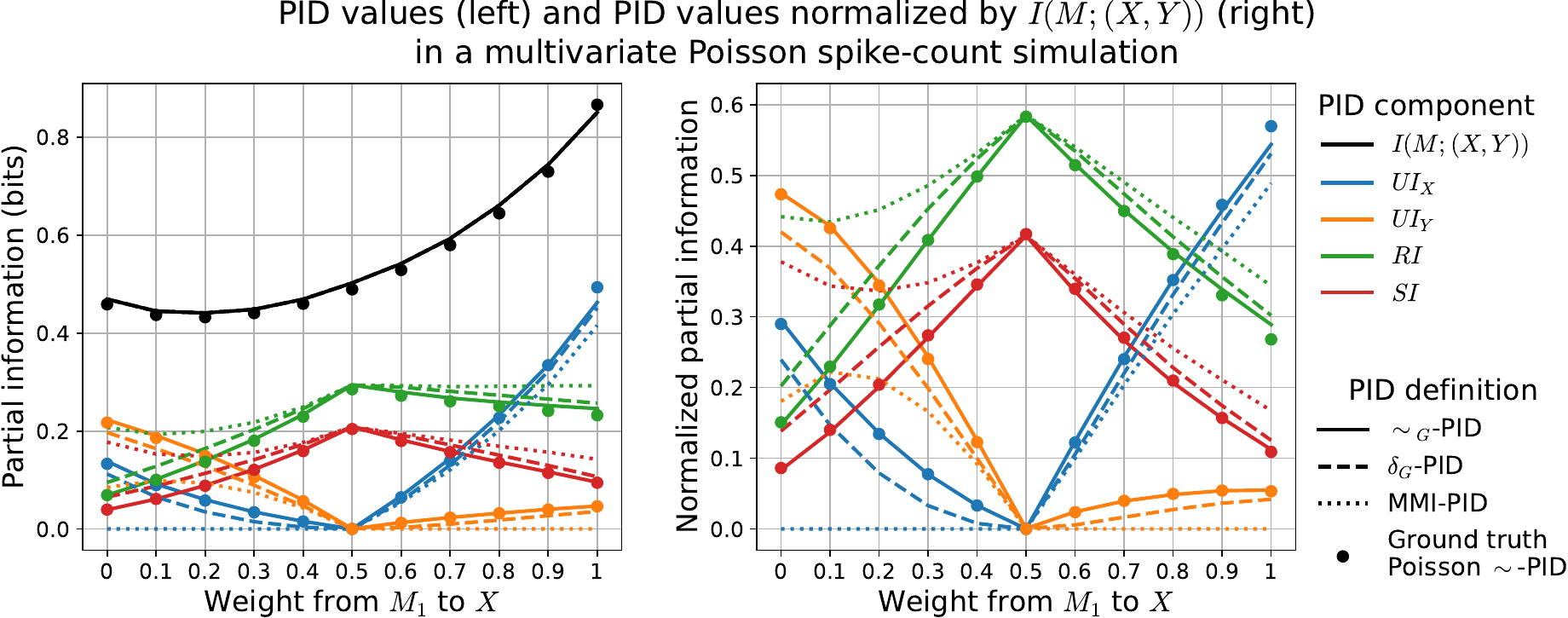}
	\caption{A comparison of the $\sim_G$-PID, the $\delta_G$-PID and the MMI-PID for a multivariate Poisson system.
		The ground truth is a discrete $\sim$-PID computed using the package of \citet{banerjee2018unique}.
		The $\sim_G$-PID comes closest to the ground truth (possibly because they compute the same PID definition), despite the fact that the $\sim_G$-PID only uses the covariance matrix of the Poisson distribution, whereas the ground truth uses knowledge of the distribution itself.}
	\label{fig_mult_poiss}
\end{figure}

\begin{figure}
	\centering
	\includegraphics[width=0.375\linewidth]{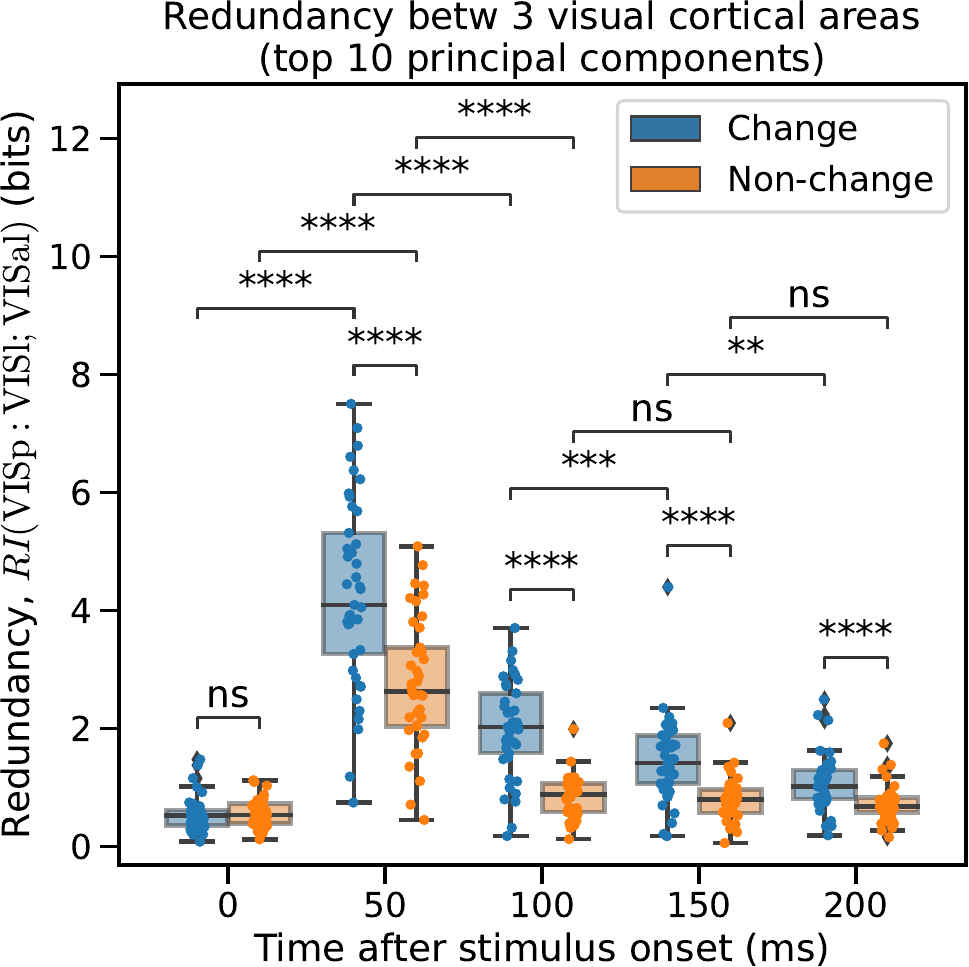}\;\;
	\includegraphics[width=0.39\linewidth]{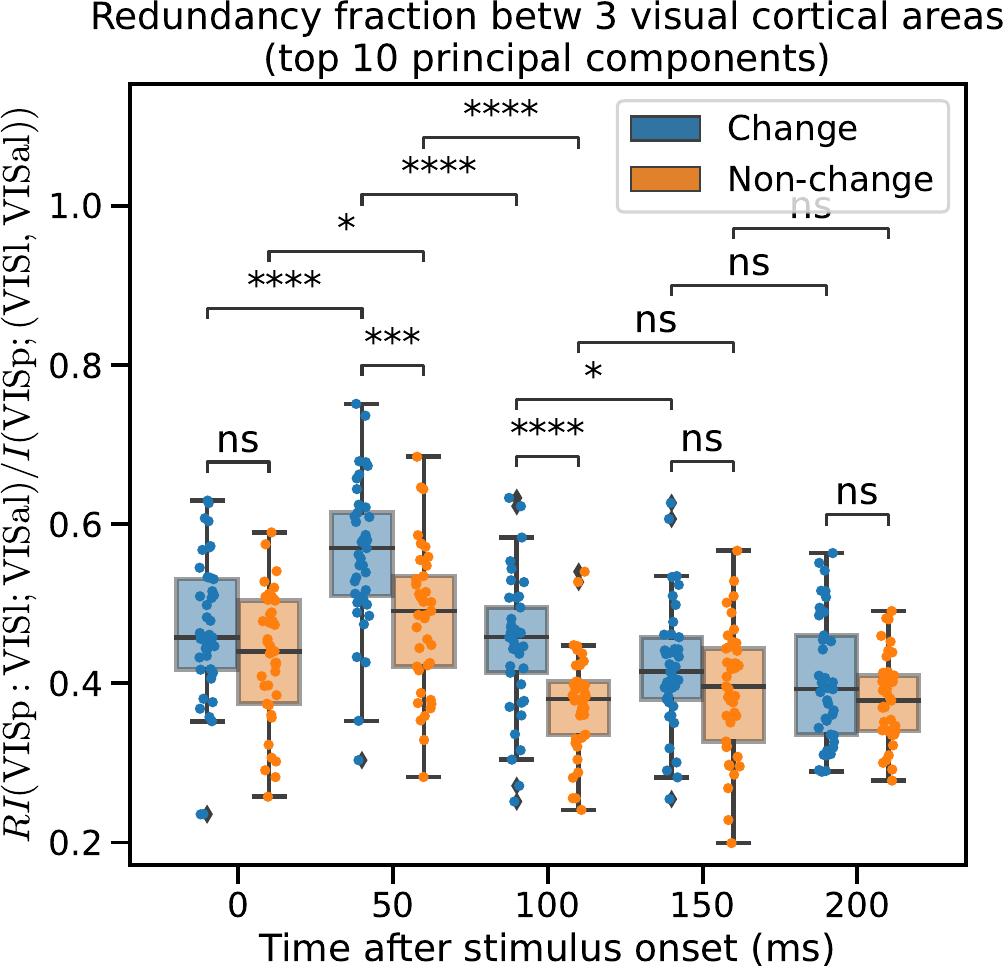}
	\caption{Bias-corrected redundancy estimates for information about VISp activity that is shared between VISl and VISal: redundancy in bits (left) and redundancy as a fraction of total mutual information (right).
		Data spread is across 40 mice.
		Statistical comparisons use a two-sided Mann-Whitney-Wilcoxon test.
	Observe that there is greater and more sustained redundancy on flashes corresponding to image changes, which are behaviorally relevant and linked to rewards in this task.\vspace{-3pt}}
	\label{fig_vbn_ri_visal}
\end{figure}

There is great interest in applying PIDs in neuroscientific applications, to understand how multiple brain regions jointly encode or communicate information~\cite{pica2017quantifying,timme2018tutorial}.
To show that our $\sim_G$-PID estimates provide reasonable results on non-Gaussian spiking neural data, we first simulate spike-count data using Poisson random variables (following~\cite{venkatesh2022partial}; described in the supplementary material).
We evaluate the ground truth $\sim$-PID for this distribution using the discrete PID estimator of \citet{banerjee2018computing}.
The $\sim_G$-PID is estimated from a sample covariance matrix using $10^6$ realizations of $M$, $X$ and $Y$.
We find that the $\sim_G$-PID closely matches the ground truth for a range of parameter values, despite the fact that the $\sim_G$-PID is effectively computed on a Gaussian approximation of a Poisson distribution (Fig.~\ref{fig_mult_poiss}).
We conclude that it is reasonable to use and interpret the $\sim_G$-PID on non-Gaussian spike count data.

We then applied our bias-corrected $\sim_G$-PID estimator to the Visual Behavior Neuropixels dataset collected by \ifarxiv{us at}{} the Allen Institute~\cite{allen2022vbn}.
\ifarxiv{We recorded over 80 mice}{Over 80 mice were recorded} using six neuropixels probes targeting various regions of visual cortex, while the mice were engaged in a visual change-detection task.
In the task, images from a set of 8 natural scenes were presented in 250~ms flashes, at intervals of 750~ms; the image would stay the same for a variable number of flashes after which it would change to a new image.
The mouse had to lick to receive a water reward when the image changed.
Thus, a given image flash could be a behaviorally relevant target if the previous image was different, or not, if the previous image was the same.

We used our bias-corrected PID estimator to understand how information is processed along the visual hierarchy during this task.
We estimated the $\sim_G$-PID to understand how information contained in the spiking activity of primary visual cortex (VISp) was represented in two higher-order visual cortical areas, VISl and VISal.
We aligned trials to the onset of a stimulus flash, binned spikes in 50~ms intervals and considered the top ten principal components (to achieve reasonable estimates at these sample sizes) from each region in each time bin.
We computed the $\sim_G$-PID on the sample covariance matrix of these principal components (shown in Fig.~\ref{fig_vbn_ri_visal}).
We found that there was a significantly larger amount of redundant information about VISp activity between VISl and VISal for stimulus flashes corresponding to an image change, compared to flashes that were not changes (Fig.~\ref{fig_vbn_ri_visal}(l)).
The larger redundancy was also sustained slightly longer for flashes corresponding to changes, than non-change flashes.
Both of these effects were maintained even when the redundancy was normalized by the joint mutual information, suggesting that the effect was not purely due to an increase in the total amount of information (Fig.~\ref{fig_vbn_ri_visal}(r)).
Our results suggest that the visual cortex propagates information throughout the hierarchy more robustly when such information is relevant for behavior.
\todo{Correct the number of mice used in the Figure caption: actually, it was 42 mice, not 40.}


\section{Discussion}

\textbf{Limitations. }
Our work has several limitations that require further theory and simulations to resolve, the most important of which are: \begin{enumerate*}[label=(\arabic*)]
	\item Our estimator is technically a bound on the PID values because we assume Gaussian optimality in Definition~\ref{def_gauss_tilde_pid};
	\item Our bias-correction method is heuristic: we do not provide a rigorous theoretical characterization of the bias of PID values.
\end{enumerate*}

\textbf{Broader impacts. }
Our work is mainly methodological, so the scope for negative impacts depends on how the methods might be used.
For example, incorrect interpretations drawn from our the use of our PID estimators may affect scientific conclusions.
Also, despite our best efforts to explore a variety of systems, we cannot tell how accurate our bias-correction method will be in novel configurations.


\begin{ack}
	We thank \L{}ukasz Ku\'smierz for providing a valuable reference on the bias of Shannon entropy estimates.
	We also thank Gabe Schamberg and Christof Koch for helpful discussions.

	P.~Venkatesh was supported by a Shanahan Family Foundation Fellowship at the Interface of Data and Neuroscience, supported in part by the Allen Institute.
	We thank the Allen Institute founder, Paul~G.~Allen, for his vision, encouragement, and support.
\end{ack}

\def\bibfont{\small}
\setlength{\bibsep}{6pt minus 1pt}
\bibliography{references}

\clearpage


\appendix
\begin{center}
	\Large
	\textbf{Gaussian Partial Information Decomposition:\\Bias Correction and Application to High-dimensional Data} \\
	Supplementary Material
\end{center}

\setlength{\abovedisplayskip}{6pt}
\setlength{\belowdisplayskip}{6pt}


\vspace{1em}
\section{Supplementary Material for Section~\ref{sec_computing_gpid}}

\newcommand{\vpm}{{\vphantom{-1}}}

\subsection{Proofs of Propositions~\ref{prop_gpid_optim} and \ref{prop_gpid_obj_grad_proj}}

\begin{proof}[Proof of Proposition~\ref{prop_gpid_optim}]
	Firstly, the differential entropy of a Gaussian random variable $M$ with covariance matrix $\Sigma_M$ is given by~\citesm[Thm.~8.4.1]{cover2012elements_sm}:
	\begin{equation} \label{eq_gauss_entropy}
		h(M) = \frac{1}{2}\log\det(2 \pi e \Sigma_M).
	\end{equation}

	Secondly, for a joint Gaussian distribution $P_{MXY}$ parameterized by a covariance matrix $\Sigma_{MXY}$, the conditional covariance matrix can be written as~\citesm[Sec.~8.1.3]{petersen2012matrix}:
	\begin{align}
		\Sigma_{XY|M}^{\vpm} &= \Sigma_{XY}^{\vpm} - \Sigma_{XY,M}^{\vpm} \Sigma_{M}^{-1} \Sigma_{XY,M}^{\vpm\T} \label{eq_cond_cov1} \\
		\Rightarrow\qquad \Sigma_{XY}^{\vpm} &= \Sigma_{XY|M}^{\vpm} + \Sigma_{XY,M}^{\vpm} \Sigma_{M}^{-1} \Sigma_{XY,M}^{\vpm\T} \label{eq_cond_cov2}
	\end{align}

	Using these two equations, we can derive the mutual information between $M$ and $(X, Y)$ as follows:
	\begin{align}
		I(M ; (X, Y)) &\overset{\hphantom{(a)}}{=} h(X, Y) - h(X, Y \given M) \\
					  &\overset{(a)}{=} \frac{1}{2}\log\det(2 \pi e \Sigma_{XY}) - \frac{1}{2} \log\det(2 \pi e \Sigma_{XY|M}) \\
					  &\overset{\hphantom{(a)}}{=} \frac{1}{2}\log\bigl((2 \pi e)^{d_M} \det(\Sigma_{XY})\bigr) - \frac{1}{2} \log\bigl((2 \pi e)^{d_M} \det(\Sigma_{XY|M})\bigr) \\
					  &\overset{\hphantom{(a)}}{=} \frac{1}{2}\log\biggl(\frac{\det(\Sigma_{XY})}{\det(\Sigma_{XY|M})}\biggr) \\
					  &\overset{(b)}{=} \frac{1}{2}\log\biggl(\frac{\det(\Sigma_{XY|M}^{\vpm} + \Sigma_{XY,M}^{\vpm} \Sigma_{M}^{-1} \Sigma_{XY,M}^{\vpm\T})}{\det(\Sigma_{XY|M})}\biggr) \\
					  &\overset{\hphantom{(a)}}{=} \frac{1}{2}\log\biggl(\frac{\det(\Sigma_{XY|M}^{\vpm}) \det(I + \Sigma_{XY|M}^{-1} \Sigma_{XY,M}^{\vpm} \Sigma_{M}^{-1} \Sigma_{XY,M}^{\vpm\T})}{\det(\Sigma_{XY|M})}\biggr) \\
					  &\overset{\hphantom{(a)}}{=} \frac{1}{2}\log\det(I + \Sigma_{XY|M}^{-1} \Sigma_{XY,M}^{\vpm} \Sigma_{M}^{-1} \Sigma_{XY,M}^{\vpm\T}) \\
					  &\overset{(c)}{=} \frac{1}{2}\log\det(I + \Sigma_{M}^{-1} \Sigma_{XY,M}^{\vpm\T} \Sigma_{XY|M}^{-1} \Sigma_{XY,M}^{\vpm}),
	\end{align}
	where in (a) we used equation~\eqref{eq_gauss_entropy}, in (b) we used equation~\eqref{eq_cond_cov2}, and in (c) we used the fact that $\det(I + AB) = \det(I + BA)$.

	The remainder of the proof follows from the arguments presented below equation~\eqref{eq_tilde_union_info}.
	The constraint in Proposition~\ref{prop_gpid_optim} arises because, when optimizing over $\Sigma_{X,Y|M}$, we require $\Sigma_{MXY}$ to be a valid positive semidefinite covariance matrix, i.e., $\Sigma_{MXY} \suff 0$.
	This happens if and only if $\Sigma_M$ and its Schur complement in $\Sigma_{MXY}$ are both positive semidefinite, i.e., $\Sigma_M \suff 0$ and $\Sigma_M^\vpm - \Sigma_{M,XY}^\vpm \Sigma_{XY}^{-1} \Sigma_{M,XY}^{\vpm\T} = \Sigma_{XY|M}^\vpm \suff 0$.
\end{proof}

\begin{proof}[Proof of Proposition~\ref{prop_gpid_obj_grad_proj}]
	The proof is divided into three parts consisting of derivations for the objective, the gradient and the projection operator.

\textbf{Objective. }
After whitening the $P_{X|M}$ and the $P_{Y|M}$ channels, and assuming that $\Sigma_M = I$ (see Remark~\ref{rem_whitening}), without loss of generality we have that
\begin{align}
	\Sigma_{X,M} &= \mathbb E\bigl[ (H_X M + N_X) M^\T \bigr] = H_X \mathbb E[MM^\T] = H_X \\
	\Sigma_{XY|M} &= \left[\begin{array}{c c}
			I & \Sigma_{X,Y|M} \\
			\Sigma_{X,Y|M}^\T & I
	\end{array}\right] \\
	\Rightarrow\quad \Sigma_M^{-1} \Sigma_{XY,M}^{\vpm\T} \Sigma_{XY|M}^{-1} \Sigma_{XY,M}^{\vpm} &= \left[\begin{array}{c c}
			H_X^\T & H_Y^\T
		\end{array}\right] \left[\begin{array}{c c}
			I & \Sigma_{X,Y|M} \\
			\Sigma_{X,Y|M}^\T & I
		\end{array}\right]^{-1} \left[\begin{array}{c}
			H_X \\
			H_Y
		\end{array}\right]
\end{align}

For the sake of brevity, let $\Sigma$ represent the optimization variable $\Sigma_{X,Y|M}$, and let $S$ be its Schur complement in $\Sigma_{XY|M}$, $I - \Sigma\Sigma^\T$.
Then, the inverse of $\Sigma_{XY|M}$ can be written as~\citesm[Sec.~9.1.5]{petersen2012matrix}:
\begin{equation}
	\Sigma_{XY|M}^{-1} = \left[\begin{array}{c c}
			S^{-1} & -S^{-1}\Sigma \\
			-\Sigma^\T S^{-1} & I + \Sigma^\T S^{-1} \Sigma
	\end{array}\right]
\end{equation}
Therefore, we get:
\begin{align}
	\Sigma_M^{-1} \Sigma_{XY,M}^{\vpm\T} \Sigma_{XY|M}^{-1} \Sigma_{XY,M}^{\vpm} &=
	\left[\begin{array}{c c}
			H_X^\T & H_Y^\T
	\end{array}\right]
	\left[\begin{array}{c c}
			S^{-1} & -S^{-1}\Sigma \\
			-\Sigma^\T S^{-1} & I + \Sigma^\T S^{-1} \Sigma
	\end{array}\right]
	\left[\begin{array}{c}
			H_X \\
			H_Y
	\end{array}\right] \\
	&= H_Y^\T H_Y + (H_X - \Sigma H_Y)^\T S^{-1} (H_X - \Sigma H_Y)
\end{align}

Thus, setting $B \coloneqq H_X - \Sigma H_Y$, the optimization problem in Proposition~\ref{prop_gpid_optim} reduces to
\begin{equation}
	\begin{aligned}
		\min_{\Sigma} \quad &\frac{1}{2} \log\det \bigl(I + H_Y^\T H_Y + B^\T S^{-1} B\bigr) \\
		\text{s.t.}\quad &\Sigma_{XY|M} \suff 0
	\end{aligned}
\end{equation}

\textbf{Gradient. }
Let the objective derived in the previous section be called $f(\Sigma)$, where $\Sigma \coloneqq \Sigma_{X,Y|M}$ as before.
We can compute the gradient of $f$ with respect to $\Sigma$ using standard identities from matrix calculus.
First, note that the gradient of a scalar function with respect to a matrix is itself a matrix with entries as follows:
\begin{equation}
	\nabla f(\Sigma)\Big\vert_{ij} = \frac{\partial f}{\partial \Sigma_{ij}}(\Sigma).
\end{equation}

Considering each element of this matrix:
\begin{align}
	\frac{\partial f}{\partial \Sigma_{ij}}(\Sigma) &\overset{\hphantom{(a)}}{=} \frac{1}{2} \frac{\partial}{\partial \Sigma_{ij}} \log\det(I + H_Y^\T H_Y + B^\T S^{-1} B) \bigg\vert_\Sigma \\
															  &\overset{(a)}{=} \frac{1}{2} \mathrm{Tr} \Big\{ (I + H_Y^\T H_Y + B^\T S^{-1} B)^{-1} \frac{\partial}{\partial \Sigma_{ij}} (I + H_Y^\T H_Y + B^\T S^{-1} B) \Big\} \bigg\vert_\Sigma \\
															  &\overset{(b)}{=} \frac{1}{2} \mathrm{Tr} \Big\{ (I + H_Y^\T H_Y + B^\T S^{-1} B)^{-1} \frac{\partial}{\partial \Sigma_{ij}} (B^\T S^{-1} B) \Big\} \bigg\vert_\Sigma, \label{eq_delfdelsigij_partial}
\end{align}
where in (a), we have used the identity $\partial \log\det(X) = \mathrm{Tr}\{X^{-1} \partial(X)\}$~\citesm[Sec.~2]{petersen2012matrix}, while in (b), we use the fact that only $B$ and $S$ depend on $\Sigma$ implicitly, with the other terms being constants.

Expanding the partial derivative alone, we get:
\begin{equation} \label{eq_btsinvb_partial}
	\frac{\partial}{\partial \Sigma_{ij}} (B^\T S^{-1} B) \Big\vert_\Sigma = \biggl[\frac{\partial}{\partial \Sigma_{ij}}(B^\T) \cdot S^{-1} B \;+\; B^\T \cdot \frac{\partial}{\partial \Sigma_{ij}} (S^{-1}) \cdot B \;+\; B^\T S^{-1} \cdot \frac{\partial}{\partial \Sigma_{ij}} (B) \biggr]_\Sigma,
\end{equation}
wherein
\begin{align}
	\frac{\partial}{\partial \Sigma_{ij}}(B) \Big\vert_\Sigma &= \frac{\partial}{\partial \Sigma_{ij}} (H_X - \Sigma H_Y) \Big\vert_\Sigma \\
															  &\overset{(b)}{=} - J^{ij} H_Y, \\
	\frac{\partial}{\partial \Sigma_{ij}}(S^{-1}) \Big\vert_\Sigma &\overset{(c)}{=} - S^{-1} \frac{\partial S}{\partial \Sigma_{ij}} S^{-1} \Big\vert_\Sigma \\
																   &= - S^{-1} \frac{\partial}{\partial \Sigma_{ij}}(I - \Sigma \Sigma^\T) S^{-1} \Big\vert_\Sigma \\
																   &\overset{(d)}{=} - S^{-1} (- J^{ij} \Sigma^\T - \Sigma J^{ij}{}^\T) S^{-1},
\end{align}
where $J^{ij}$ is the \emph{single-entry matrix}, containing a 1 at location $(i, j)$ and 0's everywhere else; in (b) and (d), we use the fact that $\partial X / \partial X_{ij} = J^{ij}$~\citesm[Sec.~9.7.6]{petersen2012matrix}; and in (c) we use the identity $\partial (X^{-1}) = X^{-1} \partial(X) X^{-1}$.
Therefore, \eqref{eq_btsinvb_partial} becomes
\begin{align}
	\MoveEqLeft \frac{\partial}{\partial \Sigma_{ij}} (B^\T S^{-1} B) \Big\vert_\Sigma \\*
	&= - (J^{ij} H_Y)^\T S^{-1} B \;+\; B^\T S^{-1} (J^{ij} \Sigma^\T \;+\; \Sigma J^{ij}{}^\T) S^{-1} B \;+\; B^\T S^{-1} (- J^{ij} H_Y) \\
	&= - H_Y^\T J^{ij}{}^\T S^{-1} B \;+\; B^\T S^{-1} J^{ij} \Sigma^\T S^{-1} B \;+\; B^\T S^{-1} \Sigma J^{ij}{}^\T S^{-1} B \;-\; B^\T S^{-1} J^{ij} H_Y.
\end{align}
Putting it all together, and letting $A \coloneqq (I + H_Y^\T H_Y + B^\T S^{-1} B)$, \eqref{eq_delfdelsigij_partial} becomes
\begin{align}
	\frac{\partial f}{\partial \Sigma_{ij}}(\Sigma) &= \frac{1}{2} \mathrm{Tr} \Bigl\{ A^{-1} \bigl(-\; H_Y^\T J^{ij}{}^\T S^{-1} B \;+\; B^\T S^{-1} J^{ij} \Sigma^\T S^{-1} B \notag \\
													&\qquad\qquad\qquad +\; B^\T S^{-1} \Sigma J^{ij}{}^\T S^{-1} B \;-\; B^\T S^{-1} J^{ij} H_Y \bigr) \Bigr\} \\
													&= \frac{1}{2} \Bigl[ -\; \mathrm{Tr} \bigl\{ A^{-1} H_Y^\T J^{ij}{}^\T S^{-1} B \bigr\} \;+\; \mathrm{Tr} \bigl\{ A^{-1} B^\T S^{-1} J^{ij} \Sigma^\T S^{-1} B \bigr\} \notag \\
													&\qquad\; +\; \mathrm{Tr} \bigl\{ A^{-1} B^\T S^{-1} \Sigma J^{ij}{}^\T S^{-1} B \bigr\} \;-\; \mathrm{Tr} \bigl\{ A^{-1} B^\T S^{-1} J^{ij} H_Y \bigr\} \Bigr] \\
													&\overset{(e)}{=} \frac{1}{2} \Bigl[ -\; \mathrm{Tr} \bigl\{ S^{-1} B A^{-1} H_Y^\T J^{ij}{}^\T \bigr\} \;+\; \mathrm{Tr} \bigl\{ \Sigma^\T S^{-1} B A^{-1} B^\T S^{-1} J^{ij} \bigr\} \notag \\
													&\qquad\;\; +\; \mathrm{Tr} \bigl\{ S^{-1} B A^{-1} B^\T S^{-1} \Sigma J^{ij}{}^\T \bigr\} \;-\; \mathrm{Tr} \bigl\{ H_Y A^{-1} B^\T S^{-1} J^{ij} \bigr\} \Bigr],
\end{align}
where in (e), we have used the fact that the trace of a matrix product is invariant under cyclic permutations of the matrices within the product.

Finally, using the fact that $\mathrm{Tr}\{W^\T J^{ij}\} = \mathrm{Tr}\{W J^{ij}{}^\T\} = W_{ij}$ for any matrix $W$~\citesm[Sec.~9.7.5]{petersen2012matrix},
\begin{align}
\frac{\partial f}{\partial \Sigma_{ij}}(\Sigma) &= \frac{1}{2} \Bigl[ - 2 \bigl( S^{-1} B A^{-1} H_Y^\T \bigr)_{ij} \;+\; 2 \bigl(S^{-1} B A^{-1} B^\T S^{-1} \Sigma \bigr)_{ij} \Bigr] \\
												&= \Bigl[ S^{-1} B A^{-1} \bigl( B^\T S^{-1} \Sigma - H_Y^\T \bigr) \Bigr]_{ij} \\
\Rightarrow\quad \nabla f(\Sigma) &= S^{-1} B A^{-1} \bigl( B^\T S^{-1} \Sigma - H_Y^\T \bigr) \\
								  &= S^{-1} B \bigl( I + H_Y^\T H_Y + B^\T S^{-1} B \bigr)^{-1} \bigl( B^\T S^{-1} \Sigma - H_Y^\T \bigr).
\end{align}

\textbf{Projection operator. }
Recall that the optimization variable, $\Sigma \coloneqq \Sigma_{X,Y|M}$ is an off-diagonal block of $\Sigma_{XY|M}$, which is the matrix upon which the constraint is defined:
\begin{equation}
	\Sigma_{XY|M} = \left[\begin{array}{c c}
			I & \Sigma \\
			\Sigma^\T & I
	\end{array}\right],
\end{equation}
wherein the diagonal blocks are identity due to Remark~\ref{rem_whitening}.
For the purposes of this section, let us suppose $\Sigma_{XY|M}$ is a function of $\Sigma$, $\Sigma_{XY|M} \eqqcolon g(\Sigma)$, so that the constraint may be written as $g(\Sigma) \suff 0$.
A suitable projection operator, therefore, will accept a value $\Sigma_0$ (that may violate $g(\Sigma_0) \suff 0$) and find a point $\Sigma^{proj}$ close to it that satisfies the constraint, i.e., $g(\Sigma^{proj}) \suff 0$.

We do not find the ``orthogonal'' projection operator, which has the minimum distance $\lVert \Sigma^{proj} - \Sigma_0 \rVert$ in some norm.
Instead, we propose a simple heuristic to find a $\Sigma^{proj}$ which satisfies the constraint.

If $\Sigma_0$ satisfies the constraint, then we are done, so let us assume that $g(\Sigma_0) \;\; \mathclap{/}\mathclap{\suff} \;\; 0$.
Then, we can find a matrix $\widebar{\Sigma}_{XY|M}$ which is close to $g(\Sigma_0)$ and satisfies $\widebar{\Sigma}_{XY|M} \suff 0$ as follows: let the eigenvalue decomposition of $g(\Sigma_0)$ be given by $V \Lambda V^\T$, with $\Lambda \eqqcolon \mathrm{diag}(\lambda_i)$ being the diagonal matrix consisting of its eigenvalues $\lambda_i$.
Then, since $g(\Sigma_0)$ is not positive semidefinite, $\exists\;i$ s.t. $\lambda_i < 0$.
We set $\widebar{\lambda}_i \coloneqq 0$ for all such $i$; effectively, $\widebar{\lambda}_i = \max\{0, \lambda_i\} \;\forall\; i$.
We then reconstruct the matrix using these ``rectified'' eigenvalues and set it to be $\widebar{\Sigma}_{XY|M} \coloneqq V \widebar{\Lambda} V^\T$, where $\widebar{\Lambda} = \mathrm{diag}(\widebar{\lambda}_i)$.

Now, we need to find $\Sigma$ such that $g(\Sigma) = \widebar{\Sigma}_{XY|M}$.
However, $\widebar{\Sigma}_{XY|M}$ may not have identity matrices on its diagonal blocks, i.e., it might not correspond to a whitened channel.
We therefore whiten $\widebar{\Sigma}_{XY|M}$ as follows:
\begin{equation} \label{eq_sigxym_whitened}
	\widebar{\Sigma}_{XY|M}^{\textit{whitened}} = \left[\begin{array}{c c}
		\widebar{\Sigma}_{X|M}^{-1/2} & 0 \\
		0 & \widebar{\Sigma}_{Y|M}^{-1/2}
	\end{array}\right] \widebar{\Sigma}_{XY|M} \left[\begin{array}{c c}
		\widebar{\Sigma}_{X|M}^{-1/2} & 0 \\
		0 & \widebar{\Sigma}_{Y|M}^{-1/2}
	\end{array}\right],
\end{equation}
where $\widebar{\Sigma}_{X|M}$ and $\widebar{\Sigma}_{Y|M}$ are the diagonal blocks of $\widebar{\Sigma}_{XY|M}$.
Crucially, since the matrix multiplying $\widebar{\Sigma}_{XY|M}$ on either side is itself (the inverse square-root of) a covariance matrix (and hence positive semidefinite), $\widebar{\Sigma}_{XY|M}^{\textit{whitened}}$ is also positive semidefinite.

Now, the off-diagonal block of $\widebar{\Sigma}_{XY|M}^{\textit{whitened}}$ will satisfy $g(\cdot) = \widebar{\Sigma}_{XY|M}^{\textit{whitened}} \suff 0$.
This off-diagonal block forms the output of our projection operation and can be written as
\begin{equation}
	\Sigma_{X,Y|M}^{proj} = \widebar{\Sigma}_{X|M}^{-1/2} \widebar{\Sigma}_{XY|M} \widebar{\Sigma}_{Y|M}^{-1/2},
\end{equation}
which comes directly from equation~\eqref{eq_sigxym_whitened}.
\end{proof}

\subsection{Details of $\sim_G$-PID Optimization and RProp Implementation}

The optimization problem for the $\sim_G$-PID, using projected gradient descent with RProp (mentioned in Section~\ref{sec_computing_gpid}), is implemented as follows:
\begin{enumerate}[leftmargin=2em]
	\item Let $\Sigma \coloneqq \Sigma_{X,Y|M}$ be shorthand for the optimization variable, and let $\mathrm{Proj}(\cdot)$ represent the projection operator defined in Prop.~\ref{prop_gpid_obj_grad_proj}.
		Let $\Sigma^{(i)}$ represent the value of $\Sigma$ at iteration $i$ of the optimization.
		Initialize $\Sigma^{(0)} = \mathrm{Proj}(H_X H_Y^+)$, where $H_Y^+$ is the pseudoinverse of $H_Y$.
	\item Evaluate the objective and the gradient as defined in Prop.~\ref{prop_gpid_obj_grad_proj}, at the current value of $\Sigma^{(i)}$.
		Compute the sign of (each element of) the gradient,
		\begin{equation}
			\psi(\Sigma^{(i)}) \coloneqq \mathrm{Sgn}\bigl(\nabla f(\Sigma^{(i)})\bigr).
		\end{equation}
		When computing the objective and the gradient, add a small regularization term to the computation of $S^{-1}$ (as defined in Prop.~\ref{prop_gpid_obj_grad_proj}): $S^{-1} = \bigl((1 + \epsilon) I - \Sigma \Sigma^\T\bigr)^{-1}$, where we take $\epsilon = 10^{-7}$.
	\item Update:
		\begin{equation}
			\Sigma^{(i+1)} = \mathrm{Proj}\bigl(\Sigma^{(i)} - \alpha^i \eta^{(i)} \odot \psi(\Sigma^{(i)})\bigr),
		\end{equation}
		where $\eta^{(i)}$ is a time-varying learning rate vector of the same dimension as $\Sigma$, describing the learning rate for each element of $\Sigma$; $\odot$ represents an element-wise (or Hadamard) product between vectors; and $\alpha \coloneqq 0.999$ is a constant, which when raised to the power of $i$, imposes a slow overall decay of the learning rate to promote convergence.
	\item $\eta^{(0)}$ is initialized to $10^{-3}$ and $\eta^{(i)}$ is updated as follows:
		\begin{equation}
			\eta^{(i+1)} = \eta^{(i)} \odot \beta^{- \psi(\Sigma^{(i+1)}) \odot \psi(\Sigma^{(i)})},
		\end{equation}
		where $\beta \coloneqq 0.9$ is a constant that determines how fast the learning rate increases or decreases; and all operations are carried out element-wise.
		Note that when some element of the gradient changes in sign, that element of $- \psi(\Sigma^{(i+1)}) \odot \psi(\Sigma^{(i)})$ will be positive, resulting in a decrease in that element of $\eta^{(i)}$.
		On the other hand, if the sign of some element of the gradient remains the same, then the learning rate for that component will increase by a factor of $1 / 0.9$.
	\item Stop when the absolute differences between the current objective and the previous objectives from the last 20 consecutive iterations are all less than $10^{-6}$ (``patience''), or when the maximum number of iterations is exceeded (set to $10^4$ iterations).
\end{enumerate}


\section{Supplementary Material for Section~\ref{sec_examples}}
\label{sec_sm_examples}

First, observe that by subtracting equation~\eqref{eq_pidx} from equation \eqref{eq_pid}, we have
\begin{equation} \label{eq_cmi_ui_si}
	\begin{aligned}
		I(M ; (X, Y)) - I(M ; X) &= UI_Y + SI \\
		\Rightarrow \qquad\qquad\qquad I(M ; Y \given X) &= UI_Y + SI.
	\end{aligned}
\end{equation}
Similarly, subtracting equations~\eqref{eq_pid} and \eqref{eq_pidy}, we get that $I(M ; X \given Y) = UI_X + SI$.
These two equations hold in general, and will be used in what follows.

\subsection{Details and Derivations for Examples in Section~\ref{sec_examples}}

\textbf{Example~\ref{ex_gauss_unique_only}} (Pure uniqueness)\textbf{.}
\begin{align}
	M &\sim \mathcal N(0, 1) \\
	X &= M + N_X               & N_X, N_Y &\sim \text{ i.i.d. } \mathcal N(0, 1) \\
	Y &= N_Y                   & (N_X, N_Y) &\indept M
\end{align}
\begin{proof}[Derivation of PID values in Example~\ref{ex_gauss_unique_only}]
	\begin{alignat}{3}
		Y \indept M& &\quad &\Rightarrow &\quad I(M ; Y) &= 0 \\
		  &&&\Rightarrow &\quad UI_Y + RI &= 0 \\
		UI_Y, RI \geq 0& & &\Rightarrow &\quad UI_Y = RI &= 0 \\
		  &&&\Rightarrow &\quad UI_X &= I(M ; X) \\
		  &&&\Rightarrow &\quad SI &= I(M ; (X, Y)) - I(M ; X) \\
		  &&&&&= I(M ; X) - I(M ; X) = 0.
	\end{alignat}
\end{proof}

\textbf{Example~\ref{ex_gauss_redundant_only}} (Pure redundancy)\textbf{.}
\begin{align}
	M &\sim \mathcal N(0, 1) \\
	X &= M + N_X                   & N_X &\sim \mathcal N(0, 1) \\
	Y &= M + N_X                   & N_X &\indept M
\end{align}
\begin{proof}[Derivation of PID values in Example~\ref{ex_gauss_redundant_only}]
	\begin{alignat}{3}
		I(M ; X \given Y) &= 0 &\quad &\Rightarrow &\quad UI_X + SI &= 0 \\
		I(M ; Y \given X) &= 0 &\quad &\Rightarrow &\quad UI_Y + SI &= 0 \\
						  &    &      &\Rightarrow &\quad RI &= I(M ; (X, Y)) = I(M ; X).
	\end{alignat}
\end{proof}

\textbf{Example~\ref{ex_gauss_synergy_only}} (Pure synergy)\textbf{.}
\begin{align}
	M &\sim \mathcal N(0, 1) \\
	X &= M + N_X                   & N_X &\sim \mathcal N(0, \sigma^2) \\
	Y &= N_X                       & N_X &\indept M
\end{align}
\begin{proof}[Derivation of PID values in Example~\ref{ex_gauss_synergy_only}]
	\begin{alignat}{3}
		Y \indept M& &\quad &\Rightarrow &\quad I(M ; Y) &= 0 \\
		           & &      &\Rightarrow &\quad UI_Y + RI &= 0 \\
		UI_Y, RI \geq 0& &  &\Rightarrow &\quad UI_Y = RI &= 0 \\
		           & &      &\Rightarrow &\quad UI_X &= I(M ; X) \\
		           & &      &\Rightarrow &\quad SI &= I(M ; (X, Y)) - I(M ; X) \\
		           & &      &            &         &= \infty - I(M ; X) = \infty.
	\end{alignat}
\end{proof}

It should be noted that certain nuances have been omitted in discussing Examples~\ref{ex_gauss_unique_only}--\ref{ex_gauss_synergy_only} above.
For instance, in Example~\ref{ex_gauss_redundant_only}, $\Sigma_{XY|M}$ is rank deficient and hence non-invertible, which would be an issue when computing the objective in Equation~\eqref{eq_union_info_obj}.
Also, in Example~\ref{ex_gauss_synergy_only}, $I(M ; (X, Y)) = \infty$, however this could be corrected by adding some noise to either $X$ or $Y$ so that their difference is a noisy representation of $M$.

\begin{figure}[t]
	\centering
	\includegraphics[width=0.25\linewidth]{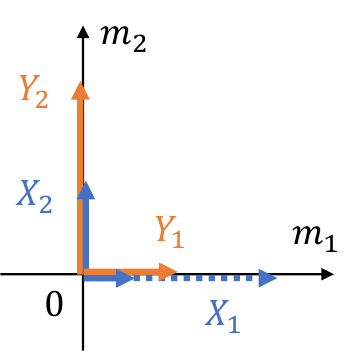}\qquad%
	\includegraphics[width=0.25\linewidth]{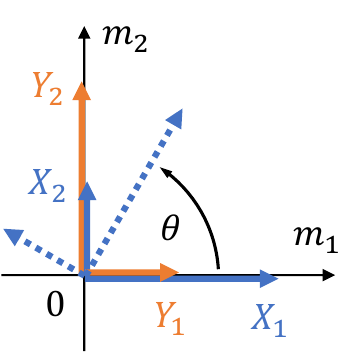}
	\caption{Diagrams explaining Examples~\ref{ex_gain_sweep} and \ref{ex_angle_sweep}. See Section~\ref{sec_gain_angle_cartoons} for details.}
	\label{fig_gain_angle_cartoons}
\end{figure}

\textbf{Example~\ref{ex_gauss_ui_ri}} (Unique and redundant information)\textbf{.}
\begin{alignat}{3}
	M &\sim \mathcal N(0, 1) \\
	X &= M + N_X              &  N_X &\sim \mathcal N(0, 1)              &  N_X &\indept M \\
	Y &= M + N_X + N_Y' \qquad& N_Y' &\sim \mathcal N(0, \sigma^2) \qquad& N_Y' &\indept (N_X, M)
\end{alignat}
\begin{proof}[Derivation of PID values in Example~\ref{ex_gauss_ui_ri}]
	Essentially, $X$ is a noisy representation of $M$, while $Y$ is a noisy representation of $X$.
	Since $M$---$X$---$Y$ forms a Markov chain, $I(M ; Y \given X) = 0$, and hence $UI_Y = SI = 0$.
	When $\sigma^2 = 0$, this example reduces to Example~\ref{ex_gauss_redundant_only} with only redundancy being present.
	For any finite non-zero value of $\sigma^2$, both $RI$ and $UI_X$ are present and are non-zero.
	Since $M$ is scalar, the redundancy for the $\sim$-PID is identical to the MMI-PID's redundancy~\cite{barrett2015exploration}:
	\begin{align}
		RI &= \min\{ I(M ; X), I(M ; Y) \} = I(M ; Y),
	\end{align}
	since $I(M ; Y) < I(M ; X)$, by the data processing inequality.
	At the limit when $\sigma^2 \to \infty$, $I(M ; Y) \to 0$, and therefore $RI \to 0$, while $UI_X$ will become equal to $I(M ; X)$.
\end{proof}

\textbf{Example~\ref{ex_gauss_ui_si}} (Unique and synergistic information)\textbf{.}
\begin{alignat}{2}
	M &\sim \mathcal N(0, 1) \\
	X &= M + N_X   \qquad& N_X, N_Y &\sim \mathcal N(0, \sigma^2), \qquad (N_X, N_Y) \indept M \\
	Y &= N_Y             & \mathrm{Corr}(N_X, N_Y) &= \rho
\end{alignat}
\begin{proof}[Derivation of PID values in Example~\ref{ex_gauss_ui_si}]
	When $\rho = 1$ and $\sigma^2 \to \infty$, this example reduces to Example~\ref{ex_gauss_synergy_only}, with only synergy being present.
	In general, $Y \indept M$, therefore, $I(M ; Y) = UI_Y + RI = 0$, meaning $UI_Y = RI = 0$.
	For any finite value of $\sigma^2$, $X$ will have some unique information about $M$ given by $UI_X = I(M ; X) > 0$.
	Correspondingly, $SI = I(M ; X \given Y) - UI_X = I(M ; X \given Y) - I(M ; X)$.
	When $\rho = 0$, $I(M ; X \given Y) = I(M ; X)$ and therefore $SI = 0$.
	As $\rho \to 1$, $X - Y \to M$; so the total mutual information $I(M ; (X, Y)) \to \infty$, driven by synergy growing unbounded, while the unique component remains finite at $I(M ; X)$.
\end{proof}

\begin{figure}[t]
	\centering
	\includegraphics[width=0.9\linewidth]{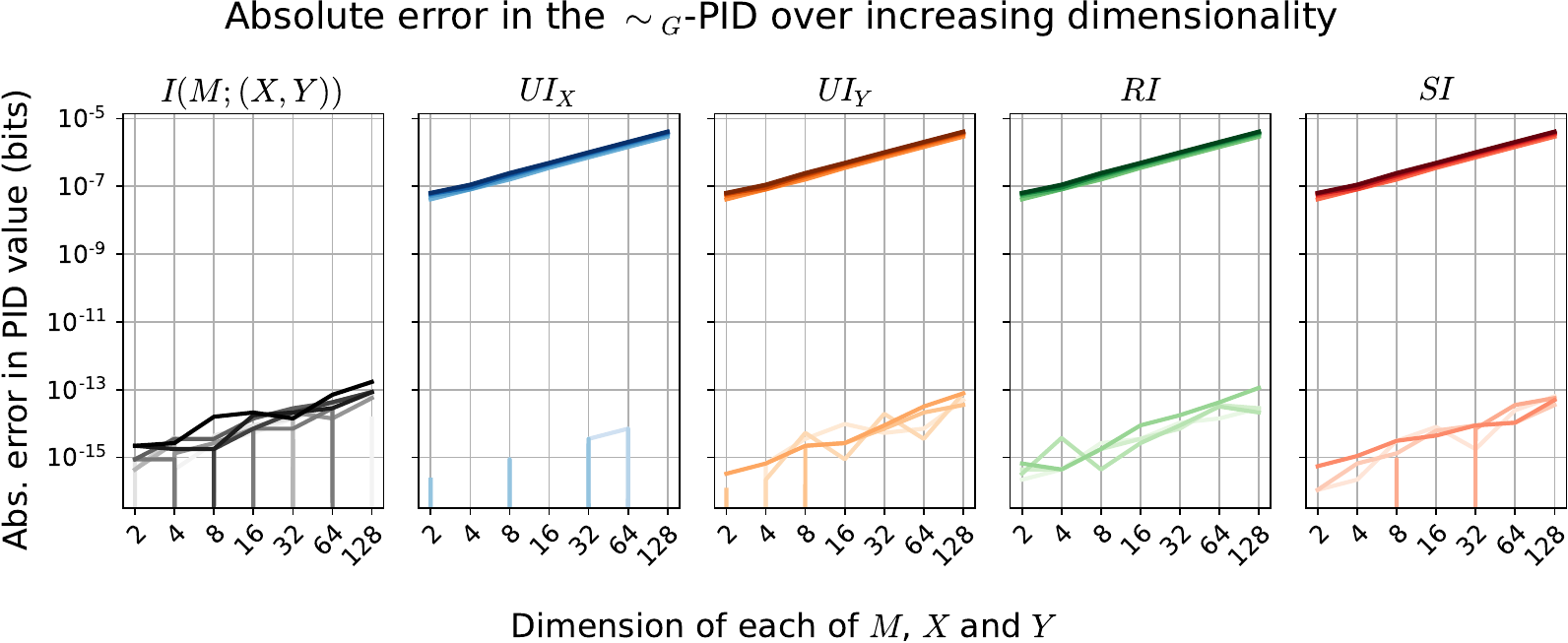}\\[12pt]
	\includegraphics[width=0.9\linewidth]{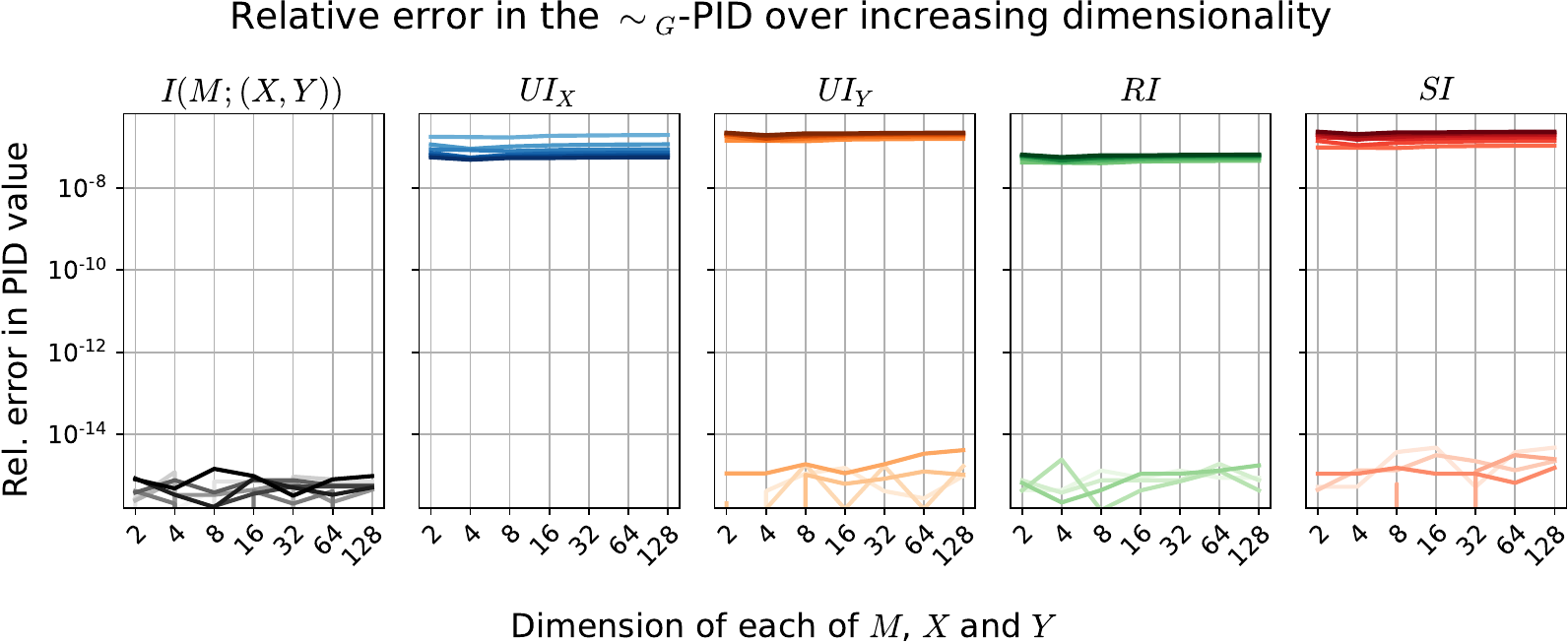}
	\caption{Absolute (top) and relative (bottom) errors in computed PID values from Example~\ref{ex_doubling}.}
	\label{fig_doubling_err}
\end{figure}

\textbf{Example~\ref{ex_gauss_ri_si}} (Redundant and synergistic information)\textbf{.}
\begin{alignat}{2}
	M &\sim \mathcal N(0, 1) \\
	X &= M + N_X   \qquad& N_X, N_Y &\sim \mathcal N(0, 1), \qquad (N_X, N_Y) \indept M \\
	Y &= M + N_Y         & \mathrm{Corr}(N_X, N_Y) &= \rho
\end{alignat}
\begin{proof}[Derivation of PID values in Example~\ref{ex_gauss_ri_si}]
	When $\rho = 1$, we once again reduce to Example~\ref{ex_gauss_redundant_only} with only redundancy.
	When $\rho < 1$, we cannot infer the PID values using Equation~\eqref{eq_pid} and non-negativity alone, since none of the individual mutual information values (or conditional mutual information values) go to zero.
	Instead, we can determine the redundancy using the MMI-PID since $M$ is scalar.
	Note that $I(M ; X)$ and $I(M ; Y)$ are both equal by symmetry, and thus equal to $RI$.
	This also implies that both $UI_X$ and $UI_Y$ must be equal to zero.
	As $\rho$ reduces, the two channels $X$ and $Y$ have noisy representations of $M$ with increasingly independent noise terms.
	Therefore, their average, $(X + Y) / 2$ will be more informative about $M$ than either one of them individually, meaning that $X$ and $Y$ jointly contain more information than any one individually.
	This extra information about $M$ is synergistic, given by $SI = I(M ; X \given Y)$, and increases as $\rho$ decreases, attaining its maximum possible value at $\rho = 0$.
\end{proof}

\subsection{Diagrams Explaining Examples~\ref{ex_gain_sweep} and \ref{ex_angle_sweep}}
\label{sec_gain_angle_cartoons}

Examples~\ref{ex_gain_sweep} and \ref{ex_angle_sweep} can be understood diagrammatically as shown in Fig.~\ref{fig_gain_angle_cartoons}(l) and Fig.~\ref{fig_gain_angle_cartoons}(r), respectively.
In both diagrams, we represent the two-dimensional plane describing $M$, with axes $m_1$ and $m_2$.
The colored vectors shown on this plane represent $H_X$ and $H_Y$, i.e., the gain with which $X$ and $Y$ represent each value of $M$.
For example, $Y_2$ captures only $M_2$, with a gain corresponding to its length.
The gains are directly representative of the signal-to-noise ratio (and hence the amount of information) in each variable, since the noise in each variable is i.i.d., with unit variance.
In Example~\ref{ex_gain_sweep}, the gain in $X_1$ is variable, while in Example~\ref{ex_angle_sweep}, the angle at which $X_1$ and $X_2$ sample $M_1$ and $M_2$ is variable.

\subsection{Absolute and Relative Errors in Example~\ref{ex_doubling}}

Figure~\ref{fig_doubling_err} shows how the absolute and relative errors in PID values scale with increasing dimensionality in Example~\ref{ex_doubling}.
The absolute errors are all less than $10^{-5}$, but increase in proportion to dimensionality.
The relative errors are all roughly constant, and remain under $10^{-6}$.


\section{Supplementary Material for Section~\ref{sec_estimation}}

\subsection{Implementation Details for Bias-correction}
\label{sec_bias_corr_desc}

We use a number of different setups based on sampling from random connectivity matrices for bias correction in Section~\ref{sec_estimation}.
All of these setups assume that $d_X = d_Y$.

The \textbf{both-unique}, \textbf{fully-redundant} and \textbf{high-synergy} setups have the following in common:
\begin{align}
	\Sigma_M &= I \\
	\Sigma_{X|M} &= I \\
	\Sigma_{Y|M} &= I \\
	\Sigma_{MXY} &= \left[\begin{array}{c c c}
		I & H_X^\T & H_Y^\T \\
		H_X & H_X H_X^\T + \Sigma_{X|M} & H_X H_Y^\T + \Sigma_W \\
		H_X & H_Y H_X^\T + \Sigma_{W}^\T & H_Y H_Y^\T + \Sigma_{Y|M}
	\end{array}\right].
\end{align}
Also, the elements of $H_X$ are either zero or one, $H_X(i,j) \sim$ i.i.d.~Ber(0.1).
These three setups differ in their definitions of $H_Y$ (the channel gain from $M$ to $Y$) and $\Sigma_W$ (which controls the extent of correlation between $X$ and $Y$).

The \textbf{both-unique} setup draws $H_Y(i,j) \sim$ i.i.d.~Ber(0.1), with all elements of $H_Y$ independent of the elements of $H_X$, and sets $\Sigma_W = 0$.

The \textbf{fully-redundant} setup is similar to Example~\ref{ex_gauss_ri_si}, by setting $H_Y = H_X$ and $\Sigma_W = 0.9 I$ (note that $\Sigma_W$ is square, since $d_X = d_Y$).
By keeping $\Sigma_W$ close to the identity matrix, we are effectively in the regime with high correlation $\rho$ in Example~\ref{ex_gauss_ri_si}.
This allows us to come close to emulating Example~\ref{ex_gauss_redundant_only}, without suffering from the issue of non-invertibility of $\Sigma_{XY|M}$, mentioned in Section~\ref{sec_sm_examples}.

The \textbf{high-synergy} setup is similar to Example~\ref{ex_gauss_ui_si}, by setting $H_Y = 0$ and $\Sigma_W = 0.8 I$.
As with the fully-redundant setup, by keeping $\Sigma_W$ close to the identity matrix, we are in the high-$\rho$ regime.
This allows us to come close to emulating Example~\ref{ex_gauss_synergy_only}, while not making the synergy or the total mutual information infinite.

The \textbf{zero-synergy} setup is similar to Example~\ref{ex_gauss_ui_ri}, and uses the following setup:
\begin{align}
	\Sigma_M &= \Sigma_{X|M} = I \\
	\Sigma_X &= H_X H_X^\T + \Sigma_{X|M} \\
	\Sigma_{Y|X} &= I \\
	\Sigma_{MXY} &= \left[\begin{array}{c c c}
		I & H_X^\T & H_Y^\T \\
		H_X & \Sigma_X & \Sigma_X H_Y'^\T \\
		H_X & H_Y' \Sigma_X^\T & H_Y' \Sigma_X H_Y'^\T + \Sigma_{Y|X}
	\end{array}\right].
\end{align}
Here, $H_X(i,j) \sim$ i.i.d.~Ber(0.1), while $H_Y = H_Y' H_X$, with $H_Y'(i, j) \sim$ i.i.d.~Ber(0.1), $H_Y' \indept H_X$.
Defined this way, $M$---$X$---$Y$ form a Markov chain, ensuring that $I(M ; X \given Y) = 0$, so that $SI = 0$ (Refer equation~\eqref{eq_cmi_ui_si}).

The \textbf{bit-of-all} setup is a combination of equal parts of the high-synergy and zero-synergy setups.
The variables $X$ and $Y$ are swapped in the zero-synergy setup, so that both $X$ and $Y$ can have some unique information.

\begin{remark}[Rectification]
	In practice, we observed that the bias correction procedure prescribed in Definition~\ref{def_bias_corr_union_info} could lead to negative values for certain PID quantities.
	This occurred because the bias-corrected union information was not guaranteed to satisfy certain bounds, which we enforce below.
	To prevent the occurrence of negative PID values after bias-correction, we require a form of rectification:
	\begin{align}
		\widetilde{I_G^\cup}\Big\vert_{\text{rect}}^{(1)} &\coloneqq \max\Bigl\{\widetilde{I_G^\cup}\big\vert_{\text{bias-corr}} \;,\; \hat I(M ; X)\big\vert_{\text{bias-corr}} \;,\; \hat I(M ; Y)\big\vert_{\text{bias-corr}}\Bigr\} \\
		\widetilde{I_G^\cup}\Big\vert_{\text{rect}}^{(2)} &\coloneqq \min\Bigl\{\widetilde{I_G^\cup}\big\vert_{\text{rect}}^{(1)} \;,\; \hat I(M ; X)\big\vert_{\text{bias-corr}} + \hat I(M ; Y)\big\vert_{\text{bias-corr}} \;,\; \hat I(M ; (X, Y))\big\vert_{\text{bias-corr}} \Bigr\},
	\end{align}
	where $\hat I(\cdot)\vert_{\text{bias-corr}}$ represents a bias-corrected mutual information estimate.
	After the second rectification equation above, the union information is bounded from below by the individual (bias-corrected) mutual information values, and bounded from above by the sum of the individual mutual information values, and by the total mutual information.
\end{remark}

\subsection{Bias-correction Performance in Additional Setups and at Higher Dimensionality}

\begin{figure}[p]
	\centering
	\includegraphics[width=0.49\linewidth]{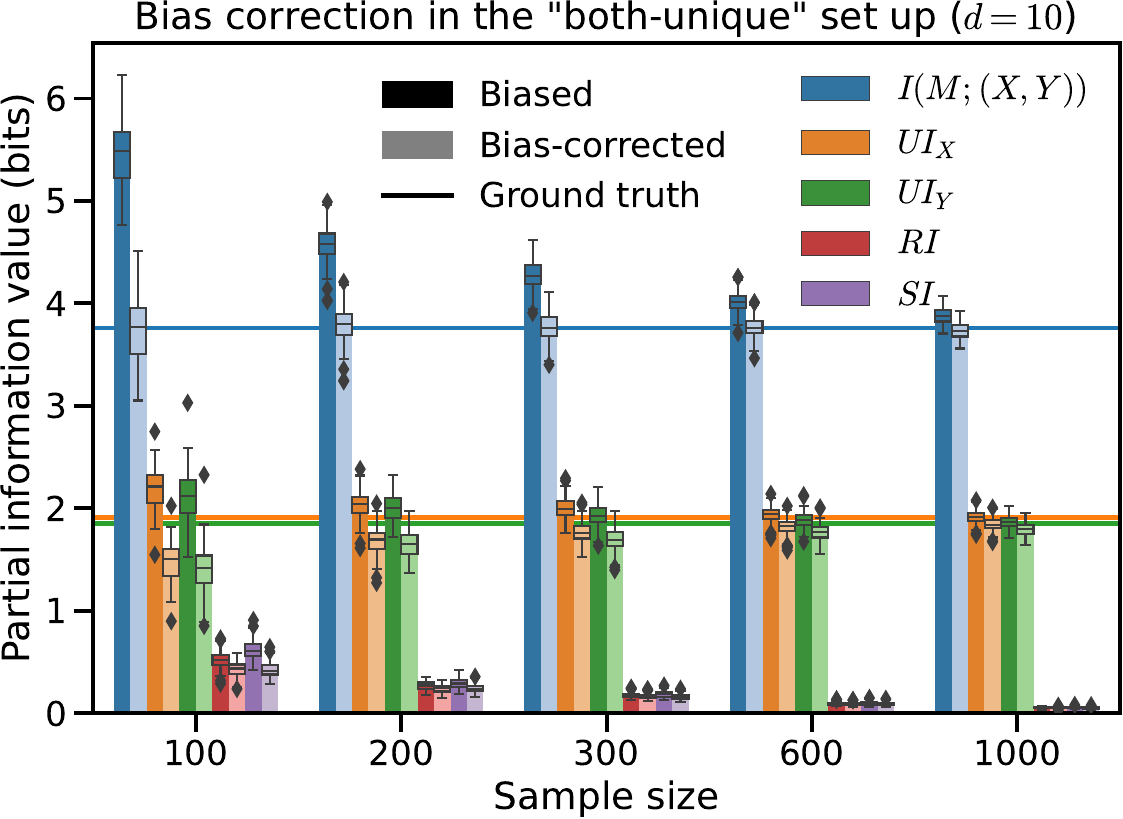}\hfill%
	\includegraphics[width=0.49\linewidth]{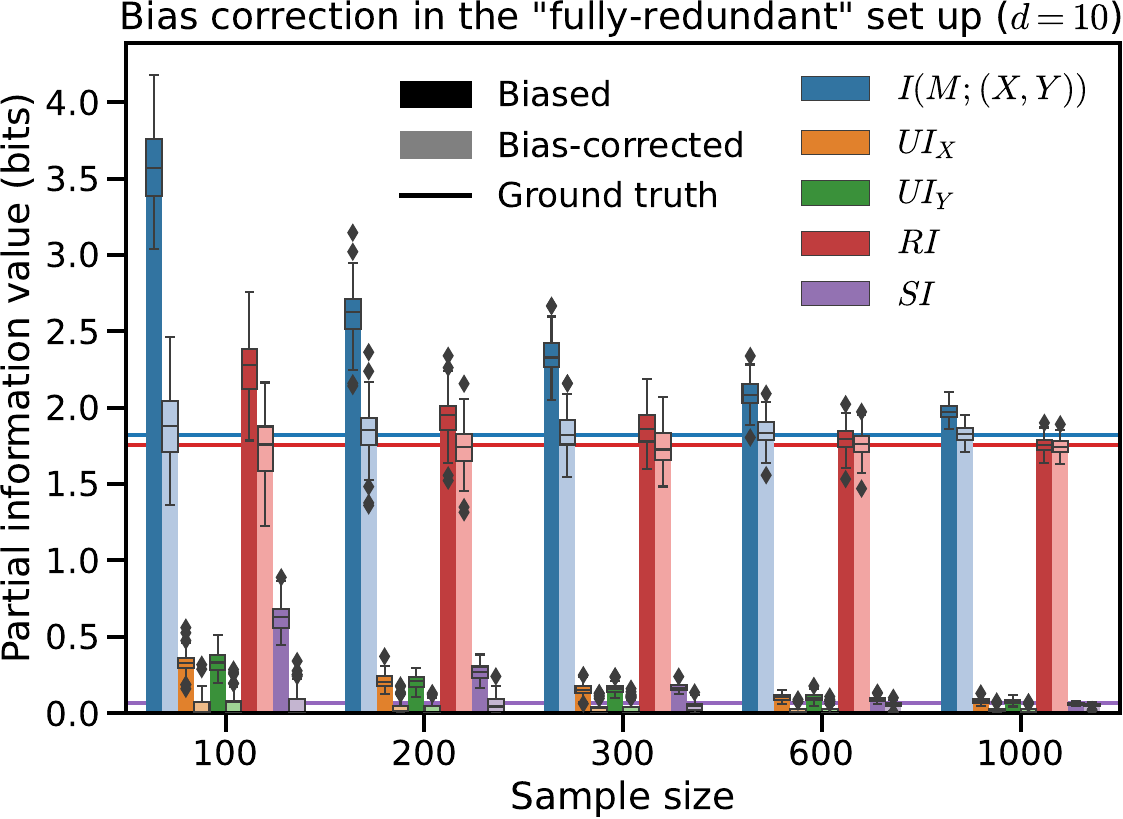}\\[12pt]
	\includegraphics[width=0.49\linewidth]{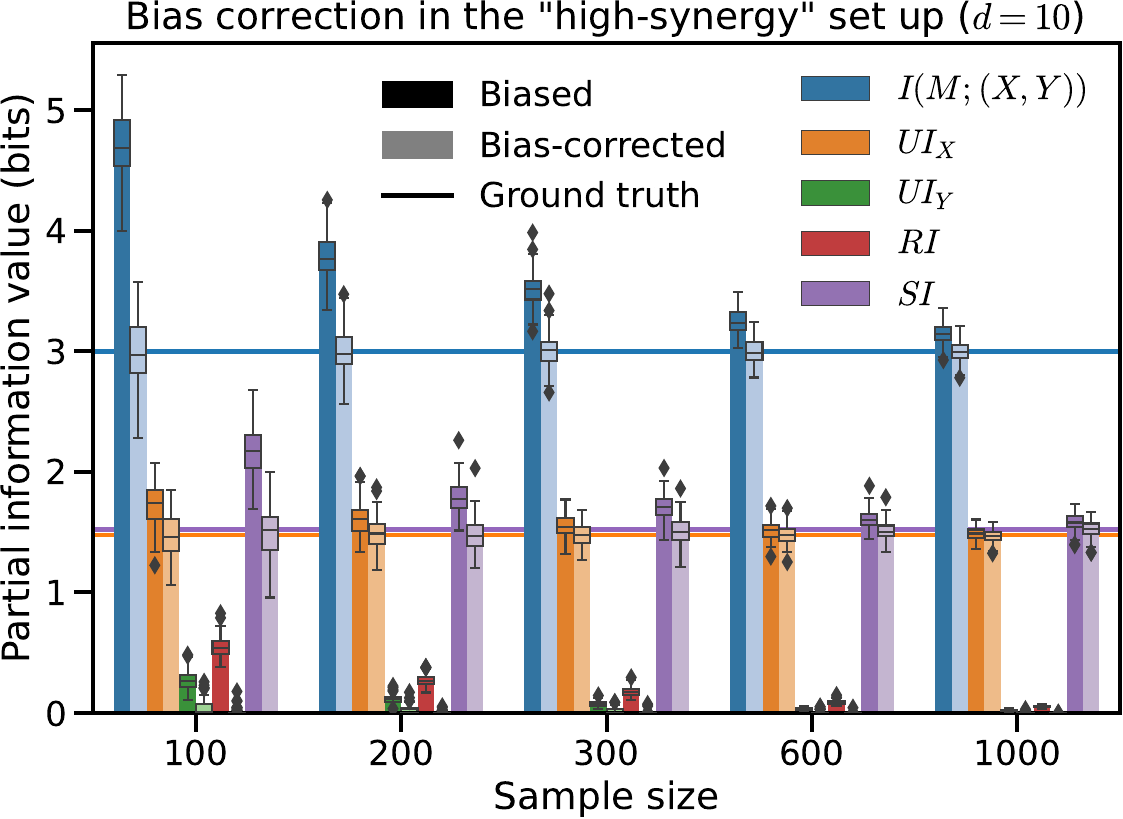}\hfill%
	\includegraphics[width=0.49\linewidth]{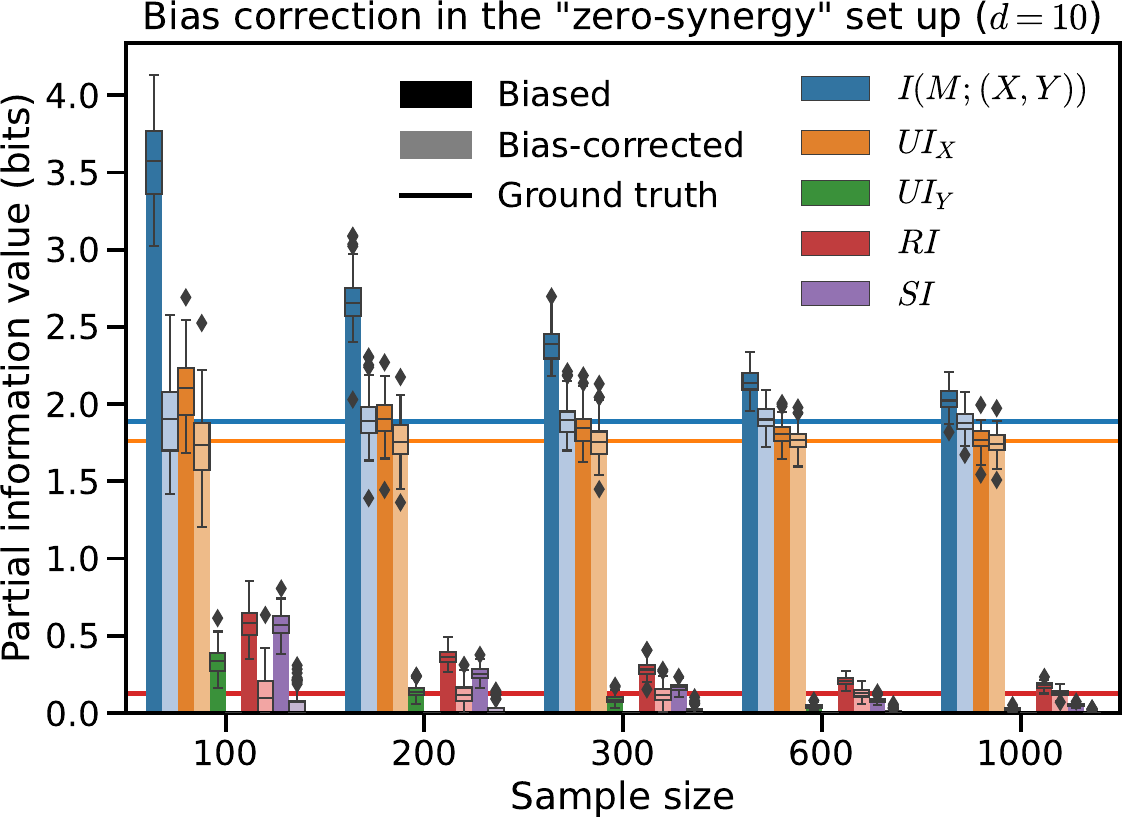}\\[12pt]
	\includegraphics[width=0.49\linewidth]{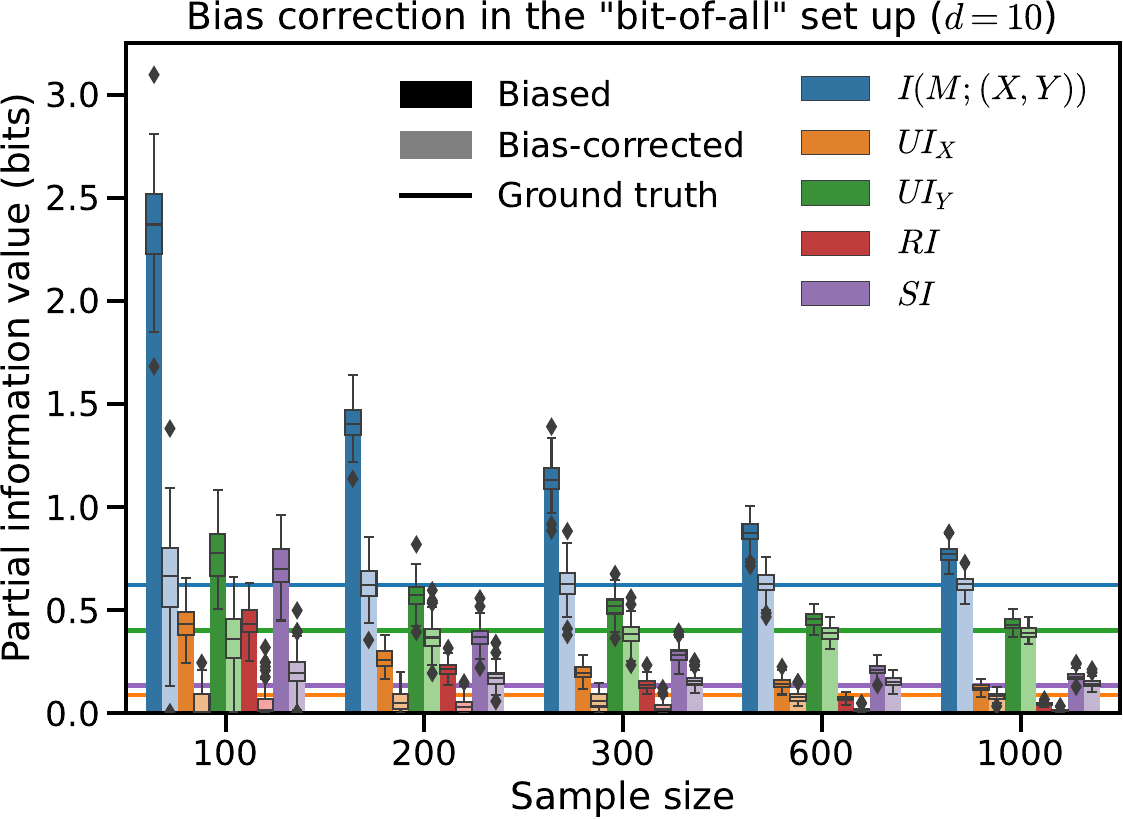}
	\caption{Bias correction for various setups described in Section~\ref{sec_bias_corr_desc} with $d \coloneqq d_M = d_X = d_Y = 10$.}
	\label{fig_bias_corr_extra_d10}
\end{figure}

\begin{figure}[p]
	\centering
	\includegraphics[width=0.49\linewidth]{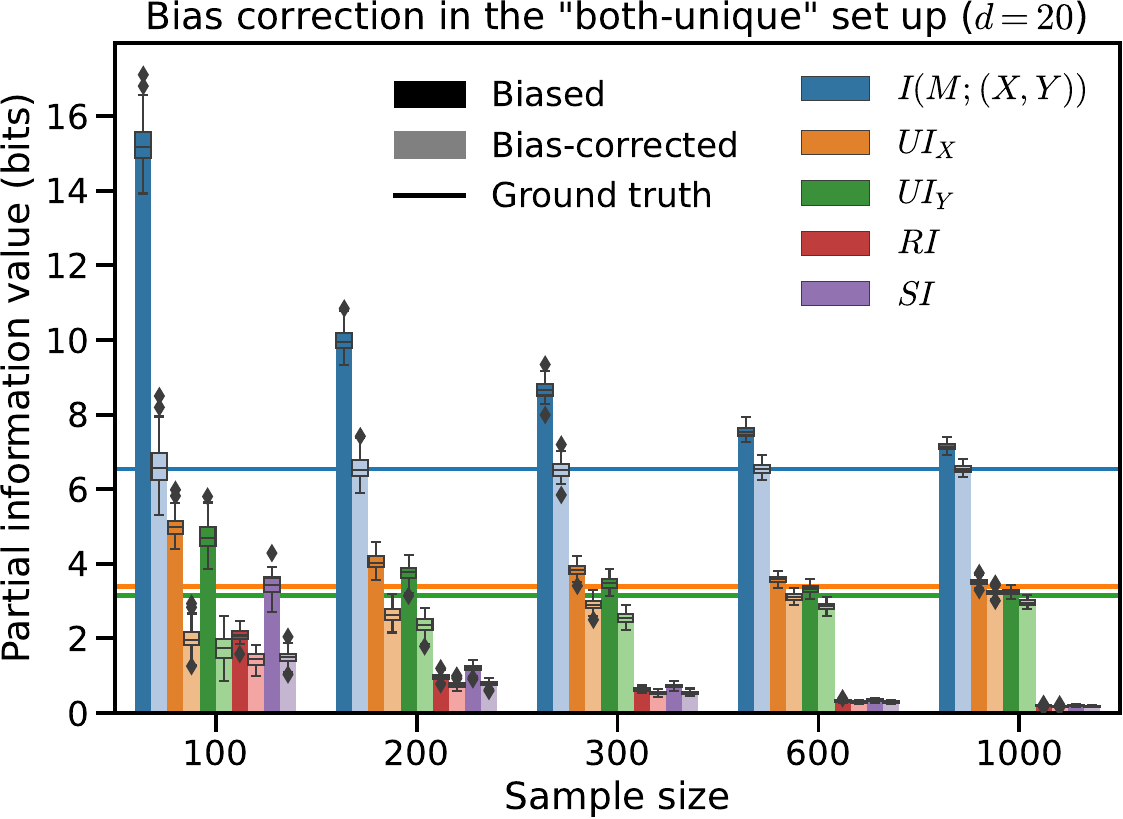}\hfill%
	\includegraphics[width=0.49\linewidth]{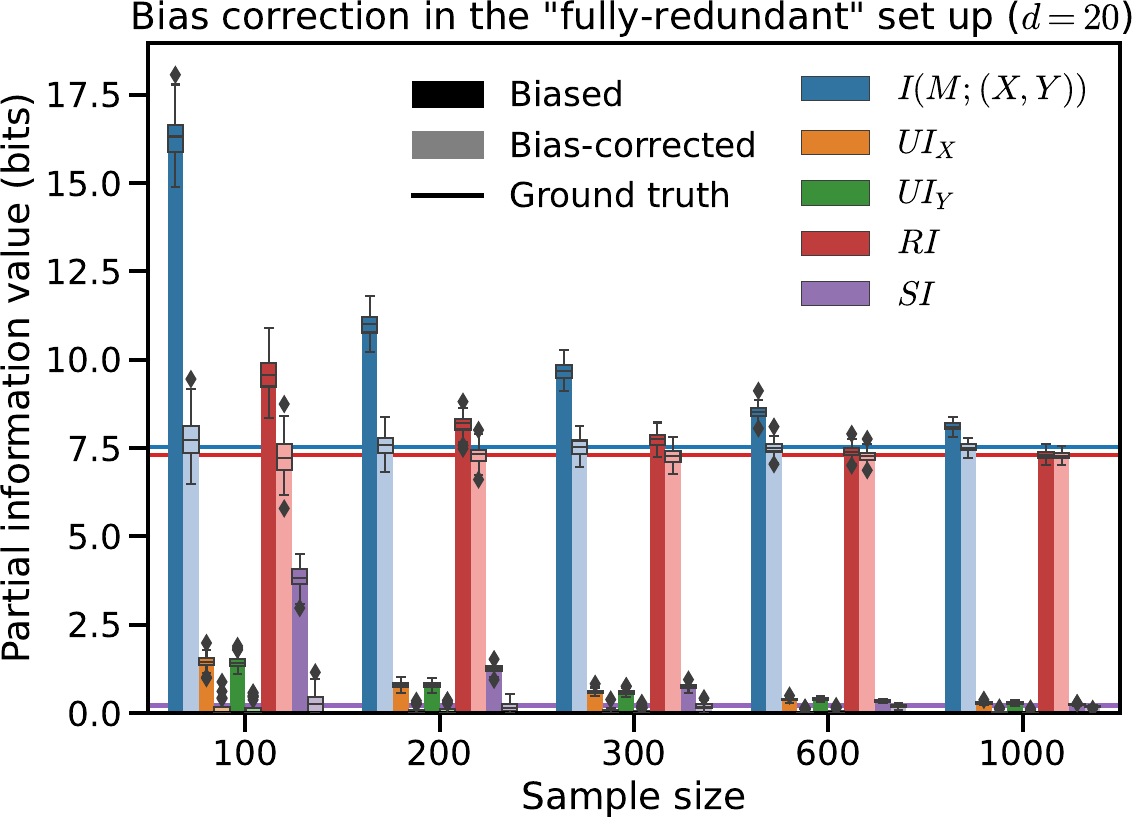}\\[12pt]
	\includegraphics[width=0.49\linewidth]{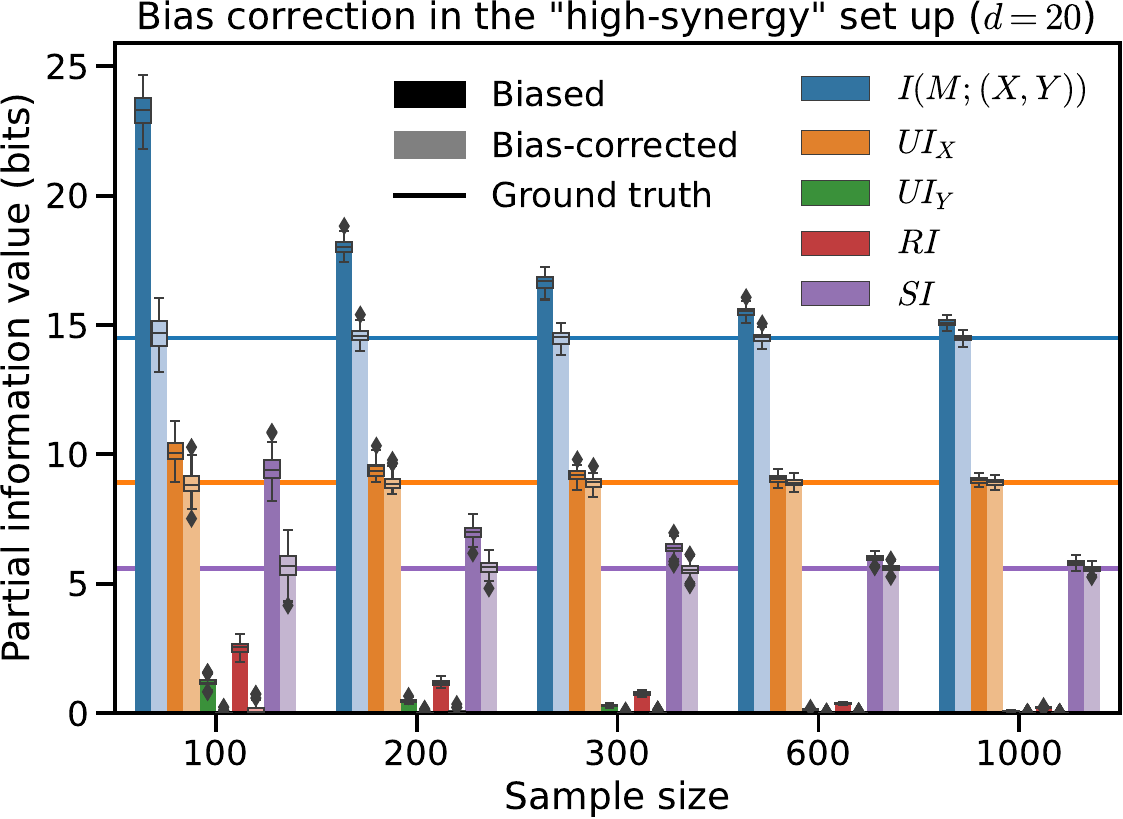}\hfill%
	\includegraphics[width=0.49\linewidth]{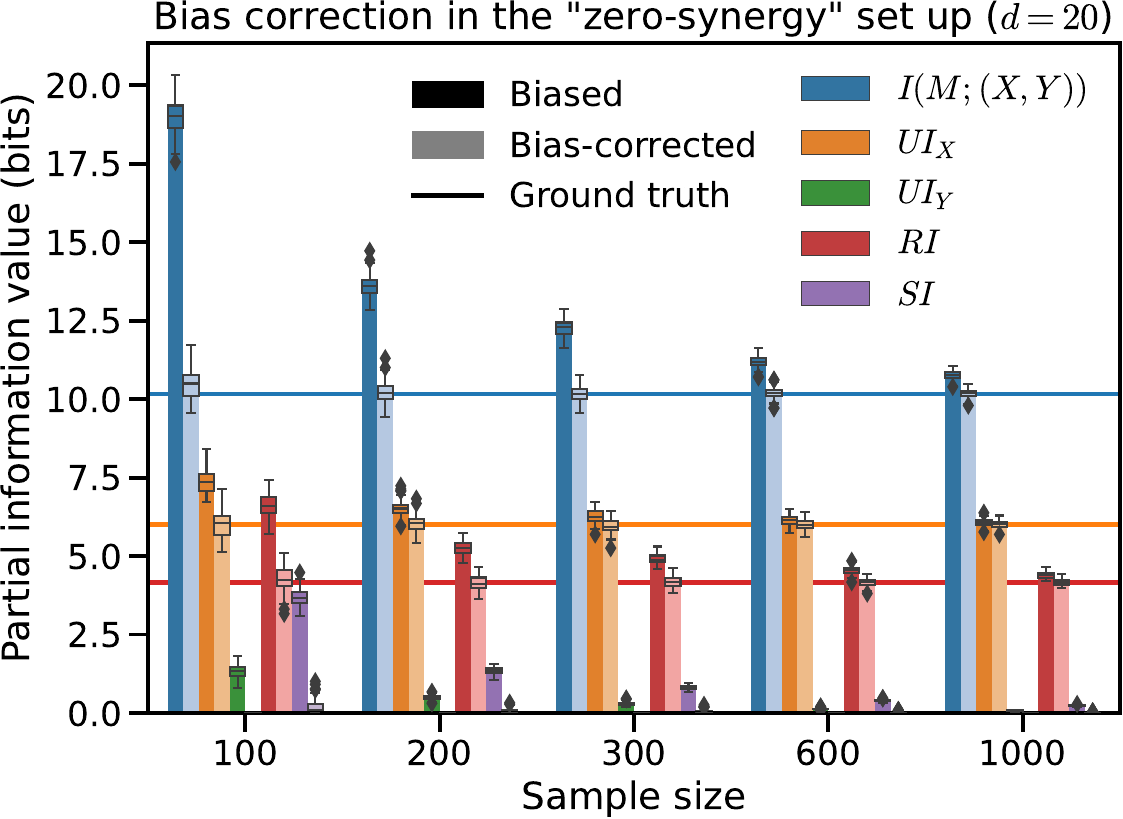}\\[12pt]
	\includegraphics[width=0.49\linewidth]{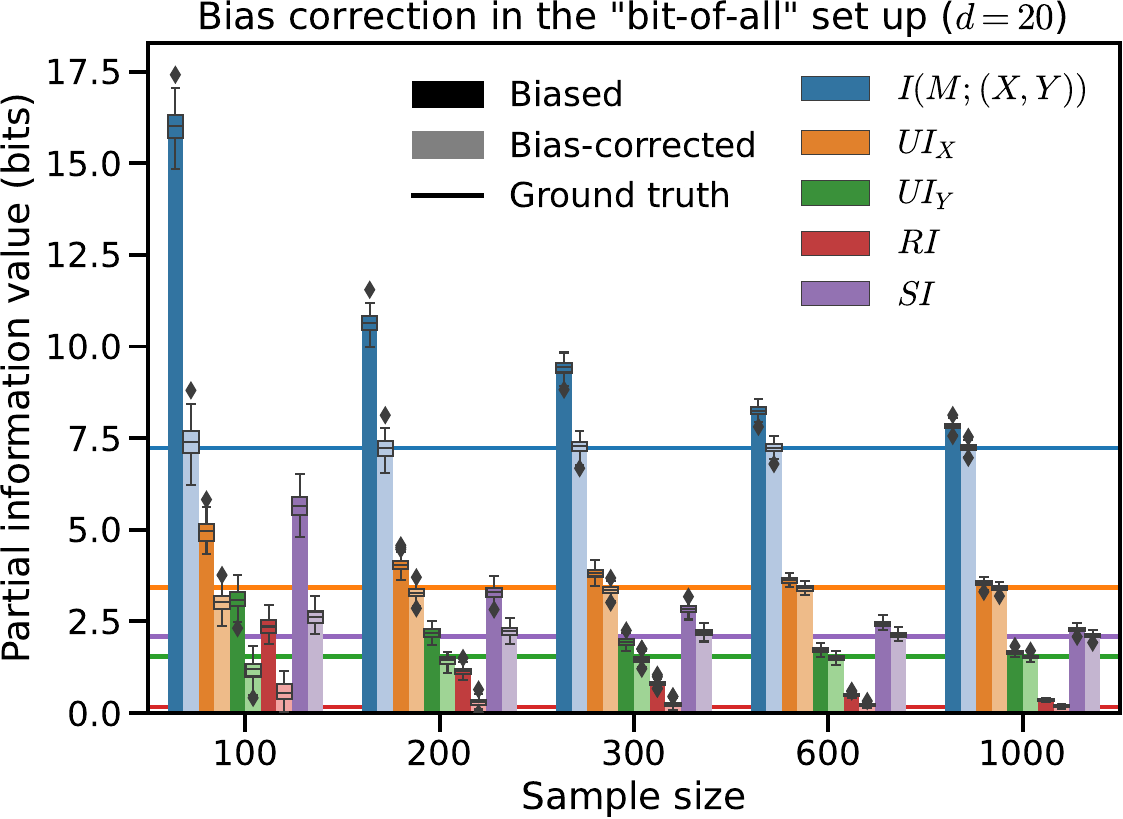}
	\caption{Bias correction for various setups described in Section~\ref{sec_bias_corr_desc} with $d \coloneqq d_M = d_X = d_Y = 20$.}
	\label{fig_bias_corr_extra_d20}
\end{figure}

Plots showing bias correction performance for all setups described in Section~\ref{sec_bias_corr_desc} are shown in Fig.~\ref{fig_bias_corr_extra_d10} for 10-dimensional, and in Fig.~\ref{fig_bias_corr_extra_d20} for 20-dimensional $M$, $X$ and $Y$.

Of the setups we examine, only the case with both $X$ and $Y$ having purely unique information appears to have somewhat poor performance, where our bias-correction method appears to over-correct the bias in unique information, while insufficiently correcting the bias in redundancy and synergy.

\subsection{A Preliminary Analysis of the Variance of PID Estimates}

\begin{figure}[p]
	\centering
	\includegraphics[width=0.49\linewidth]{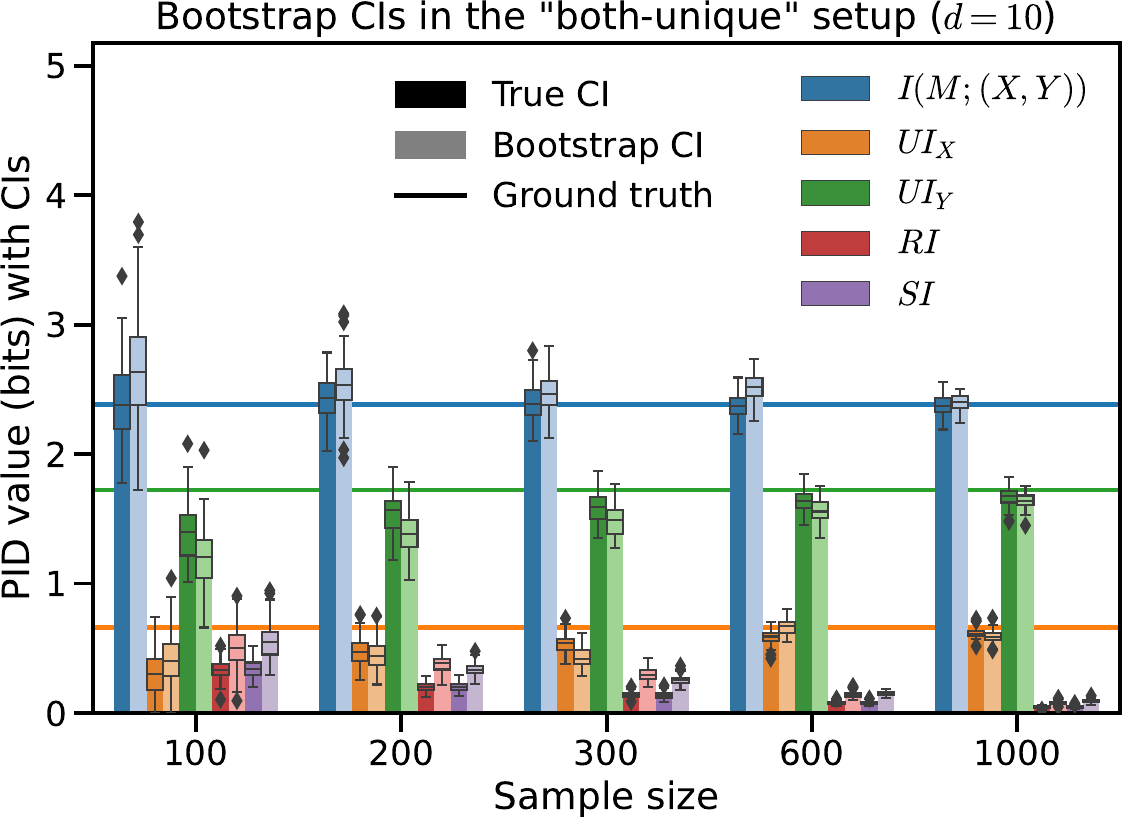}\hfill%
	\includegraphics[width=0.49\linewidth]{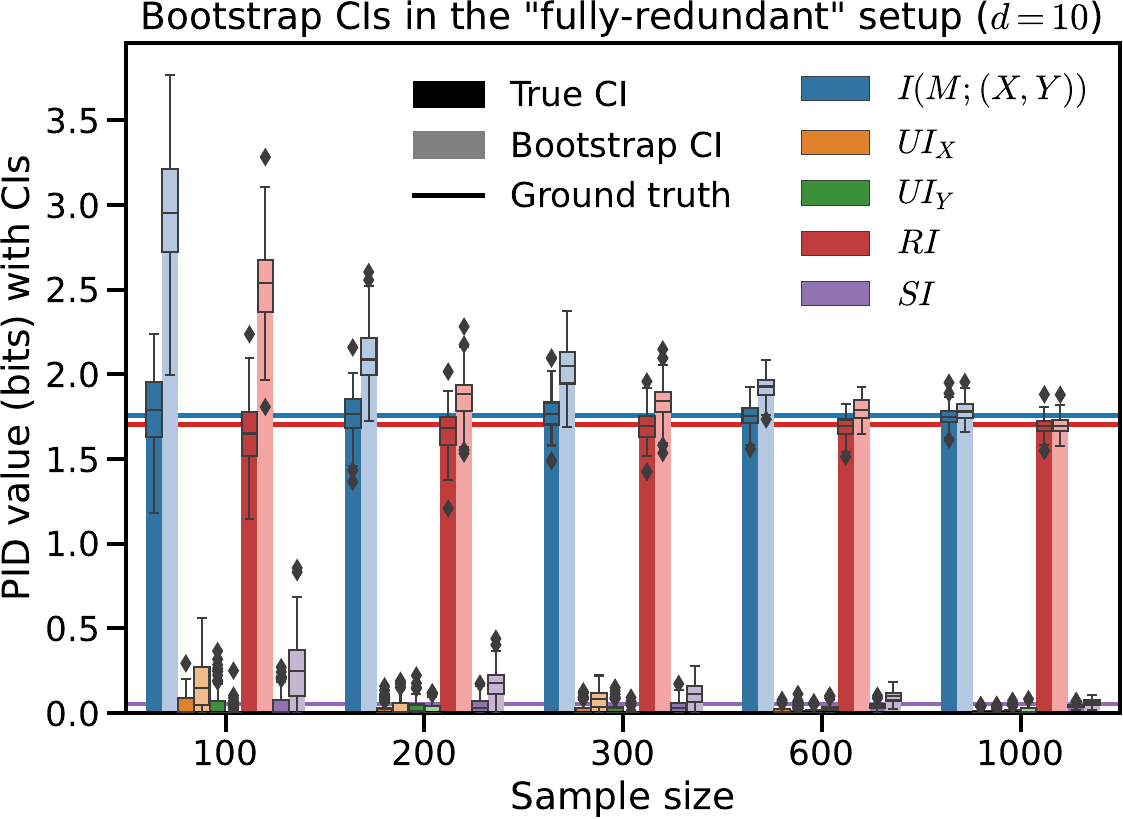}\\[12pt]
	\includegraphics[width=0.49\linewidth]{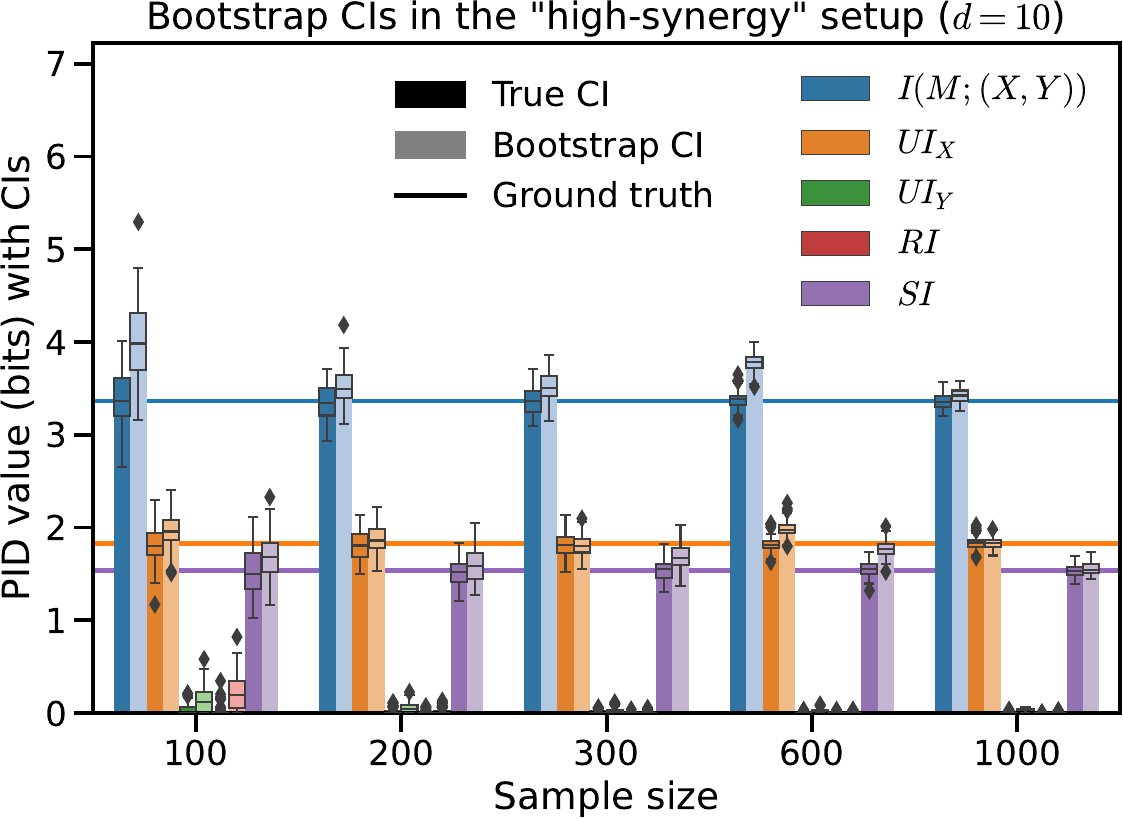}\hfill%
	\includegraphics[width=0.49\linewidth]{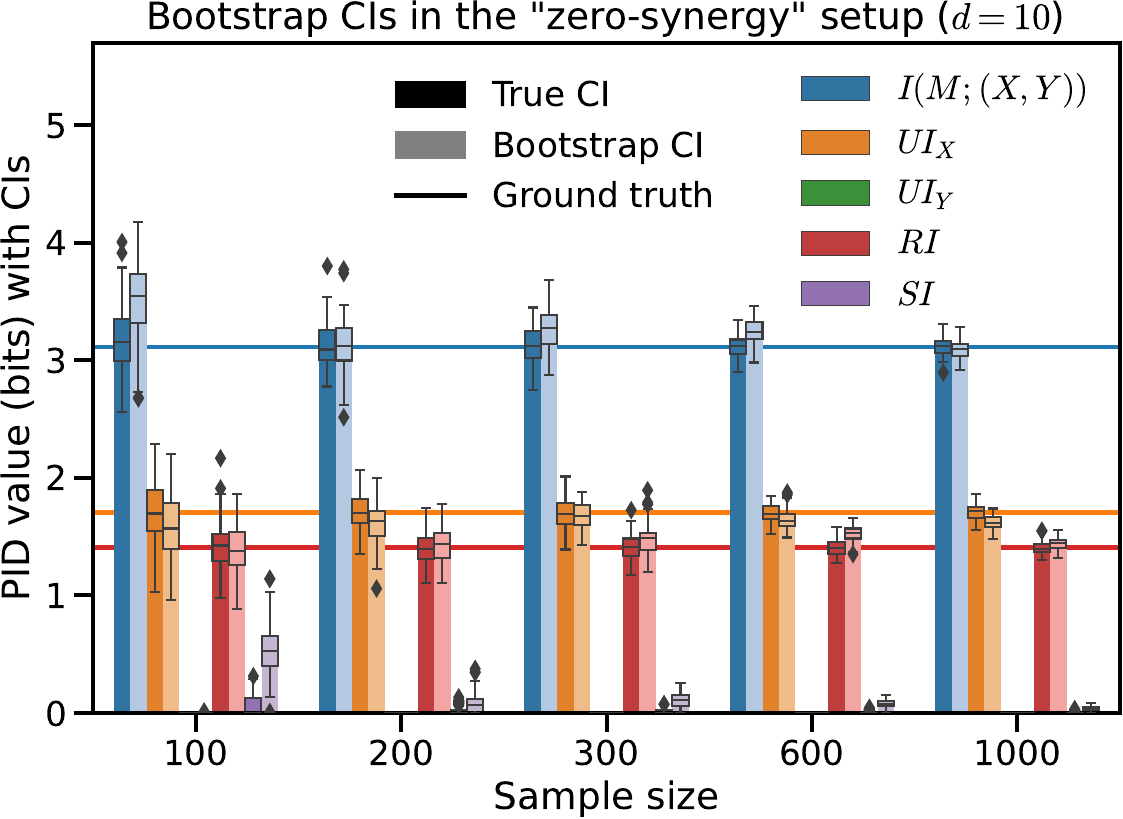}\\[12pt]
	\includegraphics[width=0.49\linewidth]{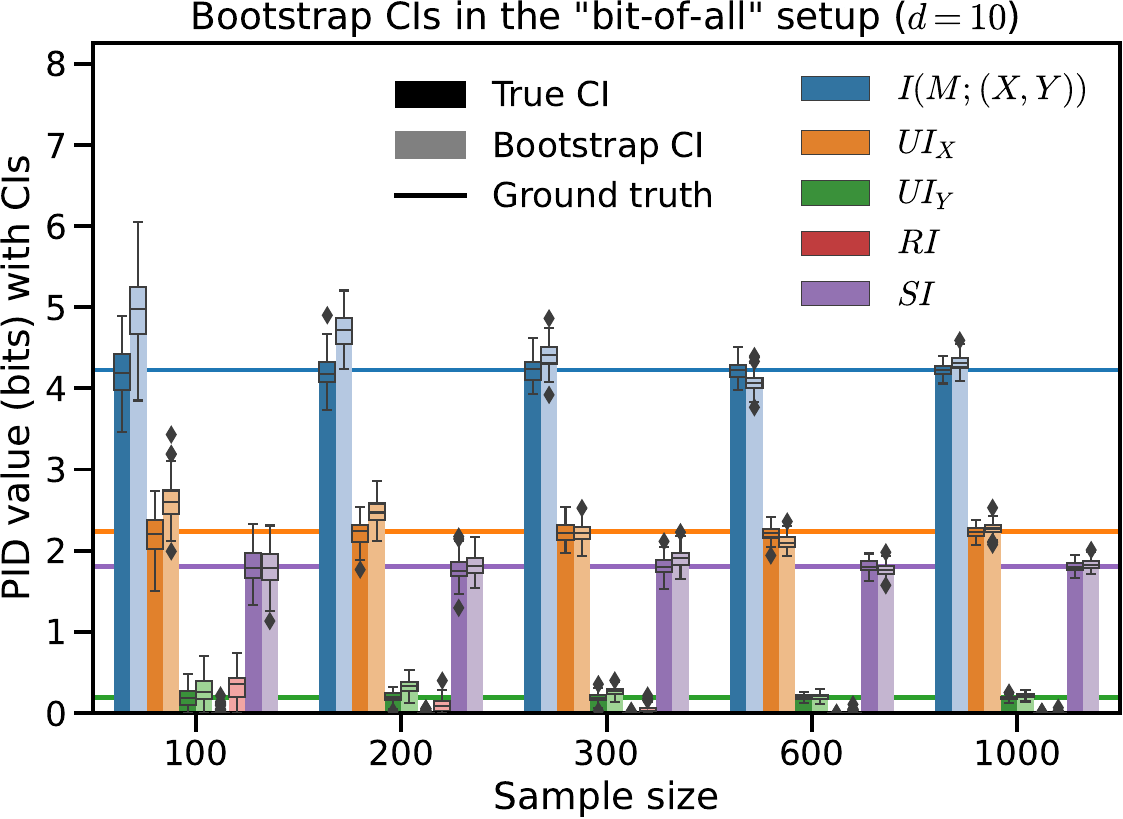}
	\caption{Bootstrap ``confidence intervals'' for various setups described in Section~\ref{sec_bias_corr_desc} with $d \coloneqq d_M = d_X = d_Y = 10$.
	Note that these are not true confidence intervals, but box-plot representations of the true variance of the estimator (over 100 runs) and the bootstrap estimate of the estimate's variance (from 100 bootstrap resamplings of a single random sample).}
	\label{fig_bootstrap_ci_d10}
\end{figure}

\begin{figure}[p]
	\centering
	\includegraphics[width=0.49\linewidth]{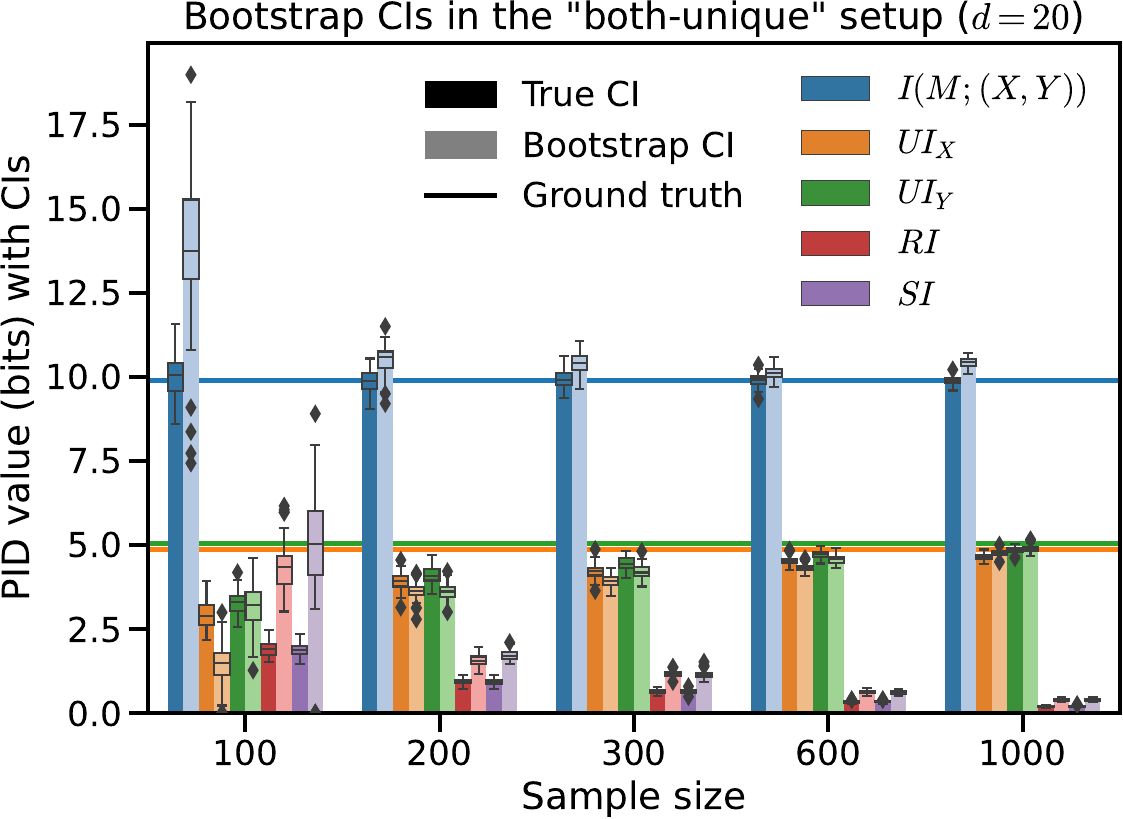}\hfill%
	\includegraphics[width=0.49\linewidth]{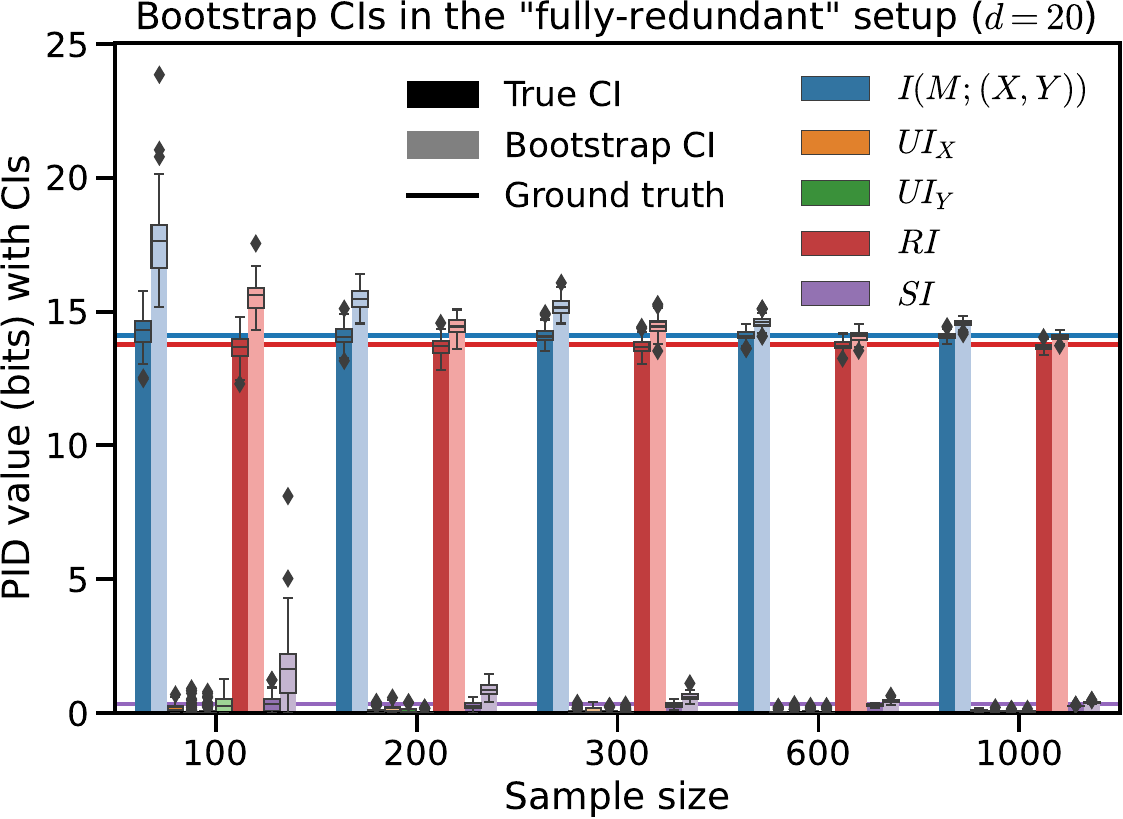}\\[12pt]
	\includegraphics[width=0.49\linewidth]{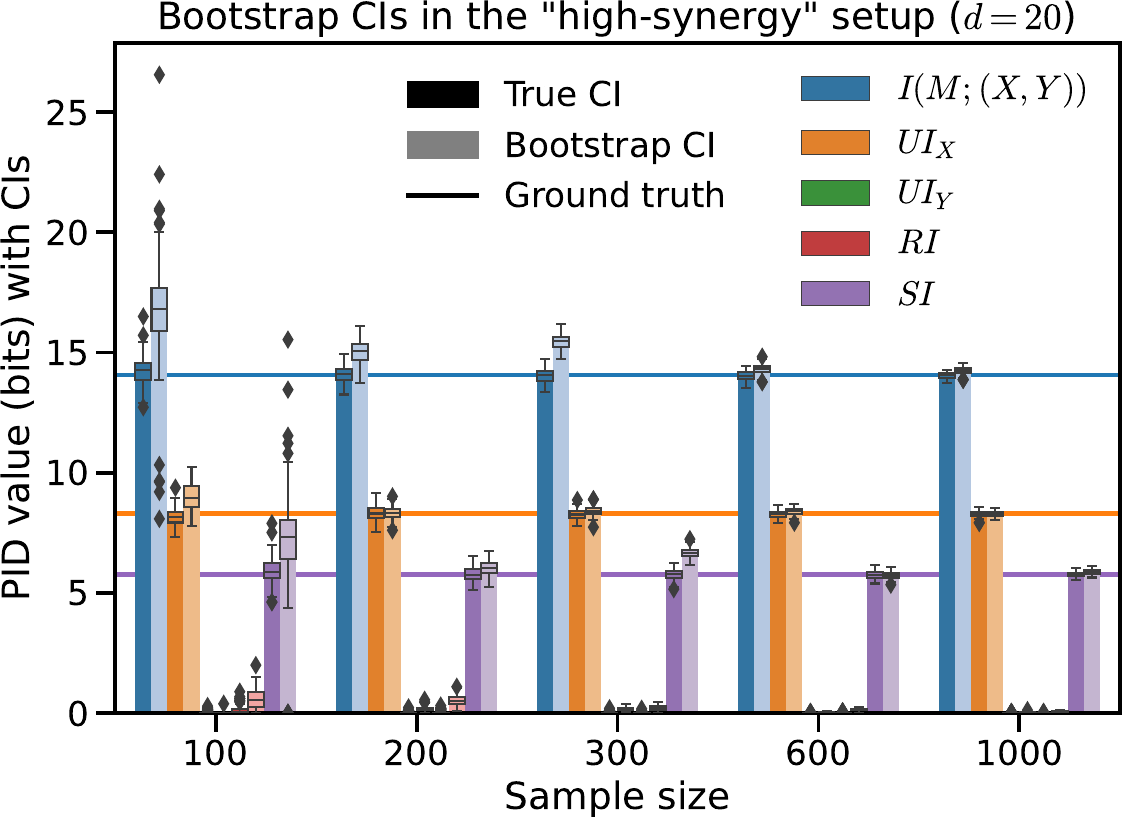}\hfill%
	\includegraphics[width=0.49\linewidth]{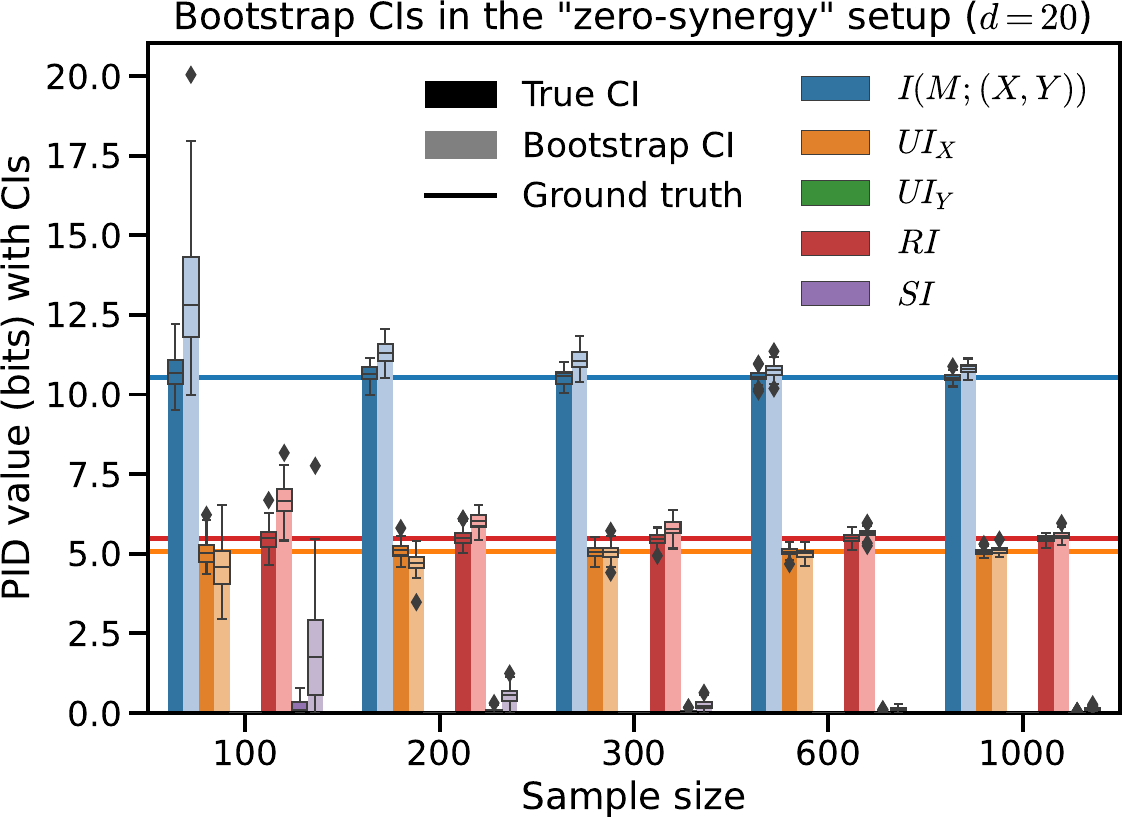}\\[12pt]
	\includegraphics[width=0.49\linewidth]{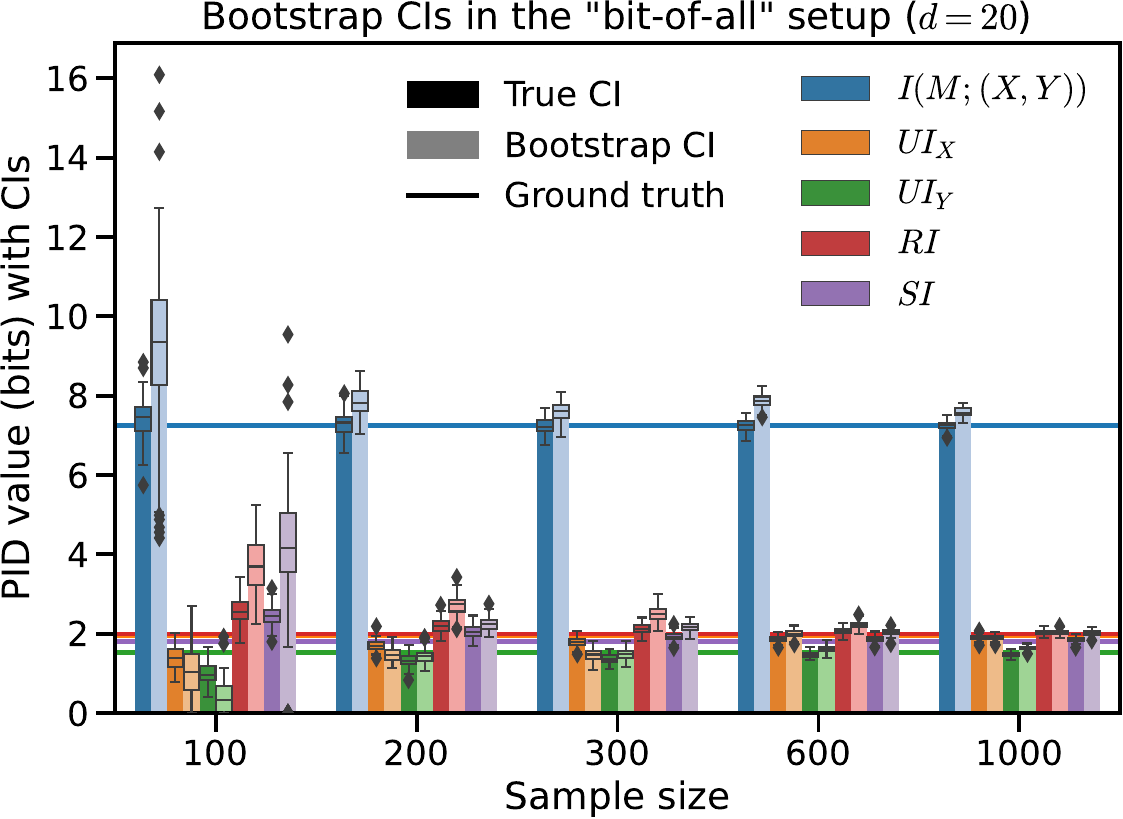}
	\caption{Bootstrap ``confidence intervals'' for various setups described in Section~\ref{sec_bias_corr_desc} with $d \coloneqq d_M = d_X = d_Y = 20$.
	Note that these are not true confidence intervals, but box-plot representations of the true variance of the estimator (over 100 runs) and the bootstrap estimate of the estimate's variance (from 100 bootstrap resamplings of a single random sample).}
	\label{fig_bootstrap_ci_d20}
\end{figure}

In Figures~\ref{fig_bootstrap_ci_d10} and \ref{fig_bootstrap_ci_d20}, we present a preliminary analysis of the variance of our PID estimates using bootstrap.
The figures represent the true distribution of the PID estimates over multiple sample draws, or over multiple bootstrap sample draws, in the form of box plots.
In what follows, we colloquially refer to these box plots as ``confidence intervals''.
The true ``confidence intervals'' were estimated using 100 runs of bias-corrected PID estimates, i.e., by drawing 100 different samples, each of size $n$.
The bootstrap ``confidence intervals'' were estimated using 100 bootstrap samples that were \emph{resampled} from a \emph{single} randomly drawn sample of size $n$.

When correcting for bias in the PID estimates on bootstrap samples, we use the number of \emph{unique} data points in each bootstrap sample in place of $n$ (refer Corollary~\ref{cor_imxy_bias}), rather than the total sample size.
This leads to more stable bootstrap-PID estimates.

The quality of the bootstrap ``confidence interval'' is affected greatly by the quality of the individual sample used for bootstrap resampling.
Nevertheless, we observe a reasonable degree of qualitative agreement between the true ``confidence interval'' and the bootstrap ``confidence interval'', particularly as the sample size increases.
Future work will assess confidence intervals with greater care, using well-defined metrics, and assess how well these confidence intervals are calibrated.


\section{Supplementary Material for Section~\ref{sec_neuroscience}}

\subsection{Details Regarding the Multivariate Poisson Spike-count Simulation}

We follow \ifarxiv{our previous paper}{} \cite{venkatesh2022partial}, where this analysis was first presented.
In this simulation, $M$ is two-dimensional, consisting of two independent and identically distributed Poisson random variables, $M_1$ and $M_2$.
$X$ and $Y$ are each generated through a linear combination of binomially thinning $M_1$ and $M_2$, along with some Poisson noise:
\begin{align}
	M_1, M_2 &\sim \mathrm{Poiss}(2) \\
	X &\sim \mathrm{Binom}(M_1, \alpha) + \mathrm{Binom}(M_2, 0.5) + \mathrm{Poiss}(1) \\
	Y &\sim \mathrm{Binom}(M_1, 0.5) + \mathrm{Binom}(M_2, 0.5) + \mathrm{Poiss}(1)
\end{align}

\subsection{Implementation Details of the Analysis Pipeline}

\begin{figure}[p]
	\centering
	\includegraphics[width=0.473\linewidth]{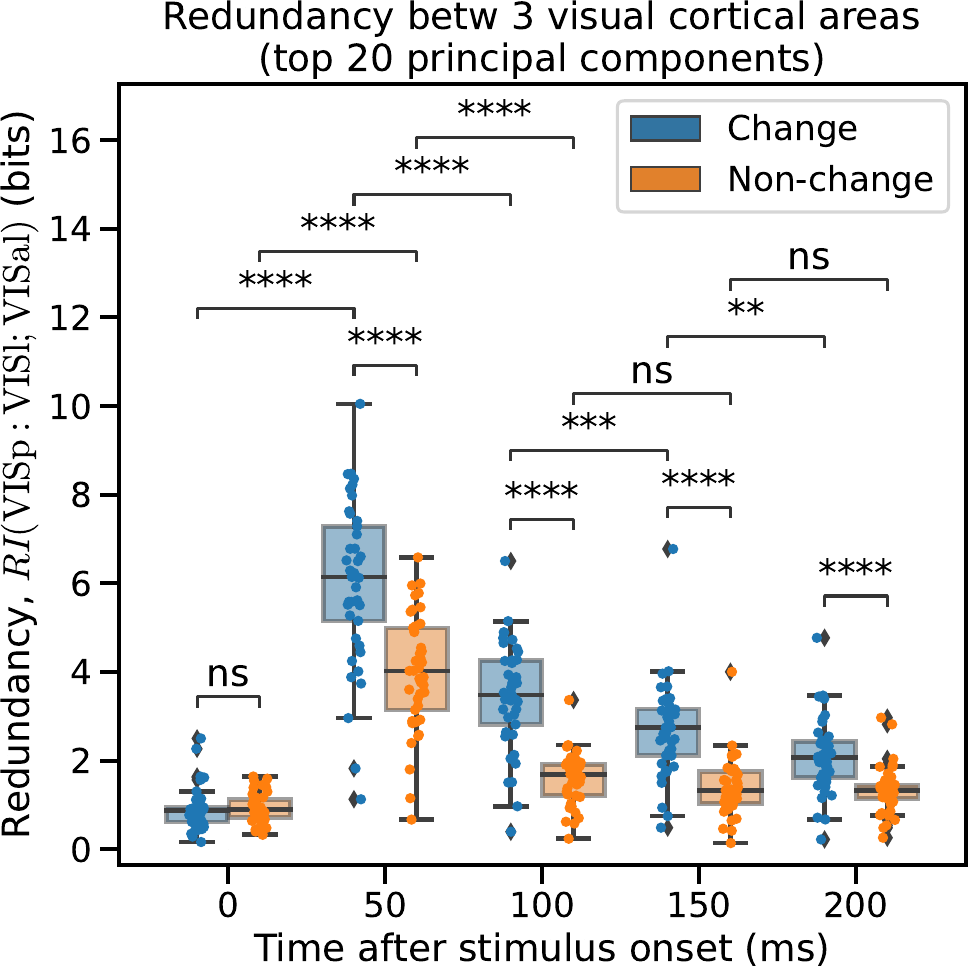}\hfill%
	\includegraphics[width=0.49\linewidth]{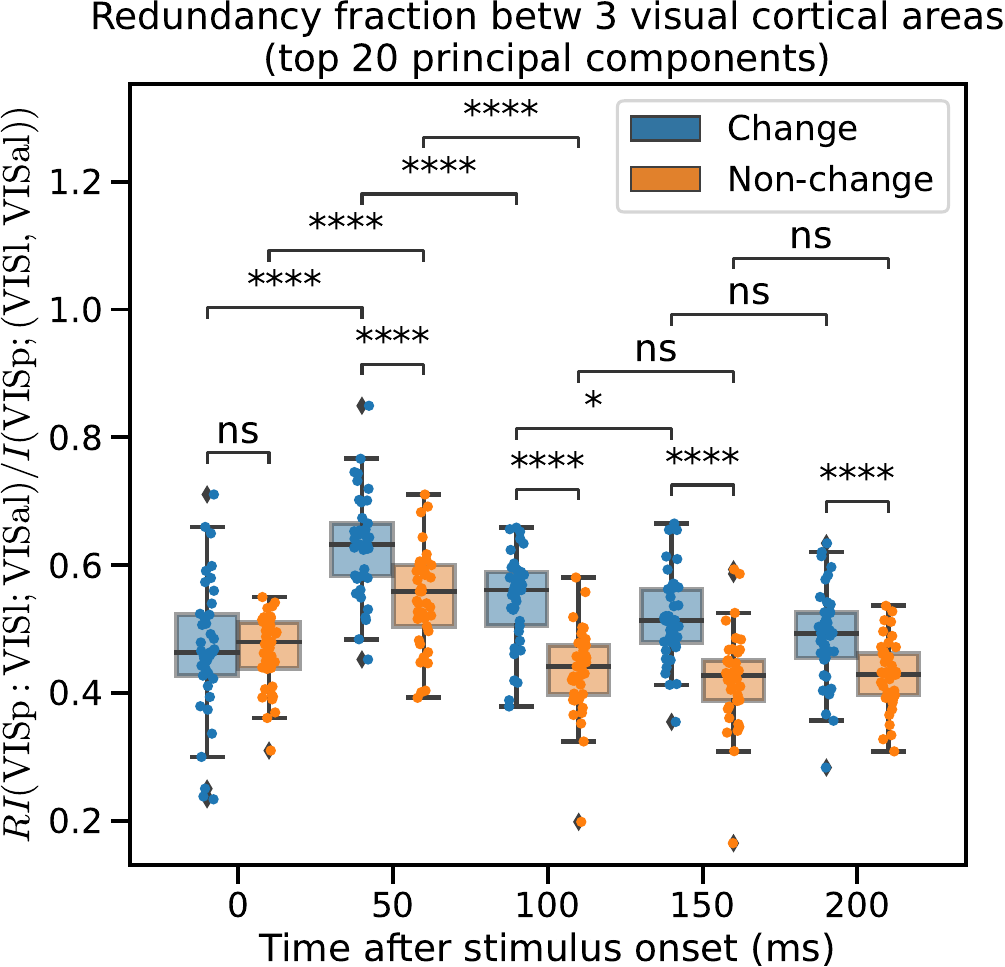}
	\caption{Redundancy about VISp activity between VISl and VISal, as a function of time, for flashes corresponding to an image change (blue) and flashes corresponding to a non-change (orange).
		Data points are across 42 mice.
		The plot on the left shows the raw redundancy in bits, while the plot on the right shows the redundancy normalized by the total mutual information.
	}
	\label{fig_vbn_ri_visal_20}
\end{figure}
\begin{figure}[p]
	\centering
	\includegraphics[width=0.473\linewidth]{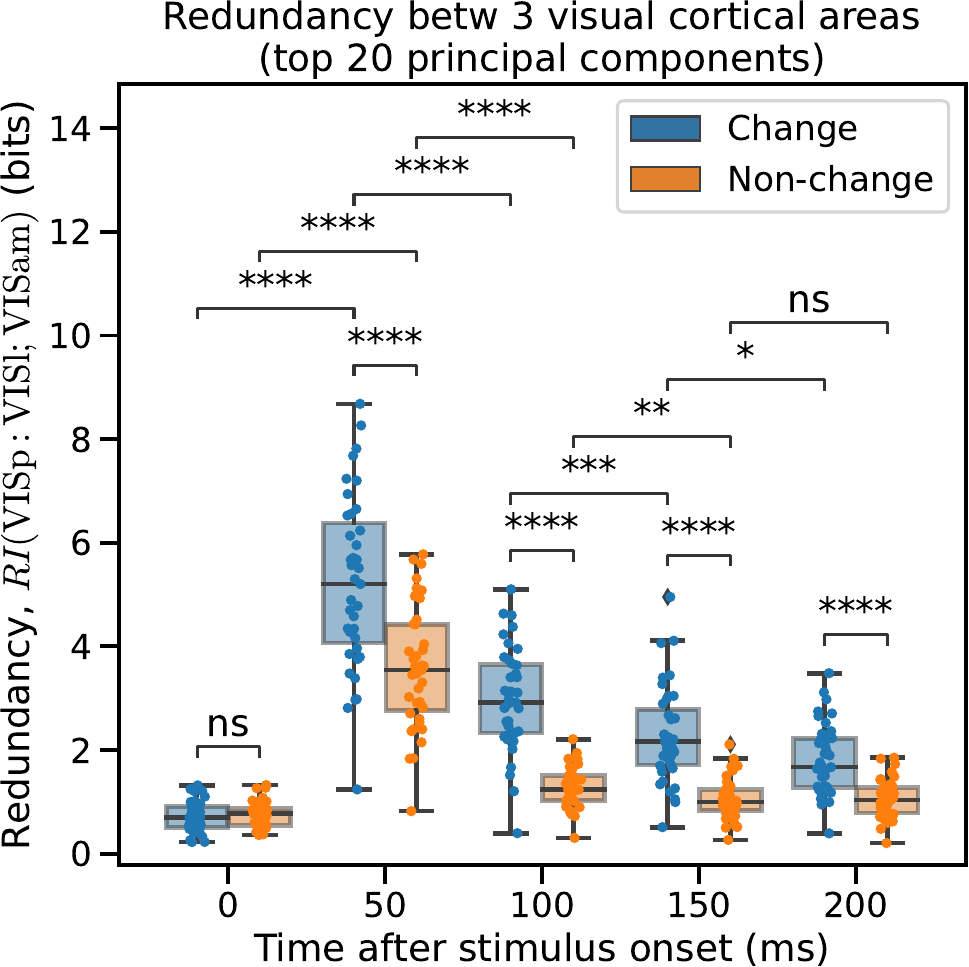}\hfill%
	\includegraphics[width=0.49\linewidth]{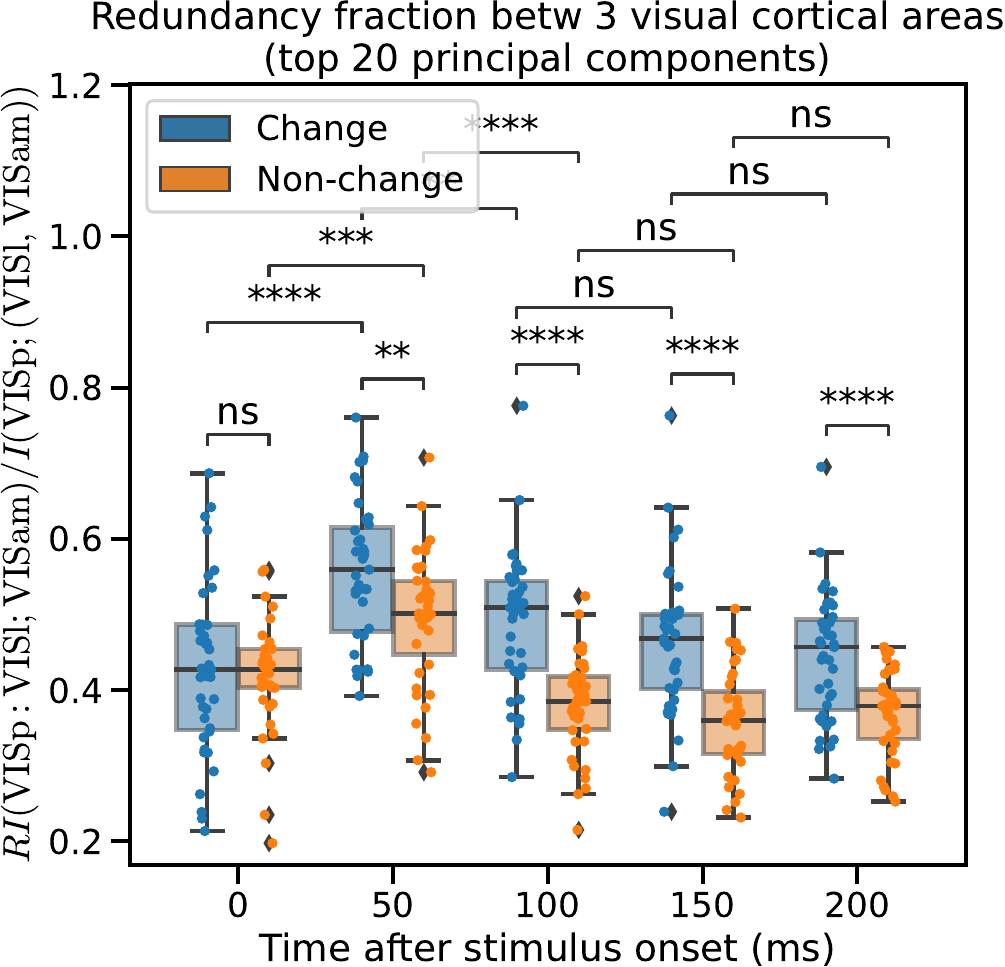}
	\caption{Redundancy about VISp activity between VISl and VISam, as a function of time, for flashes corresponding to an image change (blue) and flashes corresponding to a non-change (orange).
		Data points are across 40 mice.
		The plot on the left shows the raw redundancy in bits, while the plot on the right shows the redundancy normalized by the total mutual information.
	}
	\label{fig_vbn_ri_visam_20}
\end{figure}
\begin{figure}[p]
	\centering
	\includegraphics[width=\linewidth]{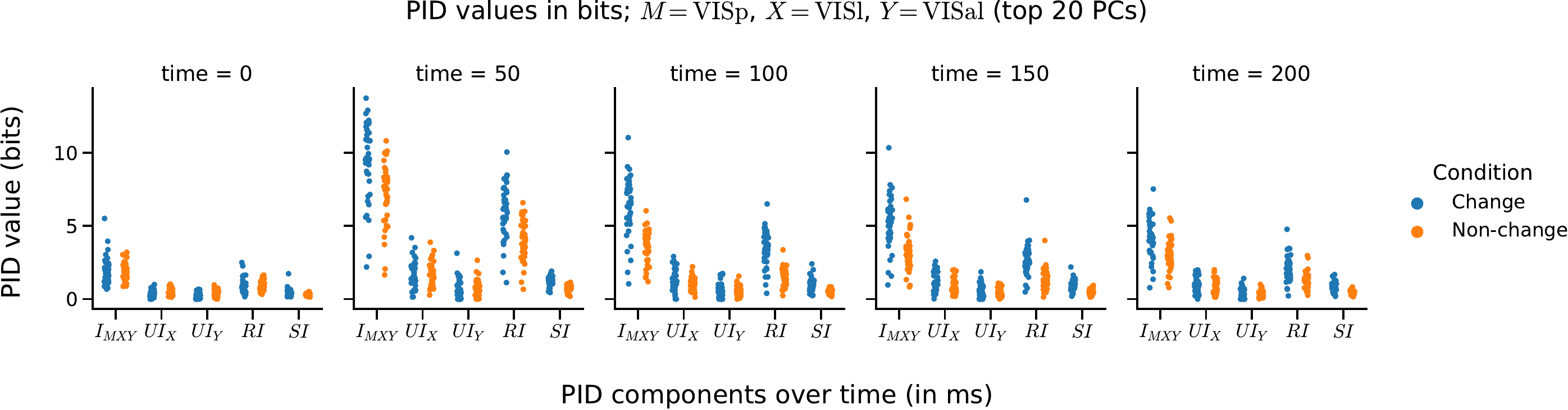}\\[12pt]
	\includegraphics[width=\linewidth]{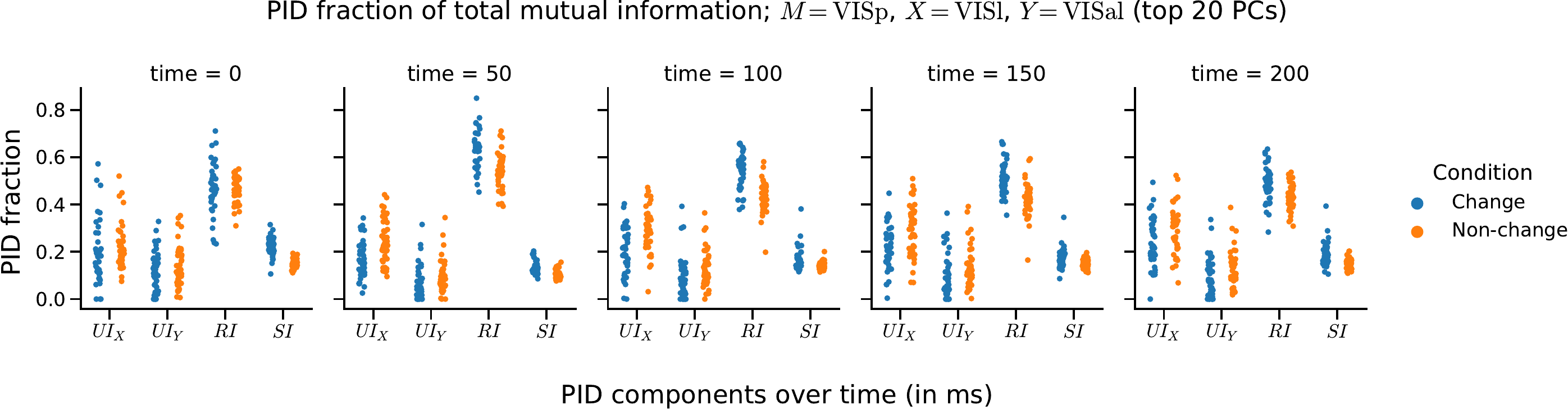}
	\caption{All PID components---about VISp activity, between VISl and VISal---in bits (top), and as a fraction of total mutual information (bottom), at various times after stimulus onset, for change and non-change flashes.
		Here, the label $I_{MXY}$ in the x-axis of the top plot refers to $I(M ; (X, Y))$.
		These plots show that redundancy is the primary driver of mutual information.
	}
	\label{fig_vbn_all_visal_20}
\end{figure}
\begin{figure}[p]
	\centering
	\includegraphics[width=\linewidth]{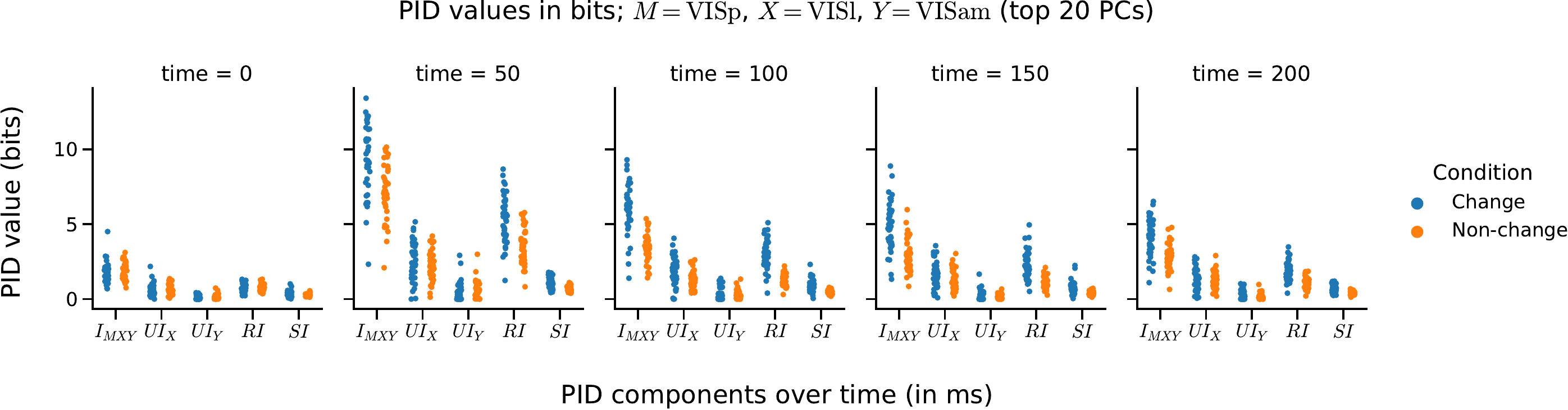}\\[12pt]
	\includegraphics[width=\linewidth]{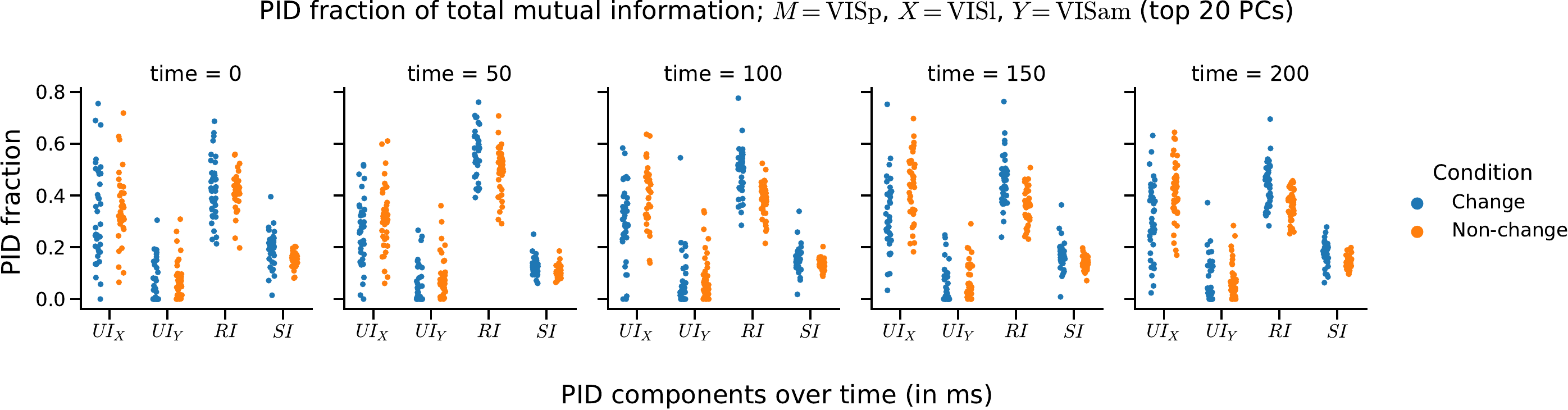}
	\caption{All PID components---about VISp activity, between VISl and VISam---in bits (top), and as a fraction of total mutual information (bottom), at various times after stimulus onset, for change and non-change flashes.
		Here, the label $I_{MXY}$ in the x-axis of the top plot refers to $I(M ; (X, Y))$.
		These plots show that redundancy is the primary driver of mutual information.
	}
	\label{fig_vbn_all_visam_20}
\end{figure}

The Visual Behavior Neuropixels data was analyzed as follows:
\begin{enumerate}[leftmargin=2em, topsep=0pt]
	\item We selected mice that had at least 20 units in each brain region of interest.
		Only mice with both familiar and novel sessions were selected.
	\item From each region we selected units of `good' quality, with SNR at least 1, and with fewer than 1 inter-spike interval violations.
	\item Trials were aligned to the start of each stimulus flash, and spikes were counted in bins of 50 ms, between 0 and 250 ms after stimulus onset (0-50 ms, 50-100 ms, etc.).
	\item Trials corresponding to a non-change flash were defined as those that occurred between 4 and 10 flashes after the start of a behavioral trial, such that the image remained the same as the original image in this behavioral trial.
		Flashes corresponding to an omission, flashes after an omission, and flashes during which the animal licked, were all removed.
		Only flashes that occurred while the animal was engaged (as measured by an average reward rate of at least 2 rewards/min) were selected.
	\item Trials corresponding to a change flash were defined as those during which an image change occurred, and the animal was engaged (as above).
	\item The top 10 or 20 principal components of neural activity were selected at each time bin, and for each brain region under consideration.
		Principal component analysis was carried out using the Scikit-learn~\citesm{pedrogosa2011scikitlearn} package in Python.
	\item $\sim_G$-PID estimates were computed on the covariance matrix between principal components across regions, for each time bin.
	\item Data were aggregated across 42 mice for the figures with VISal, and over 40 mice for the figures with VISam.
	\item Statistical significance was assessed using a two-sided unpaired Mann-Whitney-Wilcoxon test.
\end{enumerate}

\subsection{Additional Results}

Figures~\ref{fig_vbn_ri_visal_20} and \ref{fig_vbn_ri_visam_20} show results for the redundancy between three visual cortical areas over time, using the top-20 principal components in each region (rather than the top-10, as used in Fig.~\ref{fig_vbn_ri_visal}).
Fig.~\ref{fig_vbn_ri_visal_20} shows redundancy about VISp activity, between VISl and VISal.
Fig.~\ref{fig_vbn_ri_visam_20} shows redundancy about VISp activity, between VISl and a different higher-order cortical region, VISam (see, e.g.,~\citesm{siegle2021survey_sm}).

These figures show an even greater, and more sustained redundancy (as well as redundant fraction of information) about VISp activity, between VISl and the higher-order cortical region (either VISal or VISam), when the stimulus shown is a behaviorally relevant target.

Figures~\ref{fig_vbn_all_visal_20} and \ref{fig_vbn_all_visam_20} show all PID components, not just redundancy, for the same settings as in Figures~\ref{fig_vbn_ri_visal_20} and \ref{fig_vbn_ri_visam_20} respectively.
These show that redundancy is the dominant partial information component, and appears to be the main driver of changes in the overall mutual information.
This justifies why we include only a plot of redundancy in Figs.~\ref{fig_vbn_ri_visal}, \ref{fig_vbn_ri_visal_20} and \ref{fig_vbn_ri_visam_20}.

\subsection{Differences between Change and Non-change Conditions are not an Artifact of Bias-correction}

\begin{figure}[t]
	\centering
	\includegraphics[width=0.473\linewidth]{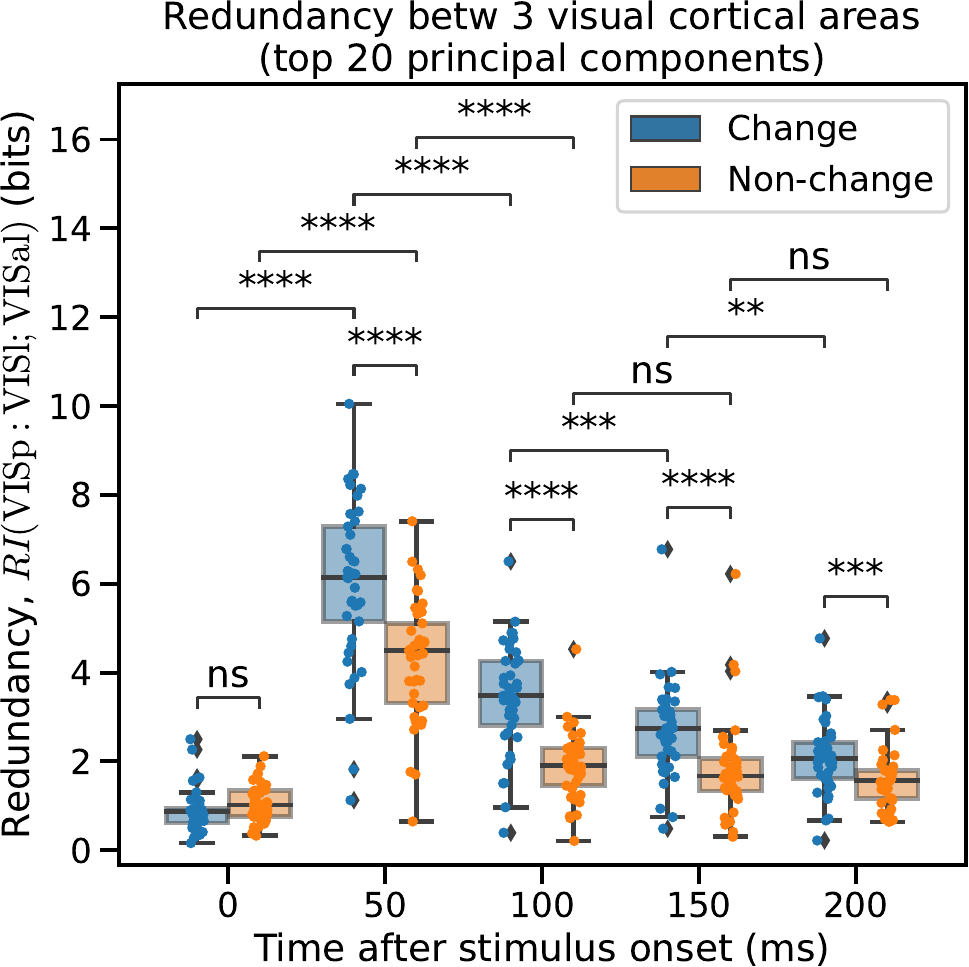}\hfill%
	\includegraphics[width=0.49\linewidth]{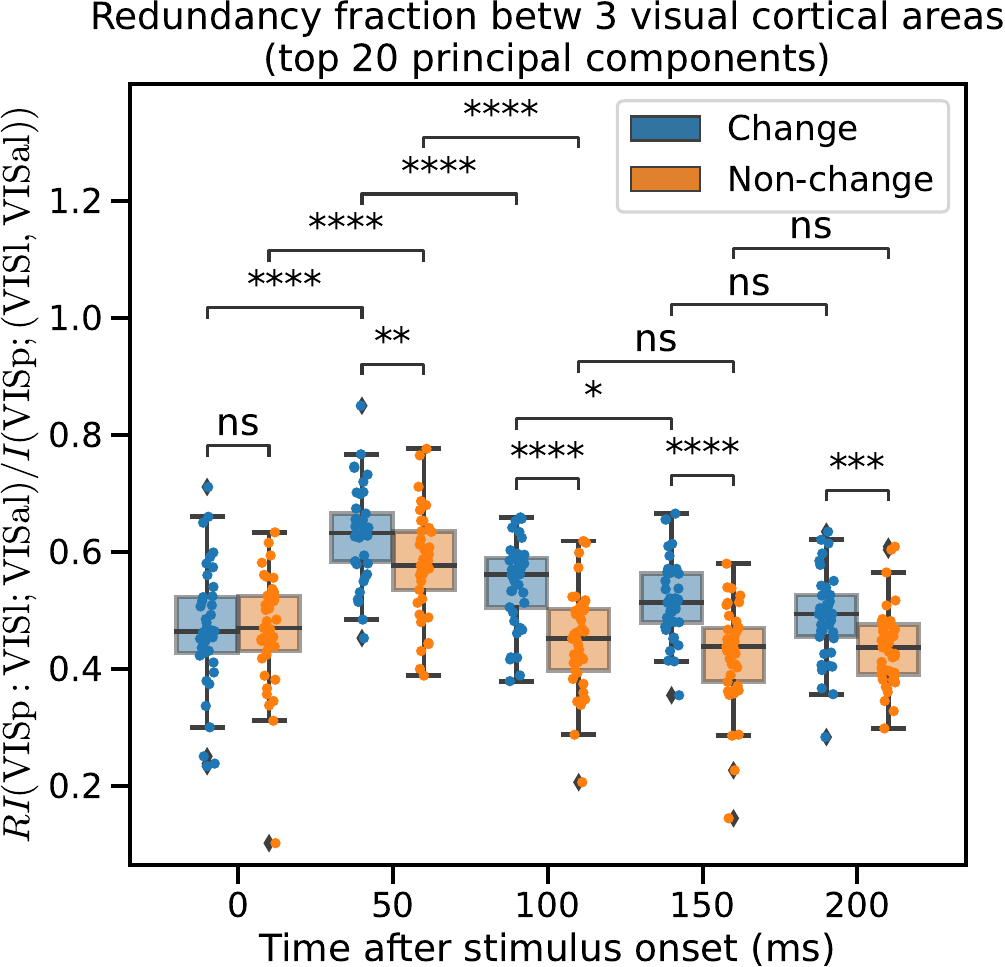}
	\caption{Redundancy about VISp activity between VISl and VISal, as a function of time, for flashes corresponding to an image change (blue) and flashes corresponding to a non-change (orange) with an equal number of samples.
		Data points are across 42 mice.
		The plot on the left shows the raw redundancy in bits, while the plot on the right shows the redundancy normalized by the total mutual information.
		The observations made in the other figures continue to hold; the differences seen between the two conditions are, therefore, not a result of differences in sample size.
	}
	\label{fig_vbn_ri_eqsamp_visal_20}
\end{figure}

The number of trials corresponding to change flashes is much smaller than the number of trials corresponding to non-change flashes.
Accordingly, the sample size used to estimate the covariance matrix is different in each of the two conditions.
Bias correction was performed using the appropriate sample size; however, as noted in Section~\ref{sec_estimation}, our bias correction process is not perfect, and may leave some residual bias.

In order to show that the results we observed were not an artifact of differences in residual bias caused by different sample sizes, we randomly subsampled the non-change flashes to produce a dataset with equal numbers of trials for change and non-change flashes.
Repeating the analysis as before, we found that our conclusions continued to hold even in the setting where both conditions have equal sample sizes, as shown in Figure~\ref{fig_vbn_ri_eqsamp_visal_20}.

\section{Compute Configuration Used and Code Availability}

All analyses were performed on a workstation equipped with an Intel Core i7-10700KF CPU with 8~cores (16~threads), 48 GiB of RAM and data stored on a 1 TB PCIe NVMe solid state drive.

Analysis of the Visual Behavior Neuropixels data (for 84 sessions) with $d = 20$ took approximately 9 minutes to run.
This included loading data for each session and computing 840 PID values on $60 \times 60$ covariance matrices, implying an average run-time of about 0.64s for each PID estimate (including amortized data-load time).

All code used to compute and estimate the $\sim_G$-PID and correct for bias, including all examples in this paper and code for neural data analysis, \ifarxiv{is}{will be made} available on Github\ifarxiv{~\citesm{code}}{}.

\bibliographysm{supplementary-refs}

\end{document}

%% file: pkg-todonotes.tex
\usepackage{marginfix}  
\usepackage{sidenotes}  
\let\todo\sidenote      
\makeatletter
\RenewDocumentCommand\sidenotetext{oo+m}{
	\IfNoValueOrEmptyTF{#1}{%
		\@sidenotes@placemarginal{#2}{\textsuperscript{\{\thesidenote\}}{}~\scriptsize{#3}}%
		\refstepcounter{sidenote}%
	}{%
		\@sidenotes@placemarginal{#2}{\textsuperscript{#1}~#3}%
	}%
}
\RenewDocumentCommand\sidenotemark{o}{
	\@sidenotes@multichecker%
	\IfNoValueOrEmptyTF{#1}{%
		\@sidenotes@thesidenotemark{\{\thesidenote\}}%
	}{%
		\@sidenotes@thesidenotemark{#1}%
	}%
	\@sidenotes@multimarker%
}
\makeatother

%% file: paper.bbl
\begin{thebibliography}{38}
\providecommand{\natexlab}[1]{#1}
\providecommand{\url}[1]{\texttt{#1}}
\expandafter\ifx\csname urlstyle\endcsname\relax
  \providecommand{\doi}[1]{doi: #1}\else
  \providecommand{\doi}{doi: \begingroup \urlstyle{rm}\Url}\fi

\bibitem[de~Vries et~al.(2020)de~Vries, Lecoq, Buice, Groblewski, Ocker,
  Oliver, Feng, Cain, Ledochowitsch, Millman, et~al.]{devries2020large}
Saskia~EJ de~Vries, Jerome~A Lecoq, Michael~A Buice, Peter~A Groblewski,
  Gabriel~K Ocker, Michael Oliver, David Feng, Nicholas Cain, Peter
  Ledochowitsch, Daniel Millman, et~al.
\newblock A large-scale standardized physiological survey reveals functional
  organization of the mouse visual cortex.
\newblock \emph{Nature neuroscience}, 23\penalty0 (1):\penalty0 138--151, 2020.

\bibitem[Siegle et~al.(2021)Siegle, Jia, Durand, Gale, Bennett, Graddis,
  Heller, Ramirez, Choi, Luviano, et~al.]{siegle2021survey}
Joshua~H Siegle, Xiaoxuan Jia, S{\'e}verine Durand, Sam Gale, Corbett Bennett,
  Nile Graddis, Greggory Heller, Tamina~K Ramirez, Hannah Choi, Jennifer~A
  Luviano, et~al.
\newblock Survey of spiking in the mouse visual system reveals functional
  hierarchy.
\newblock \emph{Nature}, 592\penalty0 (7852):\penalty0 86--92, 2021.

\bibitem[Stringer et~al.(2019)Stringer, Pachitariu, Steinmetz, Reddy,
  Carandini, and Harris]{stringer2019spontaneous}
Carsen Stringer, Marius Pachitariu, Nicholas Steinmetz, Charu~Bai Reddy, Matteo
  Carandini, and Kenneth~D Harris.
\newblock Spontaneous behaviors drive multidimensional, brainwide activity.
\newblock \emph{Science}, 364\penalty0 (6437):\penalty0 eaav7893, 2019.

\bibitem[Schneidman et~al.(2003)Schneidman, Bialek, and
  Berry]{schneidman2003synergy}
Elad Schneidman, William Bialek, and Michael~J Berry.
\newblock Synergy, redundancy, and independence in population codes.
\newblock \emph{Journal of Neuroscience}, 23\penalty0 (37):\penalty0
  11539--11553, 2003.

\bibitem[Gat and Tishby(1998)]{gat1998synergy}
Itay Gat and Naftali Tishby.
\newblock Synergy and redundancy among brain cells of behaving monkeys.
\newblock \emph{Advances in neural information processing systems}, 11, 1998.

\bibitem[Pica et~al.(2017)Pica, Piasini, Safaai, Runyan, Harvey, Diamond,
  Kayser, Fellin, and Panzeri]{pica2017quantifying}
Giuseppe Pica, Eugenio Piasini, Houman Safaai, Caroline Runyan, Christopher
  Harvey, Mathew Diamond, Christoph Kayser, Tommaso Fellin, and Stefano
  Panzeri.
\newblock Quantifying how much sensory information in a neural code is relevant
  for behavior.
\newblock \emph{Advances in Neural Information Processing Systems}, 30, 2017.

\bibitem[Pica et~al.(2019)Pica, Soltanipour, and Panzeri]{pica2019using}
Giuseppe Pica, Mohammadreza Soltanipour, and Stefano Panzeri.
\newblock Using intersection information to map stimulus information transfer
  within neural networks.
\newblock \emph{BioSystems}, 185:\penalty0 104028, 2019.

\bibitem[B{\'\i}m et~al.(2019)B{\'\i}m, De~Feo, Chicharro, Bieler,
  Hanganu-Opatz, Brovelli, and Panzeri]{bim2019non}
Jan B{\'\i}m, Vito De~Feo, Daniel Chicharro, Malte Bieler, Ileana~L
  Hanganu-Opatz, Andrea Brovelli, and Stefano Panzeri.
\newblock A non-negative measure of feature-related information transfer
  between neural signals.
\newblock \emph{BioRxiv}, page 758128, 2019.

\bibitem[Scagliarini et~al.(2020)Scagliarini, Faes, Marinazzo, Stramaglia, and
  Mantegna]{scagliarini2020synergistic}
Tomas Scagliarini, Luca Faes, Daniele Marinazzo, Sebastiano Stramaglia, and
  Rosario~N Mantegna.
\newblock Synergistic information transfer in the global system of financial
  markets.
\newblock \emph{Entropy}, 22\penalty0 (9):\penalty0 1000, 2020.

\bibitem[Wollstadt et~al.(2021)Wollstadt, Schmitt, and
  Wibral]{wollstadt2021rigorous}
Patricia Wollstadt, Sebastian Schmitt, and Michael Wibral.
\newblock A rigorous information-theoretic definition of redundancy and
  relevancy in feature selection based on (partial) information decomposition.
\newblock \emph{arXiv preprint arXiv:2105.04187}, 2021.

\bibitem[Dutta et~al.(2020)Dutta, Venkatesh, Mardziel, Datta, and
  Grover]{dutta2020information}
Sanghamitra Dutta, Praveen Venkatesh, Piotr Mardziel, Anupam Datta, and Pulkit
  Grover.
\newblock An information-theoretic quantification of discrimination with exempt
  features.
\newblock In \emph{Proceedings of the AAAI Conference on Artificial
  Intelligence}, volume~34, pages 3825--3833, 2020.

\bibitem[Venkatesh and Schamberg(2022)]{venkatesh2022partial}
Praveen Venkatesh and Gabriel Schamberg.
\newblock Partial information decomposition via deficiency for multivariate
  gaussians.
\newblock In \emph{2022 IEEE International Symposium on Information Theory
  (ISIT)}, pages 2892--2897. IEEE, 2022.

\bibitem[Barrett(2015)]{barrett2015exploration}
Adam~B Barrett.
\newblock Exploration of synergistic and redundant information sharing in
  static and dynamical {G}aussian systems.
\newblock \emph{Physical Review E}, 91\penalty0 (5):\penalty0 052802, 2015.

\bibitem[Colenbier et~al.(2020)Colenbier, Van~de Steen, Uddin, Poldrack,
  Calhoun, and Marinazzo]{colenbier2020disambiguating}
Nigel Colenbier, Frederik Van~de Steen, Lucina~Q Uddin, Russell~A Poldrack,
  Vince~D Calhoun, and Daniele Marinazzo.
\newblock Disambiguating the role of blood flow and global signal with partial
  information decomposition.
\newblock \emph{Neuroimage}, 213:\penalty0 116699, 2020.

\bibitem[Boonstra et~al.(2019)Boonstra, Faes, Kerkman, and
  Marinazzo]{boonstra2019information}
Tjeerd~W Boonstra, Luca Faes, Jennifer~N Kerkman, and Daniele Marinazzo.
\newblock Information decomposition of multichannel emg to map functional
  interactions in the distributed motor system.
\newblock \emph{NeuroImage}, 202:\penalty0 116093, 2019.

\bibitem[Krohova et~al.(2019)Krohova, Faes, Czippelova, Turianikova, Mazgutova,
  Pernice, Busacca, Marinazzo, Stramaglia, and Javorka]{krohova2019multiscale}
Jana Krohova, Luca Faes, Barbora Czippelova, Zuzana Turianikova, Nikoleta
  Mazgutova, Riccardo Pernice, Alessandro Busacca, Daniele Marinazzo,
  Sebastiano Stramaglia, and Michal Javorka.
\newblock Multiscale information decomposition dissects control mechanisms of
  heart rate variability at rest and during physiological stress.
\newblock \emph{Entropy}, 21\penalty0 (5):\penalty0 526, 2019.

\bibitem[Pakman et~al.(2021)Pakman, Nejatbakhsh, Gilboa, Makkeh, Mazzucato,
  Wibral, and Schneidman]{pakman2021estimating}
Ari Pakman, Amin Nejatbakhsh, Dar Gilboa, Abdullah Makkeh, Luca Mazzucato,
  Michael Wibral, and Elad Schneidman.
\newblock Estimating the unique information of continuous variables.
\newblock \emph{Advances in neural information processing systems},
  34:\penalty0 20295--20307, 2021.

\bibitem[Bertschinger et~al.(2014)Bertschinger, Rauh, Olbrich, Jost, and
  Ay]{bertschinger2014quantifying}
Nils Bertschinger, Johannes Rauh, Eckehard Olbrich, J{\"u}rgen Jost, and Nihat
  Ay.
\newblock Quantifying unique information.
\newblock \emph{Entropy}, 16\penalty0 (4):\penalty0 2161--2183, 2014.

\bibitem[Rauh et~al.(2022)Rauh, Banerjee, Olbrich, Mont{\'u}far, and
  Jost]{rauh2022continuity}
Johannes Rauh, Pradeep~Kr Banerjee, Eckehard Olbrich, Guido Mont{\'u}far, and
  J{\"u}rgen Jost.
\newblock Continuity and additivity properties of information decompositions.
\newblock \emph{arXiv preprint arXiv:2204.10982}, 2022.

\bibitem[Banerjee et~al.(2018{\natexlab{a}})Banerjee, Rauh, and
  Mont{\'u}far]{banerjee2018computing}
Pradeep~Kr Banerjee, Johannes Rauh, and Guido Mont{\'u}far.
\newblock Computing the unique information.
\newblock In \emph{2018 IEEE International Symposium on Information Theory
  (ISIT)}, pages 141--145. IEEE, 2018{\natexlab{a}}.

\bibitem[Liang et~al.(2023)Liang, Cheng, Fan, Ling, Nie, Chen, Deng, Mahmood,
  Salakhutdinov, and Morency]{liang2023quantifying}
Paul~Pu Liang, Yun Cheng, Xiang Fan, Chun~Kai Ling, Suzanne Nie, Richard Chen,
  Zihao Deng, Faisal Mahmood, Ruslan Salakhutdinov, and Louis-Philippe Morency.
\newblock Quantifying \& modeling feature interactions: An information
  decomposition framework.
\newblock \emph{arXiv preprint arXiv:2302.12247}, 2023.

\bibitem[Williams and Beer(2010)]{williams2010nonnegative}
Paul~L Williams and Randall~D Beer.
\newblock Nonnegative decomposition of multivariate information.
\newblock \emph{arXiv preprint arXiv:1004.2515}, 2010.

\bibitem[Cover and Thomas(2012)]{cover2012elements}
Thomas~M Cover and Joy~A Thomas.
\newblock \emph{Elements of Information Theory}.
\newblock John Wiley \& Sons, 2012.

\bibitem[Griffith and Koch(2014)]{griffith2014quantifying}
Virgil Griffith and Christof Koch.
\newblock Quantifying synergistic mutual information.
\newblock In \emph{Guided self-organization: inception}, pages 159--190.
  Springer, 2014.

\bibitem[Harder et~al.(2013)Harder, Salge, and Polani]{harder2013bivariate}
Malte Harder, Christoph Salge, and Daniel Polani.
\newblock Bivariate measure of redundant information.
\newblock \emph{Physical Review E}, 87\penalty0 (1):\penalty0 012130, 2013.

\bibitem[Kolchinsky(2019)]{kolchinsky2019novel}
Artemy Kolchinsky.
\newblock A novel approach to multivariate redundancy and synergy.
\newblock \emph{arXiv preprint arXiv:1908.08642}, 2019.

\bibitem[Lizier et~al.(2018)Lizier, Bertschinger, Jost, and
  Wibral]{lizier2018information}
Joseph~T Lizier, Nils Bertschinger, J{\"u}rgen Jost, and Michael Wibral.
\newblock Information decomposition of target effects from multi-source
  interactions: perspectives on previous, current and future work.
\newblock \emph{Entropy}, 20\penalty0 (4):\penalty0 307, 2018.

\bibitem[Venkatesh et~al.(2023)Venkatesh, Gurushankar, and
  Schamberg]{venkatesh2023capturing}
Praveen Venkatesh, Keerthana Gurushankar, and Gabriel Schamberg.
\newblock Capturing and interpreting unique information.
\newblock \emph{arXiv preprint arXiv:2302.11873}, 2023.

\bibitem[Banerjee et~al.(2018{\natexlab{b}})Banerjee, Olbrich, Jost, and
  Rauh]{banerjee2018unique}
Pradeep~Kr Banerjee, Eckehard Olbrich, J{\"u}rgen Jost, and Johannes Rauh.
\newblock Unique informations and deficiencies.
\newblock In \emph{2018 56th Annual Allerton Conference on Communication,
  Control, and Computing (Allerton)}, pages 32--38. IEEE, 2018{\natexlab{b}}.

\bibitem[Tishby et~al.(2000)Tishby, Pereira, and Bialek]{tishby2000information}
Naftali Tishby, Fernando~C Pereira, and William Bialek.
\newblock The information bottleneck method.
\newblock \emph{arXiv preprint physics/0004057}, 2000.

\bibitem[Chechik et~al.(2003)Chechik, Globerson, Tishby, and
  Weiss]{chechik2003information}
Gal Chechik, Amir Globerson, Naftali Tishby, and Yair Weiss.
\newblock Information bottleneck for gaussian variables.
\newblock \emph{Advances in Neural Information Processing Systems}, 16, 2003.

\bibitem[Globerson and Tishby(2004)]{globerson2004optimality}
Amir Globerson and Naftali Tishby.
\newblock On the optimality of the gaussian information bottleneck curve.
\newblock \emph{The Hebrew University of Jerusalem, Tech. Rep}, page~22, 2004.

\bibitem[Hinton(2018)]{hinton2018rprop}
Geoffrey Hinton.
\newblock Coursera neural networks for machine learning, lecture 6, 2018.
\newblock URL
  \url{https://www.cs.toronto.edu/~tijmen/csc321/slides/lecture_slides_lec6.pdf}.
\newblock Also see
  \url{https://optimization.cbe.cornell.edu/index.php?title=RMSProp}.

\bibitem[Paninski(2003)]{paninski2003estimation}
Liam Paninski.
\newblock Estimation of entropy and mutual information.
\newblock \emph{Neural computation}, 15\penalty0 (6):\penalty0 1191--1253,
  2003.

\bibitem[Cai et~al.(2015)Cai, Liang, and Zhou]{cai2015law}
T~Tony Cai, Tengyuan Liang, and Harrison~H Zhou.
\newblock Law of log determinant of sample covariance matrix and optimal
  estimation of differential entropy for high-dimensional gaussian
  distributions.
\newblock \emph{Journal of Multivariate Analysis}, 137:\penalty0 161--172,
  2015.

\bibitem[Wasserman(2004)]{wasserman2004all}
Larry Wasserman.
\newblock \emph{All of statistics: a concise course in statistical inference},
  volume~26.
\newblock Springer, 2004.

\bibitem[Timme and Lapish(2018)]{timme2018tutorial}
Nicholas~M Timme and Christopher Lapish.
\newblock A tutorial for information theory in neuroscience.
\newblock \emph{eneuro}, 5\penalty0 (3), 2018.

\bibitem[Institute(2022)]{allen2022vbn}
Allen Institute.
\newblock Visual behavior neuropixels dataset overview, 2022.
\newblock URL
  \url{https://portal.brain-map.org/explore/circuits/visual-behavior-neuropixels}.

\end{thebibliography}


\begin{thebibliography}{5}
\providecommand{\natexlab}[1]{#1}
\providecommand{\url}[1]{\texttt{#1}}
\expandafter\ifx\csname urlstyle\endcsname\relax
  \providecommand{\doi}[1]{doi: #1}\else
  \providecommand{\doi}{doi: \begingroup \urlstyle{rm}\Url}\fi

\bibitem[Cover and Thomas(2012)]{cover2012elements_sm}
Thomas~M Cover and Joy~A Thomas.
\newblock \emph{Elements of Information Theory}.
\newblock John Wiley \& Sons, 2012.

\bibitem[Petersen and Pedersen(2012)]{petersen2012matrix}
K.~B. Petersen and M.~S. Pedersen.
\newblock \emph{The {M}atrix {C}ookbook}.
\newblock Technical University of Denmark, 2012.
\newblock URL \url{http://www2.compute.dtu.dk/pubdb/pubs/3274-full.html}.
\newblock Version: November 15, 2012.

\bibitem[Pedregosa et~al.(2011)Pedregosa, Varoquaux, Gramfort, Michel, Thirion,
  Grisel, Blondel, Prettenhofer, Weiss, Dubourg, Vanderplas, Passos,
  Cournapeau, Brucher, Perrot, and Duchesnay]{pedrogosa2011scikitlearn}
F.~Pedregosa, G.~Varoquaux, A.~Gramfort, V.~Michel, B.~Thirion, O.~Grisel,
  M.~Blondel, P.~Prettenhofer, R.~Weiss, V.~Dubourg, J.~Vanderplas, A.~Passos,
  D.~Cournapeau, M.~Brucher, M.~Perrot, and E.~Duchesnay.
\newblock Scikit-learn: Machine learning in {P}ython.
\newblock \emph{Journal of Machine Learning Research}, 12:\penalty0 2825--2830,
  2011.

\bibitem[Siegle et~al.(2021)Siegle, Jia, Durand, Gale, Bennett, Graddis,
  Heller, Ramirez, Choi, Luviano, et~al.]{siegle2021survey_sm}
Joshua~H Siegle, Xiaoxuan Jia, S{\'e}verine Durand, Sam Gale, Corbett Bennett,
  Nile Graddis, Greggory Heller, Tamina~K Ramirez, Hannah Choi, Jennifer~A
  Luviano, et~al.
\newblock Survey of spiking in the mouse visual system reveals functional
  hierarchy.
\newblock \emph{Nature}, 592\penalty0 (7852):\penalty0 86--92, 2021.

\bibitem[cod()]{code}
Code for this paper.
\newblock URL \url{https://github.com/praveenv253/gpid}.

\end{thebibliography}
